\documentclass[
  11pt,
  a4paper,
  DIV=14
]{scrartcl}

\usepackage{microtype}
\usepackage{iftex}
\ifpdftex
    \usepackage[T1]{fontenc}
    \usepackage[utf8]{inputenc}

\else
    \usepackage{fontspec}
\fi
\usepackage[british]{babel}
\usepackage[
    backend=biber,
    sorting=nyt,
    style=numeric,
    doi=true,            
    eprint=false,            
    isbn=true,            
    url=false,            
    maxnames=6,
    dateabbrev=false,
    date=year,
    giveninits=true
]{biblatex}

\renewbibmacro{in:}{}
\DeclareFieldFormat{pages}{#1}
\DeclareFieldFormat[misc]{title}{`#1'\isdot}
\DeclareFieldFormat[online]{title}{`#1'\isdot}
\DeclareFieldFormat{journaltitle}{#1\isdot}
\DeclareFieldFormat[article,periodical]{volume}{\mkbibbold{#1}}
\DeclareFieldFormat{pages}{#1}
\AtEveryBibitem{%
	\clearfield{number}}    \AtEveryBibitem{\clearfield{pagetotal}}    \AtEveryBibitem{\clearfield{volumes}}
\AtEveryBibitem{\clearfield{language}}

    \AtEveryBibitem{%
        \iffieldundef{doi}
        {
            \iffieldundef{issn}
            {}
            {\clearfield{isbn}}
        }
        {
            \clearfield{issn}%
            \clearfield{isbn}
        }%
    }

\usepackage[normalem]{ulem}
\usepackage{soul}

\usepackage{xcolor}
    \definecolor{cblue}{rgb}{0.16, 0.32, 0.75}
    \definecolor{cred}{rgb}{0.7, 0.11, 0.11}
\usepackage{footnote}

\usepackage{
    mathtools,
    amssymb,
    amsfonts,
    cancel
}
    \allowdisplaybreaks

\usepackage{keytheorems}

\usepackage{
    bbold,
    dsfont,
    mathrsfs,
    nicefrac,
    derivative, 
    tensor,
}

\usepackage[
    colorlinks=true,
    breaklinks=false,
    draft=false,
    bookmarksopen=false,
    bookmarksopenlevel=0,
    bookmarksnumbered=true,
    linkcolor=cred,
    citecolor=cblue,
    urlcolor=black
]{hyperref}

\usepackage{etoolbox}

\usepackage{enumitem}
    \newlist{statements}{enumerate}{1}
    \setlist[statements]{label=(\roman*)}
    \newlist{assumptionlist}{enumerate}{1}
    \setlist[assumptionlist]{label=(\roman*)}

\usepackage{authblk}

\usepackage{orcidlink}

\usepackage{ifthen}

\usepackage{zref-clever}
    \newcommand{\cref}[1]{\zcref{#1}}
    \newcommand{\Cref}[1]{\zcref[S]{#1}}
    \zcsetup{
        cap,
        abbrev,
        nameinlink=false 
    }
    \zcRefTypeSetup{statement}{
        Name-sg = Statement,
        name-sg = statement,
        Name-pl = Statements,
        name-pl = statemens,
    }
    \zcsetup{
        countertype = {statementsi = statement, comp=false, abbrev=false},
    }
    \zcRefTypeSetup{assumptionlistitem}{
        Name-sg = Assumption,
        name-sg = Assumption,
        Name-pl = Assumptions,
        name-pl = Assumptions,
    }
    \zcsetup{
        countertype = {assumptionlisti = assumptionlistitem, comp=false, abbrev=false},
    }

\usepackage{csquotes}

\newkeytheorem{theorem}[
	parent=section,
    refname={Theorem,Theorems},
    Refname={Theorem,Theorems},
]

\newkeytheorem{maintheorem}[
    name={Theorem},
    refname={Theorem,Theorems},
    Refname={Theorem,Theorems},
]

\newkeytheorem{lemma}[
    sibling=theorem,
    refname={Lemma,Lemmas},
    Refname={Lemma,Lemmas},
]
\newkeytheorem{corollary}[
    sibling=theorem,
    refname={Corollary,Corollaries},
    Refname={Corollary,Corollaries},
]
\newkeytheorem{proposition}[
    sibling=theorem,
    refname={Proposition,Propositions},
    Refname={Proposition,Propositions},
]
\newkeytheorem{claim}[
    numbered=no,
    sibling=theorem,
    refname={Claim,Claims},
    Refname={Claim,Claims}
]
\newkeytheorem{definition}[
    sibling=theorem,
    style=definition,
    refname={Definition,Definitions},
    Refname={Definition,Definitions},
]
\newkeytheorem{remark}[
    sibling=theorem,
    style=remark,
    refname={Remark,Remarks},
    Refname={Remark,Remarks}
]

\newkeytheorem{assumption}[
    style=definition,
    name={List of properties},
    refname={list of properties,lists of properties},
    Refname={List of properties,Lists of properties},
]

\newkeytheorem{subassumption}[
    name={},
    parent=assumption,
    style=definition,
    refname={Property,Property},
    Refname={Property,Properties}
]

\newcommand{\assumpA}{\cref{assump:A-H-cinf,assump:A-U-cinf,assump:A-H-cont,assump:A-U-diff,assump:A-H-bound,assump:A-U-bound,assump:A-H-period,assump:A-H-symmetric,assump:A-H-conj}}

\newcommand{\assumpB}{\cref{assump:B-V-bound,assump:B-V-sym,assump:B-V-conj,assump:B-V-period,assump:B-V-conj}}

\newcommand{\assumpC}{\cref{assump:C-commutator-bound,assump:C-H-bound,assump:C-H-conj,assump:C-H-diff,assump:C-H-periodic,assump:C-H-self-adjoint}}

\newcommand{\bigo}{\mathcal{O}}

\newcommand\tmax{t_{\mathrm{max}}}
\newcommand{\Tcoeff}[2]{ \tensor{\eof*{#1}}{^{\lower0.25em \hbox{$\scriptstyle[#2]$}}}  }
\newcommand{\Ttrunc}[3]{\eof*{#1}\indices{ ^{\lower0.25em \hbox{$\scriptstyle[#2;#3]$} } } }

\newcommand{\cinfty}{\mathcal C^\infty}
\newcommand{\cinftyofh}{\cinfty(H_0)}

\newcommand{\dom}{\mathcal{D}}

\newcommand{\Kop}[1]{\mathcal{K}_{#1}}
\newcommand{\Kopeff}[2]{\mathcal{K}_{#1}^{(#2)}}

\newcommand{\Sop}[1]{S_{#1}}
\newcommand{\Sopeff}[2]{S_{#1}^{(#2)}}

\DeclarePairedDelimiterX{\avg}[2]{\langle}{\rangle_{#1}}{#2}

\DeclarePairedDelimiterXPP{\osc}[3]{\Delta_{#1}^{(#2)} }{(}{)}{}{#3}

\DeclareMathOperator{\sgn}{sgn}

\newcommand{\Heff}[2]{H_{\mathrm{eff},#1}^{(#2)}}
\newcommand{\Heffext}[2]{\hat{H}_{\mathrm{eff},#1}^{(#2)}}
\newcommand{\Heffl}[1]{ H_{\mathrm{eff}} ^{[#1]} }

\newcommand{\HFM}[2]{
    \ifblank{#1}{H_{\mathrm{FM}}^{(#2)}}{H_{\mathrm{FM},#1}^{(#2)}}
}
\newcommand{\HFMext}[2]{
    \ifblank{#1}{\hat{H}_{\mathrm{FM}}^{(#2)}}{\hat{H}_{\mathrm{FM},#1}^{(#2)}}
}
\newcommand{\HFMl}[2]{H_{\mathrm{FM}}^{\ifblank{#2}{}{(#2)}[#1]}}
\newcommand{\one}{\mathds{1}}

\newcommand{\annotateeqn}[2]{\stackrel{\mathclap{#1}}{#2}}

\newcommand{\textfrac}[2]{ {\textstyle \frac{#1}{#2}}}

\newcommand{\e}{\mathrm{e}} 
\newcommand{\iu}{\mathrm{i}\mkern1mu} 

\newcommand\dd\odif

\newcommand\creation{a^\ast}
\newcommand\ad\creation
\newcommand\adag\creation
\newcommand\annihilation{a}
\newcommand\an\annihilation
\newcommand\photonnumber{N}
\newcommand\nn\photonnumber

\newcommand{\adj}{^\ast}
\newcommand{\inv}{^{-1}}
\newcommand\close\overline
\newcommand\conj\overline
\newcommand{\primed}{^\prime}

\newcommand{\HH}{\mathcal{H}}
\newcommand{\N}{\mathds{N}} 
\newcommand{\Z}{\mathds{Z}} 
\newcommand{\R}{\mathds{R}} 
\newcommand{\C}{\mathds{C}} 

\newcommand\Ltwo[1]{L^2(#1)}

\DeclarePairedDelimiter\of()
\DeclarePairedDelimiter\cof{\{}{\}}
\DeclarePairedDelimiter\eof{[}{]}
\DeclarePairedDelimiter\floorof\lfloor\rfloor

\DeclarePairedDelimiterXPP{\textfloorof}[1]{\begingroup\textstyle}{\lfloor}{\rfloor}{\endgroup}{#1}

\DeclarePairedDelimiterX\comm[2]{[}{]}{
    \ifblank{#1}{\;\cdot\;}{#1}, \ifblank{#2}{\;\cdot\;}{#2}
}
\newcommand\commutator\comm

\DeclarePairedDelimiterX\innerp[2]{\langle}{\rangle}{
    \ifblank{#1}{\;\cdot\;}{#1},\; \ifblank{#2}{\;\cdot\;}{#2}
}

\DeclarePairedDelimiterX\braket[2]{\langle}{\rangle}{
    \ifblank{#1}{\;\cdot\;}{#1}\,\delimsize\vert\, \ifblank{#2}{\;\cdot\;}{#2}
}

\DeclarePairedDelimiterX\ket[1]{\lvert}{\rangle}{#1}

\newcommand{\emptyplaceholder}{\;\cdot\;} 
\DeclarePairedDelimiterX\norm[1]\lVert\rVert{
    \ifblank{#1}{\emptyplaceholder}{#1}
}

\DeclarePairedDelimiterXPP\pnorm[1]{}\lVert\rVert{_{p}}{
    \ifblank{#1}{\emptyplaceholder}{#1}
}

\DeclarePairedDelimiterXPP\twonorm[1]{}\lVert\rVert{_{2}}{
    \ifblank{#1}{\emptyplaceholder}{#1}
}

\DeclarePairedDelimiterXPP\opnorm[1]{}\lVert\rVert{_{\mathrm{op}}}{
    \ifblank{#1}{\emptyplaceholder}{#1}
}

\DeclarePairedDelimiterX\abs[1]\lvert\rvert{#1}
\DeclarePairedDelimiterX\aof[1]\langle\rangle{#1}
\providecommand\given{}
\DeclarePairedDelimiterXPP\Prob[1]{\mathcal{P}}(){}{
   \renewcommand\given{\nonscript\:\delimsize\vert\nonscript\:\mathopen{}}
   #1}

\providecommand\given{}
\newcommand\SetSymbol[1][]{%
   \nonscript\:#1\vert
   \allowbreak
   \nonscript\:
   \mathopen{}}
\DeclarePairedDelimiterX\Set[1]\{\}{%
\renewcommand\given{\SetSymbol[\delimsize]}
#1
}

\newcommand{\evalat}[2]{\left.#1\right\vert_{#2}}

\DeclarePairedDelimiterXPP\domof[1]{\dom}{(}
{)}{}{#1}

\newcommand{\sodv}[2]{\mathrm s\text{-}\frac{\mathrm d #1}{\mathrm d #2}}

\newcommand\blfootnote[1]{%
  \begingroup
  \renewcommand\thefootnote{}\footnote{#1}%
  \addtocounter{footnote}{-1}%
  \endgroup
}

\hypersetup{
  pdftitle={The Floquet--Magnus expansion of unbounded operators},
}

\date{30 October 2026}
\title{The Floquet--Magnus expansion of unbounded operators}

\author[1]{Daniel Burgarth\,\orcidlink{0000-0003-4063-1264}}
\author[2]{Robin Hillier\,\orcidlink{0000-0001-7015-6201}}
\author[1]{Davide Lonigro\,\orcidlink{0000-0002-0792-8122}}
\author[1]{Leonhard Richter\,\orcidlink{0009-0003-0482-6133}\thanks{Corresponding author: \href{mailto:leonhard.richter@fau.de}{leonhard.richter@fau.de}}}

\affil[1]{Institute of Theoretical Physics, Friedrich-Alexander-Universität Erlangen-Nürnberg, Staudtstraße 7/B3, 91058 Erlangen, Germany}
\affil[2]{School of Mathematical Sciences, Lancaster University, Lancaster LA1 4YX, United Kingdom}

\numberwithin{equation}{section}

\begin{document}
\pagestyle{headings}
\maketitle
\vspace{-1cm}
\begin{abstract}
  The Floquet--Magnus expansion is a widely used tool to derive effective descriptions of time-periodic quantum systems by approximating their dynamics with a time-independent Hamiltonian.
  However, its standard formulation is, strictly speaking, restricted to bounded Hamiltonians.
  In this work, we extend its definition and analysis to a broad class of time-periodic unbounded Hamiltonians.
  Our approach is based on an a priori distinct nonperturbative framework for the construction of effective Hamiltonians, which we show to reproduce the Floquet--Magnus expansion.
  A particular strength of our framework is that it allows us to prove that the resulting effective dynamics approximates the original time evolution propagators to arbitrary order in the high-frequency limit without requiring convergence of the Floquet--Magnus expansion, a condition that is
  already highly restrictive even in the bounded setting.
  We illustrate the scope of the method on representative models: the quantum Rabi Hamiltonian in the interaction picture,
  and the periodically driven quantum harmonic oscillator.
  \par\vskip\baselineskip\noindent
\end{abstract}


\BeforeTOCHead[toc]{{\pdfbookmark[1]{\contentsname}{toc}}}
\tableofcontents
\section{Introduction}
Quantum systems governed by time-dependent Hamiltonians arise naturally across a wide range of physical settings. Time dependence may originate from externally applied driving fields, as commonly encountered in quantum control and quantum technologies, from time-dependent changes of frame introduced to simplify the analysis of the dynamics, or from a combination of both, see e.g.~\cite{cohenReminiscenceClassicalChaos2023,abaninEffectiveHamiltoniansPrethermalization2017}.
While such descriptions are often physically indispensable, they significantly increase the mathematical complexity of the problem: even for comparatively simple models, solving the corresponding time-dependent Schrödinger equation exactly is typically out of reach, making systematic approximation schemes indispensable.

A widely adopted strategy is to approximate the true time-dependent dynamics by an effective autonomous evolution generated by a suitably constructed time-independent Hamiltonian. Beyond their practical usefulness, such effective descriptions play a central conceptual role, as they allow one to isolate the relevant dynamical features on appropriate time scales. However, turning this idea into a controlled and quantitatively reliable approximation scheme is highly nontrivial, especially in infinite-dimensional settings relevant to realistic quantum systems.

Among the available approaches, the Floquet--Magnus (FM) expansion, see \cite{Blanes_Casas_Oteo_Ros_2009} for a pedagogical overview, occupies a distinguished position. For periodically driven systems, Floquet theory ensures that the overall propagator can be factorised into a periodic propagator capturing the microscopic time-scale and an autonomous time evolution describing the long-term stroboscopic dynamics. The generator of the latter is commonly referred to as an effective Hamiltonian for the system. The FM expansion provides a formal series representation of this effective generator in powers of the driving period, or equivalently, in the inverse driving frequency~\cite{Blanes_Casas_Oteo_Ros_2009,bukovUniversalHighfrequencyBehavior2015}.
In the high-frequency regime, it is widely expected—and extensively exploited in applications—that truncations of this series yield accurate approximations of the true dynamics, thereby forming the basis of effective descriptions in quantum control and quantum optics.
As a consequence, the FM expansion has become a central analytical tool in areas ranging from quantum control to the engineering of effective Hamiltonians in many-body systems, where it is routinely used to design and interpret experimentally relevant regimes.

Despite its widespread use, the mathematical status of the FM expansion remains delicate.
In its standard formulation, the analysis is typically carried out in operator norm, which restricts its applicability to bounded Hamiltonians~\cite{Blanes_Casas_Oteo_Ros_2009,fangTimeDependentHamiltonianSimulation2025a}.
This limitation is severe from a physical perspective, as most Hamiltonians of interest—particularly in quantum optics and many-body physics—are unbounded,
so that the standard formulation of the FM expansion does not rigorously apply in precisely the regimes where it is most widely used.
An immediate issue in the case of unbounded operators is that the FM expansion is expressed in terms of nested commutators whose domain properties are far from trivial.
Furthermore, even when each term in the expansion is well-defined, establishing convergence of the series is a highly nontrivial task.

Starting with an influential work by \citeauthor{hochbruckMagnusIntegratorsTimeDependent2003}, there have been several results on numerical integration methods motivated by finite truncations of the Magnus series, that are in particular applicable for the numerical integration of unbounded time-dependent quantum systems~\cite{gonzalezSecondorderMagnustypeIntegrator2006,hochbruckMagnusIntegratorsTimeDependent2003}.
Similar results are known for so-called commutator-free quasi-Magnus exponential integrators, which are based on mere linear combinations instead of commutators of the time-dependent Hamiltonian and therefore circumvent the domain issues mentioned above~\cite{blanesHighorderCommutatorfreeQuasiMagnus2017,blanesConvergenceAnalysisHighorder2018,casaresQuantumSimulationTimedependent2024}.
All of these results show convergence of the integration method as the time step of the discretization vanishes without assuming the convergence of the Magnus series.

When the time-dependent Hamiltonian is periodic, properties of the Floquet Hamiltonian without series expansion have been studied in the unbounded setting. For example, in some cases sufficient conditions for the Floquet-Hamiltonian having pure point spectrum are known~\cite{wangPurePointSpectrum2008,grebertKAMQuantumHarmonic2011} and the spectral properties of the Floquet Hamiltonian dictate the long-time behaviour of the system~\cite{howlandStationaryScatteringTheory1974a}.

However, even in the bounded setting, standard approaches do not provide high-frequency convergence for finite truncations of the Magnus series applied to periodic systems.
Moreover, available error estimates are typically derived from bounds on the remainder of the exponential series and consequently exhibit an unfavourable growth in time, obscuring the effective time scales on which truncated expansions remain accurate~\cite{apelSharperMagnusExpansion2025,fangHighorderMagnusExpansion2026}.
This behaviour limits their usefulness in regimes of physical interest.

A first step towards overcoming these limitations was taken in \cite{deyErrorBoundsFloquetMagnus2025}, where two of the present authors and collaborators developed a nonperturbative framework for constructing effective Hamiltonians of arbitrary order for periodically driven systems in the bounded case, building on an integration-by-parts technique first developed in \cite{burgarth2022one}.
This approach recovers the FM expansion, while at the same time providing error bounds that do not rely on the convergence of the Magnus series, exhibit a substantially improved scaling in time, and yield errors that vanish in the high-frequency limit.

The purpose of the present work is to extend this program to the mathematically much more demanding—and physically most relevant—setting of unbounded Hamiltonians.
We develop a general construction of effective Hamiltonians for periodically driven quantum systems generated by possibly unbounded operators.
Our approach is formulated in a Hilbert space $\mathcal{H}$ endowed with a reference self-adjoint operator $H_0$ that sets the natural energy scale of the problem.
We adopt a strong-topology framework on the space $\mathcal{C}^\infty(H_0)$ of smooth vectors associated with $H_0$.
Under suitable regularity assumptions relating the time-dependent Hamiltonian to $H_0$, we first derive a general representation of the difference between the propagators generated by two time-dependent Hamiltonians in terms of higher-order relative actions (\cref{th-S:iterated-ip}).
Building on this identity, we develop an explicit recursive procedure (\cref{th-S:U-diff-periods}) to construct effective time-independent Hamiltonians  $\Heffl{L}{T}$ whose associated dynamics approximate the true evolution up to errors of order $T^{L+2}$  at stroboscopic times \(t \in \Z T\) and up to order \(T\) at arbitrary times \(t \in \R\), where $T$ denotes the driving period and \(L\) is the number of iterations in the recursion, with explicit control at each order.

Instead of generalizing the standard FM approach for finite-dimensional systems to an unbounded setting, we deliberately follow the alternative approach developed in Ref.~\cite{deyErrorBoundsFloquetMagnus2025} and described above.
The reason for this is twofold.
The first reason is that generalizing the standard approach \cite{Blanes_Casas_Oteo_Ros_2009} would involve studying the existence of the derivative of the exponential map and its inverse on an infinite-dimensional Lie algebra.
There has been a lot of progress in the field of infinite-dimensional Lie algebras in recent decades, see the compendium \cite{gloecknerInfiniteDimensionalLieGroups2026}, but unbounded generators bring further challenges.
Our approach bypasses this by directly proving error estimates for each truncation of the FM expansion without having to consider the full series and without using perturbation theory.
The second reason is that, similarly to Ref.~\cite{deyErrorBoundsFloquetMagnus2025}, the error estimates provided by our approach are stronger because they automatically converge in the high-frequency limit, no matter whether or not the overall series converges.

We then show that our construction of effective Hamiltonians reproduces the FM expansion order by order (\cref{th-S:FM-comparison}). In this sense, our results provide a natural and robust extension of the FM expansion to the unbounded setting.
To our knowledge, this constitutes the first systematic and general framework for constructing and justifying FM-type expansions for quantum systems governed by unbounded Hamiltonians.

Our proof requires additional technical assumptions; for example the existence of a suitable conjugation structure for the driving Hamiltonian.
While our assumptions are not expected to be fully necessary, some form of restriction is unavoidable in the unbounded setting, as the FM expansion cannot be expected to be valid in full generality.
The following example illustrates this fact.
Consider the Hamiltonian
\begin{equation}
    H(t)
    = \sin\of*{\omega t} p^2 + \cos\of*{\omega t} q
    ,
\end{equation}
on \(\HH = \Ltwo{[0,\infty)}\), where \(q\) is the usual position operator, \(-p^2\) is the Laplacian with Dirichlet boundary conditions at the origin, and \(\omega >0\).
It is easy to see that the first order in the FM expansion vanishes.
Therefore, the leading nontrivial order is the second term in the FM expansion, whose formal expression is given by \(-\frac{\iu}{2\omega} \comm{p^2}{q} = -\frac{1}{\omega} p\). This operator admits no self-adjoint extensions on the half line \cite{Burgarth_Facchi_Fraas_Hillier_2021}.
This means that, in this case, the leading order of the FM expansion does not generate unitary dynamics, rendering the corresponding dynamics ill-defined.
Our assumptions are designed to exclude this type of pathology.

Our framework includes, under some suitable technical requirements, two cases of particular interest.
First, we can deal with cases where the time dependence of the Hamiltonian arises from the interaction picture of a time-independent Hamiltonian of the form $ H_0 + V$ (\cref{th:summary-interaction-pic}).
Second, our framework covers cases where the time-dependent Hamiltonian $H(t)$ is a self-adjoint operator with time-independent domain (\cref{th:summary-hyperbolic-setting}), leaning on Kato's theory of evolution equations in the hyperbolic setting \cite{kato1953integration,kato1970linear}.
Additionally, we illustrate the scope of our method on physically relevant examples, including the quantum Rabi model.
In particular, we recover the rotating-wave approximation (RWA) of the quantum Rabi model as the leading-order effective Hamiltonian, and systematically construct higher-order corrections such as the Bloch--Siegert Hamiltonian.

We conclude by noting that the integration-by-parts framework developed in this work is constructive in nature, as it yields explicit error bounds at each order.
While such bounds are necessarily intricate and typically not optimal in full generality, the method can be specialised to concrete models or classes of systems, leading to significantly sharper and more tractable estimates than currently available in the literature.
This perspective will be pursued in forthcoming work, where we apply the present framework to physically relevant models of light--matter interaction.

The paper is organised as follows.
In \cref{sec:assumptions-and-main-results}, we introduce the notation, state our assumptions, and present the main results.
In \cref{sec:iterated-integration-by-parts}, we develop the main structural method, which is then applied in \cref{sec:Heff} to define the effective Hamiltonians and establish the main results.
In \cref{sec:interaction-picture} and \cref{sec:hyperbolic-setting}, we discuss the special cases of time-independent Hamiltonians in the interaction picture, and time-dependent Hamiltonians with constant self-adjointness domain, respectively.
In \cref{sec:examples} we illustrate our results on representative examples.
\Cref{sec:assumptions-and-main-results,sec:examples} are self-contained, meaning that readers mostly interested in the main results and examples may skip the more technical \cref{sec:iterated-integration-by-parts,sec:Heff,sec:interaction-picture,sec:hyperbolic-setting}.

\section{Main results}\label{sec:assumptions-and-main-results}

Throughout the paper, we consider families of strongly continuous unitary operators generated
by time-dependent Hamiltonians with period $T$.
The parameter $T$ plays a central role in our analysis, as it governs the accuracy of the approximations arising from the Floquet--Magnus expansion.

For convenience, we adopt the following convention. Given a $1$-periodic Hamiltonian $(H(t))_{t \in \R}$ which we assume to generate a unitary propagator $(U(t,s))_{t,s \in \R}$, we consider the rescaled $T$-periodic Hamiltonians
\begin{equation}\label{eq:rescaling}
H^{(T)}(t) \coloneqq H(t/T), \qquad T > 0,
\end{equation}
which we again assume to generate corresponding propagators $(U^{(T)}(t,s))_{t,s \in \R}$.
This representation will be used throughout the paper.
Unless otherwise specified, we fix the initial time to be \(0\) and describe the evolution with respect to this reference time; accordingly, we write \(U(t) \coloneqq  U(t,0)\) and \(U^{(T)}(t) \coloneqq  U^{(T)}(t,0)\).

We begin in \cref{sec:assumption-a} by formulating a minimal set of properties under which our results hold.
These properties provide a convenient framework for the desired results, although some of them may be nontrivial to verify in concrete situations.
We then turn in \cref{sec:assumption-b,sec:assumption-c} to two more explicit sets of conditions, tailored to physically relevant settings and formulated in a way that is more readily verifiable in applications.

\subsection{General setting}
\label{sec:assumption-a}
In the general setting considered in this work, we will assume that \((H(t))_{t\in\R}\) and \((U(t))_{t\in\R}\) are some families of operators satisfying certain properties which will be specified below and invoked as needed in the subsequent theorems.
In the following and throughout the paper, we denote by \(\N_0\) the non-negative integers \(\Set{0,\,1,\,2,\,\dots}\), and with \(\N\) the strictly positive integers \(\Set{1,\,2,\,\dots}\).
\begin{assumption}\label{assump:A}
Let $\HH$ be a separable Hilbert space, $(U^{(T)}(t))_{t\in\R}$ a family of unitary operators on $\HH$ for every $T>0$, with $U^{(T)}(0)=\one$, $(H(t))_{t\in\R}$ a family of (possibly unbounded but not necessarily self-adjoint) operators on $\HH$, and $H^{(T)}(t):=H(t/T)$, for all $t\in\R$ and $T>0$. Furthermore, let $H_0$ be a (possibly unbounded) self-adjoint operator on $\HH$. Hereafter we adopt the standard notation
\begin{equation}
  \cinftyofh = \bigcap_{m=1}^\infty \domof{H_0^m}.
\end{equation}
\begin{subassumption}\label{assump:A-H-cinf}
    For all $t\in\R$, $\cinftyofh\subseteq \domof{H(t)}$ and $H(t) \cinftyofh \subseteq \cinftyofh$.
\end{subassumption}
\begin{subassumption}\label{assump:A-U-cinf}
    For all $t\in\R$ and \(T >0\), $U^{(T)}(t) \cinftyofh \subseteq \cinftyofh$.
\end{subassumption}
\begin{subassumption}\label{assump:A-H-cont}
    For all $m\in\N_0$ and $\psi\in\cinftyofh$, the map $t\in\R\mapsto H_0^m H(t)\psi$ is continuous.
\end{subassumption}
\begin{subassumption}\label{assump:A-U-diff}
    For all $m\in\N_0$, \(T>0\), and $\psi\in\cinftyofh$, the map $t\in\R\mapsto H_0^m U^{(T)}(t)\psi$ is differentiable and
    \begin{equation}
        \odv{}{t} H_0^m U^{(T)}(t)\psi=-\iu H_0^m H^{(T)}(t)U^{(T)}(t)\psi.
    \end{equation}
\end{subassumption}
\begin{subassumption}\label{assump:A-H-bound}
    For every $m\in\N_0$ and every compact interval $I$, there exist $k_m \in \N_0$, $a_m,\, b_m\geq 0$ such that
    \begin{equation}\label{eq:assump:H-bound}
       \norm{H_0^m H(t) \psi}
       \leq a_m \norm{H_0^{m+k_m}\psi} + b_m \norm{\psi}
    \end{equation}
    for all $\psi\in\cinftyofh$ and $t\in I$.
\end{subassumption}

\begin{subassumption}
\label{assump:A-U-bound}
    For every $m\in \N$ and compact interval $I$, there is \(k'_m \in \N_0\) and for any \(T > 0\)  there exist $\,a'^{(T)}_m,\, b'^{(T)}_m\geq 0$ such that
    \begin{equation}\label{eq:assump:A-U-bound}
        \norm{H_0^m U^{(T)}(t) \psi }
        \leq a'^{(T)}_m \norm{H_0^{m+k'_m}\psi} + b'^{(T)}_m \norm{\psi},
    \end{equation}
    for all $\psi\in\cinftyofh$ and $t\in I$, and
    \begin{equation}\label{eq:assump:A-U-bound-asymptotics}
        a'^{(T)}_m = \bigo(1)
        ,\quad
        b'^{(T)}_m = \bigo(1)
    \end{equation}
    as \(T \to 0\) (cf.~\cref{def:big-O}).
\end{subassumption}
\begin{subassumption}\label{assump:A-H-period}
    $t\mapsto H(t)$ is 1-periodic on $\cinftyofh$, i.e., $H(t+1)\psi= H(t)\psi$ for all $t\in\R$ and $\psi\in\cinftyofh$.
\end{subassumption}
\begin{subassumption}\label{assump:A-H-symmetric}
    For every \(t\in \R\), \(H(t)\) is symmetric on \(\cinftyofh\).
\end{subassumption}
\begin{subassumption}\label{assump:A-H-conj}
    There exists a conjugation $J$ on $\HH$ such that, for all $t\in\R$, $JH(t) \subseteq H(-t)J$ and $J \cinftyofh \subseteq \cinftyofh$.
\end{subassumption}
\end{assumption}
The guiding idea behind these properties is the following.
While the time-dependent Schrödinger equation is generally difficult to analyse directly, we assume the existence of a well-behaved reference Hamiltonian $H_0$ that sets the natural energy scale of the system.
The time-dependent Hamiltonian $H(t)$ is then controlled relative to this reference operator, in the sense that its size and regularity can be estimated with respect to $H_0$ (see \cref{assump:A-H-cinf,assump:A-H-cont,assump:A-H-bound}).
This is reminiscent of perturbative frameworks, although our approach is not perturbative in nature.

Furthermore, since the structure of the propagator $U^{(T)}(t)$ is likewise intricate in the time-dependent setting, we will impose additional conditions that control its regularity and growth relative to the energy scale defined by $H_0$, cf.\ \cref{assump:A-U-cinf,assump:A-U-diff,assump:A-U-bound}.
It is important to remark that the relation between the unitary family \(U^{(T)}(t)\) and the Hamiltonian \(H^{(T)}(t)\) is imposed only by \cref{assump:A-U-diff}.
The existence of such a propagator \(U^{(T)}(t)\), even for \(m=0\), is not clear just from our assumptions on \(H^{(T)}(t)\), as deciding the solvability of the Schrödinger equation for time-dependent Hamiltonians poses a non-trivial question.
Here, we assume the existence of both \(U^{(T)}(t)\) and \(H^{(T)}(t)\) while in \cref{sec:assumption-b,sec:assumption-c} the existence of \(U^{(T)}(t)\) will be guaranteed by sufficient assumptions on \(H^{(T)}(t)\).

\Cref{assump:A-H-period} encodes the time periodicity of the system (which, as discussed above, can be rescaled to any given period).
The symmetry of $H(t)$ on $\cinftyofh$ as stated in \cref{assump:A-H-symmetric} is a common physical assumption when $H(t)$ represents the Hamiltonian of a system.
Finally, \cref{assump:A-H-conj} can be interpreted as compatibility with time-reversal symmetry via a conjugation, in analogy with structures appearing in the Charge-Parity-Time (CPT) theorem from particle physics. As we will see, this implies that our construction yields effective Hamiltonians having self-adjoint extensions at every order.

A few topological remarks are due.
We stress that \cref{assump:A-H-cont} does not impose operator-norm continuity of the operators \(H_0^m H(t)\) but continuity with respect to \(t\) in the strong operator topology.
Instead of stating our assumptions in the Hilbert space topology and demanding regularity with additional application of \(H_0^m\) one could equip \(\cinftyofh\) explicitly with the Fréchet topology induced by semi-norms \(p_m(\psi) = \norm{H_0^m\psi} + \norm{\psi}\), making it a Fréchet space.
Then, \cref{assump:A-H-bound} for example, could be interpreted as strong continuity of the operators \(H^{(T)}(t)\) in the Fréchet topology and \cref{assump:A-U-diff} could be interpreted as the Schrödinger equation not only in the Hilbert space topology, but also in the Fréchet topology.
Lastly, we want to remark that the choice for \(\cinftyofh\) as dense subspace is out of convenience.
One could also work in \(\domof{H_0^M}\) for sufficiently large \(M\) instead, assuming regularity only for \(m \leq M\), and this way obtain similar conclusions under slightly weaker assumptions.

Our first result concerns the difference between the unitary dynamics generated by two given time-dependent Hamiltonians.
More precisely, for any prescribed order $L$, this difference can be expressed in terms of a family of higher-order relative actions.
This extends to the unbounded setting the results established in \cite{deyErrorBoundsFloquetMagnus2025}, which themselves build on a higher-order generalization of an integration-by-parts technique introduced in \cite{burgarth2022one} for bounded operators (see also \cite{burgarth2024taming,richter2026quantifying,HahnTrotter24,YiFasterDD26}).

\begin{maintheorem}[see \cref{thm:iterated-integration-by-parts}]\label{th-S:iterated-ip}
    Let $H_0$ be a self-adjoint operator on $\HH$; let \((U_1(t))_{t\in \R}\) be a family of unitary operators with corresponding family \((H_1(t))_{t\in \R}\) such that \cref{assump:A-H-cinf,assump:A-H-cont,assump:A-H-bound} hold for $T=1$, and $t\mapsto U_1(t)^*$ is strongly differentiable on $\cinftyofh$ with
    \begin{equation}
        \frac{\dd}{\dd t}U_1(t)^*\psi = \iu U_1(t)^{*} H_1(t) \psi, \quad \forall\psi\in\cinftyofh,
    \end{equation}
    and $U_1(0)=\one$; furthermore, let \((U_2(t))_{t\in \R}\) be a family of unitary operators with corresponding family \((H_2(t))_{t\in \R}\) such that \cref{assump:A-H-cinf,assump:A-H-cont,assump:A-U-diff,assump:A-H-bound,assump:A-U-cinf} hold for $T=1$ and $U_2(0)=\one$.
    Then, for every \(\psi \in \cinftyofh\), \(L \in \N_{0}\), T>0, and \(t\in\R\), we have
    \begin{equation}
        \label{eq-S:iterated-integration-by-parts}
        \of[\big]{U_1(t) - U_2(t)}\psi
        =
        \begin{multlined}[t]
            \sum_{j=1}^{L+1} (-\iu)^j \Sop{j}(t) U_2(t) \psi
            \\
            \!- \iu \int_0^t U_1(t) U_1(s)\adj \of*{
                \sum_{j=0}^L (-\iu)^j \avg*{T}{\Kop{}\of*{\Sop{j}}}
                + (-\iu)^{L+1} \Kop{s}\of*{\Sop{L+1}(s)}
            } U_2(s) \psi \,\dd s
            ,
        \end{multlined}
    \end{equation}
    where we use the shorthand notation
    \begin{equation}
        \avg{T}{A}
        = \frac{1}{T} \int_0^T A(\tau) \dd \tau
    \end{equation}
    \begin{equation}\label{eq:Kop-summary}
        \Kop{s}(A)
        = H_1(s) A - A H_2(s),
    \end{equation}
    and the operators $S_j(t)$ are iteratively defined on $\cinftyofh$ as follows:
     $\Sop{0}=\one$, and
     \begin{equation}\label{eq:Sop-summary}
         \Sop{j+1}(t)
         = \int_0^t\Kop{s}(\Sop{j}(s))\dd s - t \avg*{T}{\Kop{\cdot} (\Sop{j})}
     \end{equation}
     for all $j\in\N_0$.
\end{maintheorem}

Given a fixed unitary family $(U^{(T)}(t))_{t\in\R}$ associated with a time-dependent $T$-periodic Hamiltonian $(H^{(T)}(t))_{t\in\R}$, our aim is to exploit \cref{th-S:iterated-ip} to \emph{construct} a time-independent effective Hamiltonian whose generated dynamics approximates $(U^{(T)}(t))_{t\in\R}$ up to a prescribed order $L$.
In addition, we require this approximation to be consistent with the small-period (high-frequency) limit $T \to 0$.
To this end, we consider the ansatz
\begin{equation}\label{eq:Heff-ansatz}
    \Heff{L}{T}
    = \sum_{l=0}^L T^l \Heffl{l},
\end{equation}
where the operators $\Heffl{l}$ are to be determined from the underlying $1$-periodic Hamiltonian $H(t)$ and are independent of the period $T$.

The following theorem provides such a construction and is, in addition, tailored to yield particularly accurate approximations at times that are integer multiples of the period \(T\). It generalises the results of \cite{deyErrorBoundsFloquetMagnus2025}, which are formulated in the bounded setting, to the considerably more delicate unbounded case addressed in this work.

\begin{maintheorem}[{see \cref{def:Heff-definition,thm:Heff-condition,thm:general-error-ip,thm:effective-Hamiltonian-conjugation}}]\label{th-S:U-diff-periods}
    Assume $(U^{(T)}(t))_{t\in\R}$, for every $T>0$, and $(H(t))_{t\in\R}$ satisfy \cref{assump:A-H-cinf,assump:A-U-cinf,assump:A-H-cont,assump:A-U-diff,assump:A-H-bound,assump:A-H-period,assump:A-U-bound}; let $T>0$, and let \(\Heff{L}{T}\) be an operator of the form~\eqref{eq:Heff-ansatz}, with
    \begin{equation}\label{eq:Heff-definition-0-summary}
        \Heffl{0}
		= \avg*{1}{H}
    \end{equation}
    and
    \begin{equation}\label{eq:Heff-definition-l-summary}
        \Heffl{l}
		= - \sum_{j=1}^l (-\iu)^j \Tcoeff{ \avg*{T}{\Kopeff{\cdot,L-j}{T}\of*{\Sopeff{j,L-j}{T} } } }{l}, \quad l\geq 1
        ,
    \end{equation}
    where \(\Tcoeff{A(T)}{k}\) denotes the \(k\)th coefficient of a polynomial \(A(T)\) in \(T\), and the indices \(T\) and \(L\) in \(\Kopeff{t,L}{T}\) and \(\Sopeff{j,L}{T}\) indicate that \cref{eq:Kop-summary,eq:Sop-summary} are to be considered with \(H_1 = \Heff{L}{T}\) and \(H_2(t) = H^{(T)}(t)\), respectively.
    Then the following statements hold:
    \begin{statements}
        \item Each operator $\Heffl{l}$ satisfies \cref{assump:A-H-cinf,assump:A-H-cont,assump:A-H-bound} and, for every \(L \in \N_0\), we have
    \begin{equation}\label{eq:Heff-condition-summary}
	    \sum_{j=0}^{L} (-\iu)^j \avg*{T}{\Kopeff{\cdot,L}{T}\of{\Sopeff{j,L}{T}}} \psi
        = \bigo\of*{T^{L+1}}
        ,\quad \text{as } T \to 0.
	\end{equation}
    Moreover, the coefficients $\Heffl{l}$, for $l=0,\ldots,L$, are uniquely determined by the requirement that an operator $\Heff{L}{T}$ of the form~\eqref{eq:Heff-ansatz} satisfies \cref{eq:Heff-condition-summary}.

    \item Assume further that \(\Heff{L}{T}\) admits at least one self-adjoint extension \(\Heffext{L}{T}\). Then, for every \(L \in \N_0\), \(t \in \R\), and \(\psi \in \cinftyofh\),
\begin{equation}
     \label{eq:general-error-ip-2-summary}
       \of*{\e^{-\iu t \Heffext{L}{T} } - U^{(T)}(t)}\psi
        = \bigo\of*{T}  ;
\end{equation}
moreover, for every \(q\in \Z\) and \(\psi \in \cinftyofh\),
\begin{equation}
    \label{eq:general-error-ip-3-summary}
        \of*{\e^{-\iu qT \Heffext{L}{T} } - U^{(T)}(qT)}\psi
        = \bigo\of*{T^{L+2}}
\end{equation}
as \(T \to 0\).
           \item In particular, if $(H(t))_{t\in\R}$ also satisfies \cref{assump:A-H-symmetric,assump:A-H-conj}, then \(\Heff{L}{T}\) admits at least one self-adjoint extension \(\Heffext{L}{T}\), and the estimates in \cref{eq:general-error-ip-2-summary,eq:general-error-ip-3-summary} hold.
    \end{statements}
\end{maintheorem}
In fact, our approach not only yields the scaling of the difference with $T$ as in \cref{eq:general-error-ip-2-summary}, but also provides explicit and computable error bounds.

We now turn to the relation between our construction and the Floquet--Magnus expansion.
To this end, let us first briefly consider the case where the operators $(H(t))_{t\in\R}$ (and hence, for every fixed $T$, $(H^{(T)}(t))_{t\in\R}$) are \emph{bounded}. Then, Floquet theory guarantees the existence of an effective Hamiltonian $\HFM{}{T}$ and a $T$-periodic unitary family $(P(t))_{t\in\R}$ such that
\begin{equation}
    U^{(T)}(t)
    = P(t) \e^{-\iu t \HFM{}{T}},
    \quad t\in\R
    .
\end{equation}
If the following estimate\footnote{
We remark that, in order for \emph{both} \((P(t))_{t\in\R}\) and \(\HFM{}{T}\) to admit a series expansion, the following stricter condition is known \cite{Blanes_Casas_Oteo_Ros_2009}:
\[
\int_{0}^{T} \opnorm{H^{(T)}(\tau)} \dd{\tau}
     < \xi_F
 = 0.20925.
\]
} holds true:
\begin{equation}\label{eq:FM-convergence-condition}
\int_{0}^{T} \opnorm{H^{(T)}(\tau)} \dd{\tau}
        < \pi,
\end{equation}
with $\opnorm{\cdot}$ denoting the operator norm, then $\HFM{}{T}$ admits a series expansion, known as the Floquet--Magnus expansion~\cite[Thm.~9]{Blanes_Casas_Oteo_Ros_2009}. We refer to \cite{Blanes_Casas_Oteo_Ros_2009} for an overview on the Floquet--Magnus expansion and some results; this reference mainly works in the finite-dimensional setting, but in light of \cref{rem:Blanes-infinite-dim} these results are valid in the infinite-dimensional bounded case as well.
In this case, we denote by
\begin{equation}\label{eq:HFM-non-poly-expansion}
    \HFM{L}{T}
    = \sum_{l=0}^{L} \HFMl{l}{T}
\end{equation}
its truncation at order $L$.
The first two orders take the form
\begin{align}
    \HFMl{0}{T}
    &= \frac{1}{T} \int_{0}^{T} H^{(T)}(\tau) \dd{\tau}
    \\
    \HFMl{1}{T}
    &= - \frac{\iu}{2T} \int_{0}^{T} \of*{ \int_{0}^{\tau_1} \comm{H^{(T)}(\tau_1)}{H^{(T)}(\tau_2)} \dd{\tau_2} }\dd{\tau_1},
\end{align}
and higher orders are given in a similar form by nested commutators, cf. \cref{def:floquet-magnus} and \cite[Eqs.~(136) and (51)]{Blanes_Casas_Oteo_Ros_2009}.
This construction is tailored to reproduce the dynamics exactly at integer multiples of the period \(T\): if the Floquet--Magnus expansion converges, one has
\begin{equation}
    U^{(T)}(qT)
    = \e^{- \iu qT \HFM{}{T}}
    ,\quad \forall q \in \mathbb{Z}.
\end{equation}
Moreover, the approximation error at order \(L\) satisfies
\begin{equation}
    \opnorm*{U^{(T)}(qT)
    - \e^{- \iu qT \HFM{L}{T}}}
    = \bigo\of*{T^{L+2}}
    ,\quad \text{as } T \to 0.
\end{equation}
In Ref.~\cite{deyErrorBoundsFloquetMagnus2025}, it was shown in the bounded setting that this asymptotic behaviour remains valid even in the absence of convergence of the Floquet--Magnus expansion, and that \(\HFM{L}{T}\) coincides with the operator \(\Heff{L}{T}\) given by \cref{eq:Heff-definition-0-summary,eq:Heff-definition-l-summary} at every finite order \(L \in \N_0\).

Returning now to the case where the operators \(H(t)\) are \emph{unbounded}, little is known in the literature about the Floquet--Magnus expansion. In particular, the convergence condition~\eqref{eq:FM-convergence-condition} is no longer applicable. Moreover, the operators defined by the formal Floquet--Magnus expansion in terms of nested commutators, cf. \cref{eq:FM-expansion-full-expression00,eq:FM-expansion-full-expression0,eq:FM-expansion-full-expression}, are in general not even densely defined, let alone self-adjoint. Our third main result shows that, under assumption of \assumpA{}, each finite-order truncation of the Floquet--Magnus expansion coincides with the corresponding order in the construction of \cref{th-S:U-diff-periods}. In particular, these truncated operators are densely defined, admit self-adjoint extensions, and satisfy the same asymptotic approximation properties as \(\Heff{L}{T}\).

\begin{maintheorem}[{see \Cref{thm:FM-Eff-equal-coefficients-unbounded,thm:general-error-ip}}]\label{th-S:FM-comparison}
   Let \((H(t))_{t\in\R}\) satisfy \cref{assump:A-H-cinf,assump:A-U-cinf,assump:A-H-cont,assump:A-U-diff,assump:A-H-bound,assump:A-U-bound,assump:A-H-period,assump:A-H-symmetric,assump:A-H-conj}, and let \(\HFM{L}{T}\) be
   the \(L\)th truncation \(\HFM{L}{T}\) of the corresponding Floquet--Magnus expansion. Then \(\HFM{L}{T}\) is a well-defined operator on \(\cinftyofh\) and coincides with \(\Heff{L}{T}\).
In particular, all conclusions of \cref{th-S:U-diff-periods} extend to \(\HFM{L}{T}\): it is symmetric on \(\cinftyofh\), admits at least one self-adjoint extension \(\HFMext{L}{T}\), and satisfies
    \begin{align}
        \label{eq:FM-general-error-ip-2}
        \of*{\e^{-\iu t \HFMext{L}{T} } - U^{(T)}(t)}\psi
        &= \bigo\of*{T}
    \shortintertext{and}
        \label{eq:FM-general-error-ip-3}
        \of*{\e^{-\iu qT \HFMext{L}{T} } - U^{(T)}(qT)}\psi
        &= \bigo\of*{T^{L+2}},
    \end{align}
    for every \(q\in \Z\), \(t \in \R\), and \(\psi \in \cinftyofh\).
\end{maintheorem}

Thus, our construction provides a natural extension of the Floquet--Magnus expansion to the unbounded setting.
We emphasise again that our approach allows one to define and analyse finite-order truncations of the Floquet--Magnus expansion without requiring convergence of the series. More importantly, in contrast to standard results, where quantitative error bounds typically rely on convergence, our framework yields explicit error estimates for \cref{eq:FM-general-error-ip-2,eq:FM-general-error-ip-3}, independently of any convergence properties of the expansion. At the same time, we note that our method does not directly address the question of convergence of the Floquet--Magnus series itself in the unbounded setting, which remains a highly nontrivial problem.

\subsection{Time-independent Hamiltonians in the interaction picture}\label{sec:assumption-b}
We now specialise our discussion to time-periodic Hamiltonians obtained from a given time-independent operator via an interaction-picture transformation.
This framework is widely used in physics to facilitate the analysis of quantum dynamics by isolating a reference evolution and treating the remaining interaction in a more tractable form.

To this end, let $H_0$ be a fixed (possibly unbounded) self-adjoint operator on $\HH$, let $V$ be another operator, and consider the Hamiltonian $H_0 + V$.
Correspondingly, we define the interaction-picture Hamiltonian as
\begin{equation}
\label{eq:def-interaction-picture-summary}
    H(t)
    \coloneqq  \e^{\iu t H_0}V \e^{-\iu t H_0}
    , \quad t\in\R.
\end{equation}
The associated propagator is given by $U(t) = \e^{\iu t H_0} \e^{-\iu t (H_0 + V)}$, which, at a formal level, solves the Schrödinger equation generated by $H(t)$.

In this case, we consider the following set of properties:
\begin{assumption}\label{assump:B}
Let $\HH$ be a separable Hilbert space, $H_0$ a fixed (possibly unbounded) self-adjoint operator on $\HH$ with domain $\domof{H_0}$, $V$ another operator with domain $\domof{V}\supseteq\domof{H_0}$, and \(H(t)\) be given as in \cref{eq:def-interaction-picture-summary}.
\begin{subassumption}\label{assump:B-V-sym}
   $V$ is symmetric and $\frac{1}{T}H_0 + V$ is self-adjoint on $\domof{H_0}$ for any \(T>0\).
\end{subassumption}
\begin{subassumption}\label{assump:B-V-bound}
       For every $m\in\N_{0}$, \(H_0^m V\) is infinitesimally \(H_0^{m+1}\)-bounded, i.e.
        \begin{gather}
            \domof{H_0^{m+1}} \subseteq \domof{H_0^m V}
            \intertext{and for any $a_m >0$ there is $b_m \geq 0$ such that}
            \norm{H_0^m V \psi}
            \leq a_m \norm{H_0^{m+1}\psi} + b_m \norm{\psi},
        \end{gather}
        for every \(\psi \in \domof{H_0^{m+1}}\).
\end{subassumption}
\begin{subassumption}[{cf.~\ref{assump:A-H-period}}]\label{assump:B-V-period}
    $H(t+1)= H(t)$, for all $t\in\R$, i.e., $t\mapsto H(t)$ is 1-periodic.
\end{subassumption}
\begin{subassumption}\label{assump:B-V-conj}
    There exists a conjugation $J$ on $\HH$ such that $J V \subseteq V J$ and $J H_0 \subseteq H_0 J$.
\end{subassumption}
\end{assumption}
We remark that, by the Kato--Rellich theorem, \cref{assump:B-V-sym} holds, in particular, if \(V\) is symmetric on \(\domof{H_0}\) and \cref{assump:B-V-bound} holds.

In \cref{sec:interaction-picture}, we show how \assumpB{} imply \assumpA{}, one by one.
We summarise this result as follows:
\begin{maintheorem}\label{th:summary-interaction-pic}
    Let $V$ and \(H_0\) be operators on $\HH$ as above satisfying \cref{assump:B-V-bound,assump:B-V-sym,assump:B-V-conj,assump:B-V-period,assump:B-V-conj}.
    Then the time evolution unitary families \(\of*{U^{(T)}(t)}_{t \in \R}\) given by \(U^{(T)}(t) = \e^{+\iu \frac{t}{T} H_0}\e^{-\iu t \of*{\frac{1}{T}H_0+V}}\), for all $T>0$, and the time-dependent Hamiltonian $(H(t))_{t\in\R}$ given by \(H(t) = \e^{\iu t H_0} V \e^{-\iu t H_0}\) satisfy \cref{assump:A-H-cinf,assump:A-U-cinf,assump:A-H-cont,assump:A-U-diff,assump:A-H-bound,assump:A-U-bound,assump:A-H-period,assump:A-H-symmetric,assump:A-H-conj}.
    In particular, the results of \cref{th-S:U-diff-periods,th-S:FM-comparison} hold in this setting.
\end{maintheorem}
We illustrate this in \cref{sec:quantum-Rabi} with the quantum Rabi model.

\subsection{Time-dependent self-adjoint Hamiltonians with constant domain}\label{sec:assumption-c}

While the previous discussion addresses situations in which the time dependence of the Hamiltonian arises from an interaction-picture transformation, in many cases it is instead intrinsic to the Schrödinger picture.
For general unbounded time-dependent Hamiltonians, the well-posedness of the Schrödinger equation is a delicate issue: in particular, the existence and uniqueness of a unitary propagator $(U(t))_{t\in\R}$ associated with the Hamiltonian---assumed \emph{a priori} in the abstract framework of \assumpA{}, and automatically granted in the interaction picture case of \assumpB{}---are not guaranteed in general.

Two important settings in which this problem can be addressed are the so-called hyperbolic and parabolic cases, originating in the work of Kato and further developed in subsequent contributions (see, e.g., \cite[Ch.~5]{pazySemigroupsLinearOperators1983}). In the present work, we focus on the hyperbolic setting.
Under additional regularity assumptions, we will ensure existence and uniqueness of the dynamics $(U(t))_{t\in\R}$ and, at the same time, recover \assumpA{} within our framework.

A key advantage of these assumptions is that they are formulated directly in terms of the Hamiltonian $H(t)$, without requiring prior knowledge of the associated propagator.
As such, they are not only more readily verifiable in concrete situations, but also naturally satisfied by a broad class of physically relevant models.
\begin{assumption}\label{assump:C}
    Let \((H(t))_{t\in \R}\) be a family of operators \(H(t)\colon \domof{H(t)}\subseteq\HH \to \HH\), and let \(H_0\colon\domof{H_0}\subseteq\HH\to \HH\) be a positive self-adjoint operator.
	\begin{subassumption}\label{assump:C-H-self-adjoint}
	    For every \(t\in \R\), \(H(t)\) is self-adjoint.
	    Furthermore, the domains of all \(H(t)\) coincide with the domain of \(H_0\):
		\begin{equation}
			\domof{H(t)}
			= \domof{H(0)}
			= \domof{H_0}
            ,\quad \forall t \in \R
			.
		\end{equation}
	\end{subassumption}
	\begin{subassumption}\label{assump:C-H-bound}
	   For every \(m\in\N_0\) and \(t \in\R\), $H_0^m (H(t) - H_0)$ is infinitesimally $H_0^{m+1}$-bounded, i.e.,
		\begin{equation}
			\domof{H_0^{m + 1} }
			\subseteq \domof{H_0^m (H(t)-H_0)};
		\end{equation}
		moreover, for every compact interval $I$ and any \(\epsilon >0\) there is \(\delta \geq 0\) such that
		\begin{equation}\label{eq:assump:C-H-bound}
			\norm{H_0^m (H(t)-H_0)\psi}
			\leq \epsilon \norm{H_0^{m + 1}\psi} + \delta\norm{\psi}
		\end{equation}
        for all $\psi \in \domof{H_0^{m + 1}}$ and $t\in I$.
	\end{subassumption}
	  \begin{subassumption}\label{assump:C-H-diff}
        For every \(m \in \N_0\), \(\psi \in \domof{H_0^{m+1}}\) and $t\in\R$, the map \(s\mapsto H_0^m H(s)\psi\) is continuously differentiable, with \(\odv{}{t} H_0^m H(t)\psi=H_0^m H'(t)\psi\) and \(H'(t)\) satisfying \cref{assump:C-H-bound} with $H'(t)$ in place of $H(t)$.
    \end{subassumption}
    \begin{subassumption}[{cf.~\ref{assump:A-H-period}}]\label{assump:C-H-periodic}
        The map \(t \mapsto H(t)\) is \(1\)-periodic, i.e.,
        \begin{equation}
            H(t+1) = H(t)            ,\quad \forall t \in \R
            .
        \end{equation}
    \end{subassumption}
	\begin{subassumption}\label{assump:C-commutator-bound}
		For every \(m \in \N_0\) and for any compact interval $I$, there exist \(\tilde{a}_m ,\, \tilde{b}_m \geq 0\) such that
		\begin{equation}\label{eq:assump:C-commutator-bound}
			\norm*{\comm{H_0^m}{H(t)}\psi}
			\leq \tilde{a}_m\norm*{H_0^{m}\psi} + \tilde{b}_m\norm*{\psi}
			,
		\end{equation}
        for all $\psi \in \cinftyofh$ and $t\in I$.
    \end{subassumption}
    \begin{subassumption}[cf.~\ref{assump:A-H-conj}]\label{assump:C-H-conj}
        There exists a conjugation $J$ on $\HH$ such that $J H(t) \subseteq H(-t) J$, for all $t\in\R$ and \(J \cinftyofh \subseteq \cinftyofh\).
    \end{subassumption}
\end{assumption}
Note that, if $H_0$ is not positive but only bounded from below, the situation can be reduced to the positive case by a constant energy shift, which leaves the dynamics unchanged.

In \cref{sec:hyperbolic-setting}, we show that \assumpC{} imply \assumpA{}.
More precisely, the assumptions of Kato’s theorem allow us to establish \cref{assump:A-H-cinf,assump:A-U-cinf,assump:A-U-diff,assump:A-U-bound,assump:A-H-cont}, while \cref{assump:A-H-conj,assump:A-H-bound,assump:A-H-symmetric,assump:A-H-period} follow directly from the corresponding assumptions in \assumpC{}.
We summarise the outcome of \cref{sec:hyperbolic-setting} in the following result:
\begin{maintheorem}\label{th:summary-hyperbolic-setting}
    Let \((H(t))_{t\in \R}\) be a family of operators \(H(t)\colon \domof{H(t)} \to \HH\) and let \(H_0\colon\domof{H_0}\to \HH\) be a self-adjoint, positive operator such that \assumpC{} are fulfilled.
    Then, for every $T>0$, there is a unique strongly continuous unitary family $(U^{(T)}(t,s))_{s,t\in\R}$ such that, for every \(\psi_0 \in \domof{H_0}\) and fixed $s\in\R$, \(\psi(t) = U^{(T)}(t,s)\psi_0\) is the unique solution of the initial value problem
    \begin{equation}\label{eq:non-autonomous-schrödinger-equaiton-summary}
        \begin{cases}
            \odv{}{t} \psi(t)
            = -\iu H^{(T)}(t)\psi(t),
            & \forall t\in \R,
            \\
            \psi(s)
            = \psi_0,
        \end{cases}
    \end{equation}
    with \(t \mapsto \psi(t) \in \domof{H_0}\) being continuous in \(\domof{H_0}\) with respect to the graph norm of $H_0$, and continuously differentiable in \(\HH\) for every \(t\in\R\).
    Moreover, the operators $H_0$, and the families $(H(t))_{t\in\R}$ and $(U^{(T)}(t)\coloneqq U^{(T)}(t,0))_{t\in\R}$, for every $T>0$, satisfy \assumpA{}.
    In particular, the results of \cref{th-S:U-diff-periods} and \cref{th-S:FM-comparison} are valid.
\end{maintheorem}

We illustrate this in \cref{sec:periodically-driven-oscillator} with the periodically driven quantum harmonic oscillator.

\section{The iterated integration-by-parts formula}\label{sec:iterated-integration-by-parts}

In this section, we develop the technical framework required for the proof of \cref{th-S:U-diff-periods}, which consists of an iterated integration-by-parts formula yielding an expression for the difference of two unitary propagators, generalising \cite[Eq. (11)]{deyErrorBoundsFloquetMagnus2025} to unbounded generators.
To this end, we first introduce the necessary notation and establish a series of auxiliary results.

We begin by recalling the definition of Bochner integral, which will be used throughout the paper.

\begin{definition}[{\cite{bochnerIntegrationFunktionenDeren1933}, \cite[Ch.~V.5]{yosidaFunctionalAnalysis1995}}]
    A map \(\psi \colon\: \R \to \HH\) is said to be \emph{Bochner measurable} if there exists a sequence \((\psi_n)_{n\in\N}\) of simple functions \(\psi_n\colon\:\R \to \HH\) such that
    \begin{equation}
        \psi_n(t) \to \psi(t)
    \end{equation}
    for almost all \(t\in \R\).
    A Bochner measurable map \(\psi\) is \emph{Bochner integrable} if it is Bochner measurable and
    \begin{equation}
        \lim_{n\to \infty} \int_\R \norm*{ \psi_n(t) - \psi(t)} \dd t = 0.
    \end{equation}
    In this case the \emph{Bochner integral} of \(\psi\) is defined by
    \begin{equation}
        \int_\R \psi(t) \dd t = \lim_{n\to \infty} \int_\R \psi_n(t) \dd t
        ,
    \end{equation}
    where the integral of simple functions is defined in the standard way by linearity.
\end{definition}

For convenience, we recall here some useful properties of the Bochner integral.
\begin{theorem}[{\cite{bochnerIntegrationFunktionenDeren1933}}]
\label{thm:bochner-integrability-test}
    A Bochner
    measurable map \(\psi\colon\:
    \R
    \to \HH\) is Bochner
    integrable if and only if \(
    t
    \mapsto \norm{\psi(t)}\) is
    integrable.
\end{theorem}
\begin{theorem}[Hille's theorem, cf Thm. 3.7.12 in \cite{hilleFunctionalAnalysisSemigroups1957}]\label{thm:Hille}
    Let $A$ be a closed operator on $\HH$, and let  \(\psi\colon\: \R\to\domof{A}\) be a Bochner integrable map such that \(A\psi\colon\: \R\to\HH\) is also Bochner integrable. Then
    \begin{equation}
        A \int_a^b \psi(s)\dd s = \int_a^b A\psi(s) \dd s
    \end{equation}
    for every $a,b\in\R$.
\end{theorem}
Hereafter, $\mathcal{L}\of{\HH}$ will denote the set of all (possibly unbounded) linear operators on $\HH$.
\begin{definition}\label{def:integral-of-operator}
    Given \(A\colon\: \R \to \mathcal{L}\of{\HH}\) and $t\in\R$, we define the operator $\int_0^tA(s)\dd s$ via
    \begin{gather}
        \domof*{\int_0^t A(s)\dd s}
        = \Set*{
            \psi \in \HH \given
            \begin{aligned}
             &\forall s \in [0,t]\colon \psi \in \domof{A(s)},
             \\
             &s \mapsto A(s)\psi\ \text{Bochner integrable on } [0,t]
             \end{aligned}
             }
        \\
        \left(\int_0^t A(s)\dd s\: \right)\psi = \int_0^t A(s)\psi \dd s
        .
    \end{gather}
    Moreover, for every $T>0$, the $T$-\emph{average} $\avg{T}{A}$ of \(A\colon\: \R \to \mathcal{L}\of{\HH}\) is the operator defined by
    \begin{equation}
        \avg*{T}{A} = \frac{1}{T} \int_0^T A(s) \dd s.
    \end{equation}
    At this stage, the parameter $T$ is purely formal and should not be interpreted as a period. This interpretation will only become relevant later, beginning with \cref{sec:Heff}.
\end{definition}
With the framework of Bochner and operator-valued integrals in place, we now introduce the objects appearing in the integration-by-parts formula.

\begin{definition}\label{def:Kop}
    Consider two families of operators, \(\of*{H_1(t)}_{t\in \R}\) and \(\of*{H_2(t)}_{t\in \R}\).
    For every \(t\in \R\) and every linear operator \(A\), we define the operator
    \begin{equation}\label{eq:def:Kop}
        \Kop{t}\of{A} = H_1(t) A - A H_2(t)
    \end{equation}
    on its standard domain.
\end{definition}
\begin{definition}
\label{def:iterated-actions}
    Let \(T > 0 \), and let \(\of*{H_1(t)}_{t\in \R}\), \(\of*{H_2(t)}_{t\in \R}\) be two families of operators.
    For every $j\in\N_0$, we define the operator \(\Sop{j}\of{t}\) iteratively by setting
    \begin{gather}
        \Sop{0}\of{t} = \one,
        \\
        \label{eq:def:iterated-actions}
        \Sop{j}\of{t}
        = \int_0^t \Kop{s}\of{\Sop{j-1}(s)} \dd s - t \avg{T}{\Kop{\cdot}\of{\Sop{j-1}(\cdot)}}
        ,
        \quad \forall j \in \N,
    \end{gather}
    for all \(t\in \R\), on its standard domain.
\end{definition}

In general, the domains of the operators \(\Kop{t}\of*{A}\) and \(\Sop{j}(t)\) as specified in \cref{def:integral-of-operator} need not be dense, or even nontrivial. \cref{thm:properties-iterated-actions} shows that this pathology does not occur in our setting and establishes several additional properties that will be used throughout the remainder of the paper. In particular, for every \(t \in \R\) and \(j \in \N\), one has \(\cinftyofh \subseteq \domof{\Sop{j}(t)}\). We begin with a preliminary lemma:

\begin{lemma}\label{thm:strong-continuity-monomials}
    Let \(H_0\) be self-adjoint and $(H(t))_{t\in\R}$ a family of operators satisfying \cref{assump:A-H-cinf,assump:A-H-cont,assump:A-H-bound}.
    For every \(j \in \N\), \(m \in \N_0\) and \(\psi \in \cinftyofh\), the map
    \begin{equation}
        \underline{s} = (s_1,\dots,s_j)\in \R^j \mapsto H_0^m H(s_1) \dots H(s_j)\psi
    \end{equation}
    is continuous, and hence integrable over any compact subset \(\mathcal{M}\subseteq \R^j\).
\end{lemma}
\begin{proof}
    We prove the claim by induction over \(j \in \N\).
    For $j=1$ this is simply \cref{assump:A-H-cont}.
    Now assume the claim holds for some \(j \in \N\).
    Then, for $\psi\in\cinftyofh$, we have
    \begin{align}
        \MoveEqLeft
        \norm{\of{H_0^m H(s_1) \dots H(s_j)H(s_{j+1}) - H_0^m H(\hat{s}_1) \dots H(\hat{s}_j)H(\hat{s}_{j+1})}\psi}
        \nonumber\\
        &\leq\begin{multlined}[t]
            \norm{\of{H_0^m H(s_1) \dots H(s_{j}) -H_0^m H(\hat{s}_1) \dots H(\hat{s}_j)} H(\hat{s}_{j+1})\psi}
            \nonumber\\
            +\norm{H_0^m H(s_1) \dots H(s_j) \of{H(s_{j+1})-H(\hat{s}_{j+1})}\psi}
        \end{multlined}
        \nonumber\\
        &\leq \begin{multlined}[t]
            \underbrace{\norm{\of{H_0^m H(s_1) \dots H(s_j) -H_0^m H(\hat{s}_1) \dots H(\hat{s}_j)} H(\hat{s}_{j+1})\psi}}_{\to 0,\ \text{by induction assumption for}\ j\ \text{and \cref{assump:A-H-cinf}}}
            \nonumber\\
            +a_{m,j}\underbrace{\norm{H_0^{k_m^{(j)}}\of{H(s_{j+1})-H(\hat{s}_{j+1})}\psi}}_{\to 0\ \text{by base case}\ j=1}
            +b_{m,j}\underbrace{\norm{\of{H(s_{j+1})-H(\hat{s}_{j+1})}\psi}}_{\to 0\ \text{by base case}\ j=1}
        \end{multlined}
        \nonumber\\
        &\to 0
    \end{align}
    as \(\underline{s} \to \hat{\underline{s}}\), where \(k_m^{(1)} \coloneqq m + k_m\) and \(k_m^{(j+1)} \coloneqq k_m^{(j)} + k_{k_m^{(j)}}\) and with certain constants $a_{m,j}, b_{m,j}\geq 0$, using \cref{assump:A-H-bound}.
    So the claim holds for $j+1$ as well.
    From that, integrability follows immediately due to \cref{thm:bochner-integrability-test}.
\end{proof}

\begin{proposition}\label{thm:properties-iterated-actions}
    Let \(T >0\), \(H_0\) be a self-adjoint operator, and $(H_1(t))_{t\in\R}, (H_2(t))_{t\in\R}$ satisfy \cref{assump:A-H-cinf,assump:A-H-cont,assump:A-H-bound}.
    For every \(j \in \N_0\), the following properties hold:
    \begin{statements}
        \item\label{itm:properties-iterated-actions-domain} \(\displaystyle \cinftyofh \subseteq \domof{\Kop{t}\of{\Sop{j}(t)}}\) for all $t\in\R$;
        \item\label{itm:properties-iterated-actions-integrability} For all $\psi \in \cinfty(H_0)$ and $m\in \N_0$ and $t\in\R$, \(\displaystyle s\mapsto H_0^m \Kop{s}\of{\Sop{j}(s)}\psi\) is continuous and Bochner integrable on \([0,t]\);
        \item\label{itm:properties-iterated-actions-rel-boundedness} For all $m\geq 0$ and every compact interval $I\subseteq\R$ there exist $\tilde{k}_m \in \N_0$ and $\tilde{a}_m,\,\tilde{b}_m\geq 0$ such that
        \begin{equation}\label{eq:properties-iterated-actions-rel-boundedness}
            \norm{H_0^m \Sop{j}(t)\psi}
            \leq \tilde{a}_m \norm{H_0^{m+\tilde{k}_m}\psi} + \tilde{b}_m\norm{\psi}
        \end{equation}
        for all $\psi\in\cinftyofh$ and all $t\in I$;
        \item\label{itm:properties-iterated-actions-invariance} \(\displaystyle\Sop{j}(t) \cinfty\of{H_0} \subseteq \cinfty\of{H_0}\) for all $t\in\R$.
    \end{statements}
    In particular, \(\cinftyofh \subseteq \domof{\Sop{j}(t)}\) for every \(j \in \N_0\) and \(t \in \R\).
\end{proposition}

\begin{proof}
    The proof is by induction over \(j \in \N_0\).
    For \(j = 0\) \cref{itm:properties-iterated-actions-rel-boundedness,itm:properties-iterated-actions-invariance} hold trivially, because \(\Sop{0}(t) = \one\) is bounded for every \(t\in \R\).
    \Cref{itm:properties-iterated-actions-domain} follows directly from \cref{assump:A-H-cinf}, and \cref{itm:properties-iterated-actions-integrability} from \cref{assump:A-H-cont}.

    Now, assume \cref{itm:properties-iterated-actions-domain,itm:properties-iterated-actions-integrability,itm:properties-iterated-actions-invariance,itm:properties-iterated-actions-rel-boundedness} hold for some \(j \in \N_0\).
    We show that this implies \cref{itm:properties-iterated-actions-domain,itm:properties-iterated-actions-integrability,itm:properties-iterated-actions-invariance,itm:properties-iterated-actions-rel-boundedness} also hold for \(j+1\).
    Due to \cref{itm:properties-iterated-actions-integrability,itm:properties-iterated-actions-domain} for \(j\),
    \begin{equation}
        \Sop{j+1}(t)\psi
        = \int_0^t \Kop{s}\of{\Sop{j}(s)}\psi \dd s - \frac{t}{T} \int_0^T \Kop{s}\of{\Sop{j}(s)}\psi \dd s
    \end{equation}
    is well-defined for \(\psi \in \cinftyofh\).
    Moreover, by
\cref{itm:properties-iterated-actions-integrability},
for every \(m\in\N_0\) the map
$s\mapsto H_0^m\Kop{s}\of{\Sop{j}(s)}\psi$ is Bochner integrable. Hence, by repeated application of \cref{thm:Hille}, the operator \(H_0^m\) can be interchanged with the integrals above, which shows $\Sop{j+1}(t)\psi \in \dom(H_0^m)$ for all \(m\in\N_0\), thereby proving \cref{itm:properties-iterated-actions-invariance} for $j+1$.
    The invariance \(\Sop{j+1}(t) \cinftyofh \subseteq \cinftyofh\) in \cref{itm:properties-iterated-actions-invariance} then implies \cref{itm:properties-iterated-actions-domain} for \(j+1\) by definition of $\Kop{}$.

    Regarding \cref{itm:properties-iterated-actions-rel-boundedness}, we begin with showing that \(H_0^{m} \Kop{s}\of*{\Sop{j}(s)}\) is relatively \(H_0^{m+\hat{k}_m}\)-bounded for some \(\hat{k}_m\):
        \begin{align}
            \norm{H_0^m \Kop{s}\of*{\Sop{j}(s)} \psi}
            &\leq \norm{H_0^m H_1(s) \Sop{j}(s)\psi}
            + \norm{H_0^m \Sop{j}(s) H_2(s)\psi}
            \nonumber\\
            &\leq
            \begin{multlined}[t]
                a_m \norm*{H_0^{m+ k_m}\Sop{j}(s)\psi}
                + b_m \norm*{\Sop{j}(s)\psi}
                + \tilde{a}_m \norm*{H_0^{m+\tilde{k}_m}H_2(s)\psi}
                + \tilde{b}_m \norm*{H_2(s)\psi}
            \end{multlined}
            \nonumber\\
            &\leq
            \begin{aligned}[t]
                &a_m \of*{
                    \tilde{a}_{m+ k_m}\norm*{H_0^{m+ k_m + \tilde{k}_{m+ k_m}}\psi}
                    + \tilde{b}_{m+ k_m}\norm*{\psi}
                }
                + b_m \of*{
                    \tilde{a}_0\norm*{H_0^{\tilde{k}_0}\psi}
                    + \tilde{b}_0 \norm*{\psi}
                }
            \nonumber\\
                &+ \tilde{a}_m \of*{
                    a_{m+\tilde{k}_m} \norm*{H_0^{m+\tilde{k}_m + k_{m+\tilde{k}_m}} \psi}
                    + b_{m+\tilde{k}_m} \norm*{\psi}
                }
                + \tilde{b}_m\of*{
                    a_0 \norm*{H_0^{k_0}\psi}
                    + b_0 \norm*{\psi}
                }
            \end{aligned}
            \nonumber\\
            &\leq \hat{a}_m \norm*{H_0^{m+\hat{k}_m}\psi} + \hat{b}_m \norm*{\psi}
            \label{eq:properties-iterated-actions-induction-step-proof-9}
            ,
        \end{align}
    with suitable $\hat{a}_m,\,\hat{b}_m\geq0$, and $\hat{k}_m \in \N_0$. Here we used \cref{assump:A-H-bound} and \cref{itm:properties-iterated-actions-rel-boundedness} with $j$.
    We can use this to show \cref{itm:properties-iterated-actions-rel-boundedness} for $j+1$:
    \begin{align}
        \norm*{H_0^m \Sop{j+1}(t)\psi}
        &\leq  \sgn(t)\int_0^t \norm*{H_0^m\Kop{s}\of*{\Sop{j}(s)}\psi} \dd s + \frac{|t|}{T} \int_0^T \norm*{H_0^m \Kop{s}\of*{\Sop{j}(s)}\psi} \,\dd s\,
        \nonumber\\
        &\leq 2\max_{s\in I} (|s|)\big(\hat{a}_m\norm{H_0^{m+\hat{k}_m}\psi} + \hat{b}_m\norm{\psi}\big)
        ,
    \end{align}
    for every \(t\in I\), where \(\sgn(x) = x / \abs{x}\) if $x\not=0$ and $\sgn(0)=0$. This is of the form in \cref{eq:properties-iterated-actions-rel-boundedness} with $j+1$, hence completing our induction step.

    We now show \cref{itm:properties-iterated-actions-integrability} for $j+1$. We first claim that \(s\mapsto H_0^m \Sop{j+1}(s)\) is strongly continuous for every $m\in\N_0$. Indeed, we can write
    \begin{equation}
        \of*{\Sop{j+1}(s) - \Sop{j+1}(s_0)}\psi = \int_{s_0}^{s} \Kop{r}\of*{\Sop{j}(r)} \psi \dd r - \frac{s - s_0}{T} \int_{0}^T \Kop{r}\of*{\Sop{j}(r)} \psi \dd r,
    \end{equation}
    so that we have
    \begin{align}
        \MoveEqLeft
        \norm*{ H_0^m \of*{\Sop{j+1}(s) - \Sop{j+1}(s_0)}\psi }
        \nonumber\\
        &\leq
        \begin{aligned}[t]
             \sgn(s - s_0)\int_{s_0}^{s} \norm*{H_0^m\Kop{r}\of*{\Sop{j}(r)} \psi}\dd r
            + \frac{\abs{s - s_0}}{T} \int_{0}^T \norm*{H_0^m\Kop{r}\of*{\Sop{j}(r)} \psi} \dd r
        \end{aligned}
        \nonumber\\
        &
        \begin{aligned}[t]
             \to0
        \end{aligned}
        \label{eq:properties-iterated-actions-induction-step-proof-3}
    \end{align}
as $s\to s_0$
    as a consequence of \cref{itm:properties-iterated-actions-integrability} for $j$. This proves our claim. In order to show \cref{itm:properties-iterated-actions-integrability}, it suffices to show continuity of
    \begin{align}
        \phi^{(m)}(s) = H_0^m \Sop{j+1}(s) H_2(s) \psi
    \shortintertext{and}
        \psi^{(m)}(s) = H_0^m H_1(s) \Sop{j+1}(s) \psi
        \label{eq:properties-iterated-actions-induction-step-phi}
    \end{align}
    separately.
    We have
    \begin{align}
        \norm*{\phi^{(m)}(s) - \phi^{(m)}(s_0)}
        &=\begin{multlined}[t]
            \norm*{
                H_0^m \Sop{j+1}(s) H_2(s) \psi
                - H_0^m \Sop{j+1}(s_0) H_2(s_0) \psi
            }
        \end{multlined}
        \nonumber\\
        &\leq
        \begin{aligned}[t]
            &\norm*{
                \of*{H_0^m \Sop{j+1}(s) - H_0^m \Sop{j+1}(s_0)} H_2(s_0)\psi
            }
            +\norm*{
                H_0^m \Sop{j+1}(s) \of*{H_2(s)-H_2(s_0)}\psi
            }
        \end{aligned}
        \nonumber\\
        &\leq
        \begin{aligned}[t]
            &\norm*{
                \of*{H_0^m \Sop{j+1}(s) - H_0^m \Sop{j+1}(s_0)} H_2(s_0)\psi
            }
            + \tilde{a}_m\norm*{
                H_0^{m+\tilde{k}_m} \of*{H_2(s)-H_2(s_0)}\psi
            }
            \\
            &+ \tilde{b}_m\norm*{
                \of*{H_2(s)- H_2(s_0)}\psi
            }
        \end{aligned}
        \nonumber\\
        & \to 0
    \end{align}
    as $s\to s_0$, by strong continuity of \(s\mapsto H_0^m \Sop{j+1}(s)\) proven in the previous step and strong continuity of \(s\mapsto H_0^m H_2(s)\) which holds by \Cref{assump:A-H-cont}. The proof for $\psi^{(m)}$ goes analogously. Hence, \(s\mapsto \psi^{(m)}(s) - \phi^{(m)}(s) = H_0^m \Kop{s}\of*{\Sop{j+1}(s)}\psi\) is continuous and, by \Cref{thm:bochner-integrability-test}, it is integrable over compact intervals.
\end{proof}

We now come to the main result of this section.

\begin{theorem}\label{thm:iterated-integration-by-parts}
    Let $H_0$ be a self-adjoint operator on $\HH$; let \((U_1(t))_{t\in \R}\) be a family of unitary operators with corresponding family \((H_1(t))_{t\in \R}\) such that \cref{assump:A-H-cinf,assump:A-H-cont,assump:A-H-bound} hold for $T=1$, and $t\mapsto U_1(t)^*$ is strongly differentiable on $\cinftyofh$ with
    \begin{equation}\label{eq:assump-diff-U1}
        \frac{\dd}{\dd t}U_1(t)^*\psi = \iu U_1(t)^{*} H_1(t) \psi, \quad \forall\psi\in\cinftyofh,
    \end{equation}
    and $U_1(0)=\one$; furthermore, let \((U_2(t))_{t\in \R}\)  be a family of unitary operators with corresponding family \((H_2(t))_{t\in \R}\) such that \cref{assump:A-H-cinf,assump:A-H-cont,assump:A-U-diff,assump:A-H-bound,assump:A-U-cinf} hold for $T=1$, and $U_2(0)=\one$.
    Then, for every \(\psi \in \cinftyofh\), \(L \in \N_{0}\), T>0, and \(t\in\R\), we have
    \begin{equation}
    \label{eq:iterated-integration-by-parts}
        \of[\big]{U_1(t) - U_2(t)}\psi
        =
        \begin{multlined}[t]
        \sum_{j=1}^{L+1} (-\iu)^j \Sop{j}(t) U_2(t) \psi
        \\
        - \iu \int_0^t U_1(t) U_1(s)\adj \Bigg(
            \sum_{j=0}^L (-\iu)^j \avg*{T}{\Kop{}\of*{\Sop{j}}}
            + (-\iu)^{L+1} \Kop{s}\of*{\Sop{L+1}(s)}
        \Bigg) U_2(s) \psi \dd s
        ,
        \end{multlined}
    \end{equation}
    with $\Kop{s}(\cdot)$ and $\Sop{j}(t)$ as per \cref{def:Kop,def:iterated-actions}.
\end{theorem}

We remark that while \cref{assump:A-U-cinf,assump:A-U-diff} are about $U_2^{(T)}(t)$ for all $T>0$, here we only deal with the case of fixed $T=1$, i.e., $U_2(t)=U_2^{(1)}(t)$, for all $t\in\R$, therefore we require this weakened version of \cref{assump:A-U-cinf,assump:A-U-diff}.

Before proving \Cref{thm:iterated-integration-by-parts}, we prove two additional lemmas of technical nature.

\begin{lemma}
    \label{thm:integrability-with-dynamics}
    Suppose the assumptions of \cref{thm:iterated-integration-by-parts} hold.
    For every $j\in\N_0$, $t\in\R$, and \(\psi \in \cinftyofh\), the mapping
    \begin{equation}
        s \mapsto \psi(s) = U_1(t)U_1(s)\adj \Kop{s}\of*{\Sop{j}(s)}U_2(s)\psi
    \end{equation}
    is integrable on bounded intervals.
\end{lemma}
\begin{proof}
    By \cref{itm:properties-iterated-actions-integrability} of \cref{thm:properties-iterated-actions}, the map
\begin{equation}
    s \mapsto \Kop{s}\of*{\Sop{j}(s)}\psi
\end{equation}
is continuous for every \(\psi\in\cinftyofh\).
Together with the strong continuity of \(U_2\), \cref{assump:A-U-diff}, and the relative boundedness estimate \eqref{eq:properties-iterated-actions-induction-step-proof-9} and \cref{assump:A-U-cinf}, this implies
    \begin{align}
        \MoveEqLeft
        \norm*{ \Kop{s}\of*{\Sop{j}(s)}U_2(s)\psi - \Kop{s_0}\of*{\Sop{j}(s_0)}U_2(s_0)\psi }
        \nonumber\\
        \leq& \norm*{ \of*{ \Kop{s}\of*{\Sop{j}(s)} - \Kop{s_0}\of*{\Sop{j}(s_0)} }U_2(s_0) \psi }
        + \hat{a}_{0} \norm*{H_0^{\hat{k}_0}\of*{U_2(s) - U_2(s_0)}\psi} + \hat{b}_{0} \norm*{\of*{U_2(s) - U_2(s_0)}\psi}
        \nonumber\\
        \to& 0
        ,\quad \text{as}\ s\to s_0,
    \end{align}
    using \cref{eq:properties-iterated-actions-induction-step-proof-9}. Together with strong continuity of $s\mapsto U_1(s)^*$, which follows from the required strong differentiability, this implies continuity of $\psi(s)$. Due to \cref{thm:bochner-integrability-test}, \(\psi(s)\) is integrable.
\end{proof}

The second lemma concerns the applicability of the iterated integration-by-parts procedure.
\begin{lemma}
\label{thm:iterated-differentiability}
    Suppose the assumptions of \cref{thm:iterated-integration-by-parts} hold.
    The functions \(t\mapsto\Sop{j}(t)\) and \(t\mapsto U_1(t)\adj \Sop{j}(t)U_2(t)\) are strongly differentiable on \(\cinftyofh\), and for every \(\psi \in \cinftyofh\) their strong derivatives act as
    \begin{gather}
        \label{eq:derivative-sop}
        \odv{}{t} \Sop{j}(t) \psi
        = \Kop{t}\of*{\Sop{j-1}(t)}\psi - \avg*{T}{\Kop{\cdot}\of*{\Sop{j-1}(\cdot)}}\psi
        ,\quad j\geq 1
        \\
        \label{eq:derivative-U-sop-U}
        \odv{}{t} U_1(t)\adj \Sop{j}(t)U_2(t) \psi
        = \iu U_1(t)\adj \Kop{t}\of*{\Sop{j}(t)}U_2(t)\psi + U_1(t)\adj \of*{\sodv{}{t}\Sop{j}(t)} U_2(t)\psi
        .
    \end{gather}

\end{lemma}
\begin{proof}
    For \(j=0\), the derivatives are trivial, so let \(j\geq 1\). Clearly, \(t \avg*{T}{\Kop{\cdot}\of*{\Sop{j-1}(\cdot)}}\) is strongly differentiable with strong derivative \(\avg*{T}{\Kop{\cdot}\of*{\Sop{j-1}(\cdot)}}\).
    Let \(\psi \in \cinftyofh\).
    Then \(\int_0^t \Kop{s}\of*{\Sop{j-1}(s)}\psi \dd s\) is differentiable for almost all \(t\in\R\) with derivative \(\psi(t)  = \Kop{t}\of*{\Sop{j-1}(t)}\psi\) \cite[Ch.~II, Thm.~9]{diestelVectorMeasures1977a}.
    As \(\psi(t)\) is continuous, the statement also holds for every \(t\in\R\):
    for every \(\epsilon>0\) there is \(\delta > 0\) such that for every \(s\in\R\) with \(\abs{s-t}<\delta\) follows \(\norm{\psi(s) - \psi(t)} < \epsilon\) and hence, for \(0<\abs{h}<\delta\), we have
    \begin{equation}
        \norm*{ \frac{1}{h} \of*{\int_0^{t+h} \psi(s) \dd s - \int_0^{t} \psi(s) \dd s} - \psi(t)}
        \leq \frac{1}{h}\int_t^{t+h} \norm*{\psi(s) - \psi(t)} \dd s < \epsilon
        .
    \end{equation}
    Thus, the claim holds for \(\Sop{j}(t) = \int_0^t \Kop{s}\of*{\Sop{j-1}(s)} \dd s - t \avg*{T}{\Kop{\cdot}\of*{\Sop{j-1}(\cdot)}}\).

    Concerning \cref{eq:derivative-U-sop-U}, for every $t\in\R$ and $\psi\in\cinftyofh$, we write
    \begin{align}
        \MoveEqLeft\nonumber
        \frac{1}{h}\Big( U_1(t+h)\adj \Sop{j}(t+h) U_2(t+h)\psi - U_1(t)\adj \Sop{j}(t) U_2(t)\psi \Big)
        \\
        \label{eq:USU-difference-quotients}
        =&
        \begin{aligned}[t]
            &\frac{1}{h} \Big( U_1(t+h)\adj - U_1(t)\adj \Big) \Sop{j}(t) U_2(t)\psi
            +  U_1(t+h)\adj \frac{1}{h}\Big( \Sop{j}(t+h) - \Sop{j}(t)\Big) U_2(t)\psi
            \\
            & +  U_1(t+h)\adj \Sop{j}(t+h) \frac{1}{h} \Big(U_2(t+h) - U_2(t)\Big)\psi.
        \end{aligned}
    \end{align}
    The first term on the right hand side converges to
    \begin{equation}
        \iu U_1(t)^* H_1(t) S_j(t) U_2(t) \psi,
    \end{equation}
    as $h\to 0$, by \cref{eq:assump-diff-U1}.
    The second term on the right hand side converges to
    \begin{equation}
        U_1(t)^* \Big(\sodv{}{t}\Sop{j}(t)\Big) U_2(t)\psi
    \end{equation}
    by \cref{eq:derivative-sop} proven above and strong continuity of $s\mapsto U_1(s)^*$, which in turn follows from \cref{eq:assump-diff-U1}. By \cref{assump:A-U-diff} for $U_2$, for every $m\in\N_0$ we have
    \begin{equation}
        \frac{1}{h} \Big(H_0^m U_2(t+h) - H_0^m U_2(t)\Big)\psi \to -\iu H_0^m H_2(t) U_2(t)\psi.
    \end{equation}
    Furthermore, using \cref{thm:properties-iterated-actions}, in particular the strong continuity of $s\mapsto S_j(s)$ shown in its proof and \cref{itm:properties-iterated-actions-rel-boundedness}, we get
    \begin{align}
        \MoveEqLeft
        \norm*{ \Sop{j}(t+h) \frac{1}{h}\big(U_2(t+h) - U_2(t)\big)\psi - \Sop{j}(t) (-\iu H_2(t)) U_2(t)\psi }
        \nonumber\\
        &\leq
        \!\begin{multlined}[t]
            \underbrace{\norm*{ \of*{\Sop{j}(t+h) - \Sop{j}(t)} (-\iu H_2(t)) U_2(t)\psi } }_{\to 0}
            \\
            + \underbrace{ \norm*{\tilde{a}_0\eof*{\frac{1}{h}\of*{H_0^{\tilde{k}_0} U_2(t+h) - H_0^{\tilde{k}_0} U_2(t)} - (-\iu H_0^{\tilde{k}_0} H_2(t)) U_2(t)}\psi }}_{\to 0}
            \\
            +  \underbrace{\norm*{\tilde{b}_0\eof*{\frac{1}{h}\of*{U_2(t+h) - U_2(t)} - (-\iu H_2(t)) U_2(t)}\psi }}_{\to 0}
        \end{multlined}
        \nonumber\\
        & \to 0
        ,
    \end{align}
    as $h\to 0$, which shows that the third term in \cref{eq:USU-difference-quotients} converges to
    \begin{equation}
        -\iu U_1(t)^*\Sop{j}(t) H_2(t) U_2(t)\psi,
    \end{equation}
    so altogether, the right hand side of \cref{eq:USU-difference-quotients} converges to
    \begin{equation}
        \begin{multlined}[t]
        \iu U_1(t)^* H_1(t) S_j(t) U_2(t) \psi  + U_1(t)^* \Big(\sodv{}{t}\Sop{j}(t)\Big) U_2(t)\psi -\iu U_1(t)^*\Sop{j}(t) H_2(t) U_2(t)\psi \\
        = \iu U_1(t)^* \Kop{t}(\Sop{j}(t)) U_2(t) \psi + U_1(t)^* \Big(\sodv{}{t}\Sop{j}(t)\Big) U_2(t)\psi,
        \end{multlined}
    \end{equation}
    which proves \cref{eq:derivative-U-sop-U}.
\end{proof}

With \Cref{thm:integrability-with-dynamics,thm:iterated-differentiability} at hand we can prove the main result of this section, \Cref{thm:iterated-integration-by-parts}.

\begin{proof}[Proof of {\Cref{thm:iterated-integration-by-parts}}]
    First notice that, for every \(j \in \N_{0}\), we have
    \begin{align}
        \Kop{s}\of*{\Sop{j}(s)} \psi
        &= \avg*{T}{\Kop{s}\of*{\Sop{j}(s)} }\psi + \Kop{s}\of*{\Sop{j}(s)}\psi - \avg*{T}{\Kop{s}\of*{\Sop{j}(s)} }\psi
        \nonumber\\
        &= \avg*{T}{\Kop{s}\of*{\Sop{j}(s)}}\psi + \odv{}{s} \Sop{j+1}(s)\psi
    \label{eq:iterated-integration-by-parts-proof-1}
    \end{align}
    and
    \begin{equation}
    \label{eq:iterated-integration-by-parts-proof-2}
        U_1(s)\adj \of*{\sodv{}{s}\Sop{j}(s)} U_2(s)
        = \sodv{}{s} \of*{U_1(s)\adj \Sop{j}(s) U_2(s)} - \iu U_1(s)\adj \Kop{s}\of*{\Sop{j}(s)} U_2(s)
        ,
    \end{equation}
    due to \Cref{thm:iterated-differentiability}.
    This allows us to rewrite certain integrals involving \(\Sop{j}(s)\) in terms of similar integrals involving \(\Sop{j+1}(s)\):
    \begin{align}
    \MoveEqLeft
        \int_0^t U_1(t)U_1(s)\adj \Kop{s}\of*{\Sop{j}(s)} U_2(s) \psi \dd s
        \nonumber\\
        &= \int_0^t U_1(t) U_1(s)\adj \avg*{T}{\Kop{}\of*{\Sop{j}}} U_2(s) \psi \dd s
        + \int_0^t U_1(t)U_1(s)\adj \of*{\sodv{}{s}\Sop{j+1}(s)} U_2(s) \psi \dd s
        \nonumber\\
        &=
        \begin{multlined}[t]
            \Sop{j+1}(t) U_2(t) \psi
            + \int_0^t U_1(t)U_1(s)\adj \avg*{T}{\Kop{}\of*{\Sop{j}}} U_2(s)\psi \dd s
            \\
            + (-\iu) \int_0^t U_1(t) U_1(s)\adj \Kop{s}\of*{\Sop{j+1}(s)} U_2(s) \psi \dd s
            ,
        \end{multlined}
    \label{eq:iterated-integration-by-parts-proof-4}
    \end{align}
    where we used \cref{eq:iterated-integration-by-parts-proof-1} in the first step, and \cref{eq:iterated-integration-by-parts-proof-2} in the second step.

    We can now prove \cref{eq:iterated-integration-by-parts} by induction over $L$. For \(L = 0\), we have
    \begin{align}
        \of*{U_1(t) - U_2(t)}\psi
        &= - U_1(t) \of*{ U_1(t)^* \Sop{0}(t) U_2(t) - \one }\psi
        \nonumber\\
        &= - U_1(t) \int_0^t \odv{}{s} \of*{U_1(s)^* \Sop{0}(s) U_2(s) \psi} \dd s
        \nonumber\\
        &\annotateeqn{\eqref{eq:derivative-U-sop-U}}{=}
        - \iu \int_0^t U_1(t) U_1(s)\adj \Kop{s}\of*{\Sop{0}(s)} U_2(s) \psi \dd s
        \nonumber\\
        &= - \iu \Sop{1}(t) U_2(t) \psi
        - \iu \int_0^t U_1(t) U_1(s)\adj \of*{\avg*{T}{\Kop{}\of*{\Sop{0}}} - \iu \Kop{s}\of*{\Sop{1}(s)} } U_2(s) \psi \dd s
        \nonumber\\
        &=
        \begin{multlined}[t]
            \sum_{j=1}^{L+1} (-\iu)^j \Sop{j}(t) U_2(t) \psi
            \\
            - \iu \int_0^t U_1(t) U_1(s)\adj \of*{ \sum_{j=0}^L (-\iu)^j \avg*{T}{\Kop{}\of*{\Sop{j}}} + (-\iu )^{L+1} \Kop{s}\of*{\Sop{L+1}(s)}} U_2(s) \psi \dd s,
        \end{multlined}
    \end{align}
    which coincides with \cref{eq:iterated-integration-by-parts} for \(L=0\).

    Now suppose \cref{eq:iterated-integration-by-parts} holds for some \(L \in \N_{0}\). We want to show that it also holds for $L+1$. Using \cref{eq:iterated-integration-by-parts-proof-4}, we get
    \begin{align}
        \of*{U_1(t) - U_2(t)}\psi
        &=
        \begin{multlined}[t]
            \sum_{j=1}^{L+1} (-\iu)^j \Sop{j}(t) U_2(t) \psi
            \\
            - \iu \int_0^t U_1(t) U_1(s)\adj \Bigg(
                \sum_{j=0}^L (-\iu)^j \avg*{T}{\Kop{}\of*{\Sop{j}}}
                + (-\iu)^{L+1} \Kop{s}\of*{\Sop{L+1}(s)}
            \Bigg) U_2(s) \psi \dd s
        \end{multlined}
        \nonumber\\
        &=
        \begin{multlined}[t]
            \sum_{j=1}^{L+1} (-\iu)^j \Sop{j}(t) U_2(t) \psi
            - \iu \int_0^t U_1(t) U_1(s)\adj \Bigg(\sum_{j=0}^L (-\iu)^j \avg*{T}{\Kop{}\of*{\Sop{j}}}\Bigg) U_2(s)\psi \dd s
            \\
            + (-\iu)^{L+2} \Sop{L+2}(t) U_2(t) \psi
            -\iu \int_0^t U_1(t)U_1(s)\adj (-\iu)^{L+1} \avg*{T}{\Kop{}\of*{\Sop{L+1}}} U_2(s)\psi \dd s
            \\
            -\iu \int_0^t U_1(t) U_1(s)\adj (-\iu)^{L+2} \Kop{s}\of*{\Sop{L+2}(s)} U_2(s) \psi \dd s
        \end{multlined}
        \nonumber\\
        &=
        \begin{multlined}[t]
            \sum_{j=1}^{L+2} (-\iu)^j \Sop{j}(t) U_2(t) \psi
            \\
            - \iu \int_0^t U_1(t) U_1(s)\adj \Bigg(
                \sum_{j=0}^{L+1} (-\iu)^j \avg*{T}{\Kop{}\of*{\Sop{j}}}
                + (-\iu)^{L+2} \Kop{s}\of*{\Sop{L+2}(s)}
            \Bigg) U_2(s) \psi \dd s
            ,
        \end{multlined}
    \end{align}
    which concludes our induction step.
\end{proof}

\section{Construction of effective Hamiltonians}\label{sec:Heff}

The results of \cref{sec:iterated-integration-by-parts} can be used to estimate the difference between the evolutions generated by two \emph{given} time-dependent Hamiltonians $H_1(t)$ and $H_2(t)$. Here we adopt a different perspective: starting from a single time-dependent, $T$-periodic Hamiltonian $H^{(T)}(t)$, we aim to \emph{construct} a second, time-independent effective Hamiltonian $\Heff{L}{T}$ whose associated dynamics provides, at a prescribed order $L$, an optimal approximation of the true evolution.

To this end, and recalling the discussion in \cref{sec:assumption-a}, we consider a reference Hamiltonian $H_0$ on a Hilbert space $\HH$, together with a unitary evolution $(U(t))_{t\in\R}$ and corresponding Hamiltonian $H(t)$.
Given $T>0$, we denote by $U^{(T)}(t)$ and $H^{(T)}(t)$ the associated $T$-periodic rescaled dynamics (cf.~\cref{eq:rescaling}). We seek to construct effective Hamiltonians of the form
\begin{equation}\label{eq:Heff-polynomial-Ansatz}
    \Heff{L}{T}
    = \sum_{l=0}^{L} T^l \Heffl{l},
\end{equation}
with the coefficients $\Heffl{l}$ to be chosen in such a way that the corresponding unitary evolution optimally approximates the one generated by $U^{(T)}$ in the high-frequency limit $T \to 0$.

We specialise \Cref{def:Kop,def:iterated-actions} to the current setting, where $H^{(T)}(t)$ takes the role of $H_2(t)$ and $\Heff{L}{T}$ the role of $H_1(t)$, and fix hereafter
\begin{gather}
    \label{eq:def-Kopeff}
    \Kopeff{t,L}{T}(A) = \Heff{L}{T} A - A H^{(T)}(t),
    \\
    \label{eq:def-Sopeff-0}
    \Sopeff{0,L}{T}(t) = \one,
    \\
    \label{eq:def-Sopeff}
    \Sopeff{j,L}{T}(t)
    = \int_0^t \Kopeff{s,L}{T}\of*{\Sopeff{j-1,L}{T}(s) }  \dd s - t \avg*{T}{\Kopeff{\cdot,L}{T}\of*{\Sopeff{j-1,L}{T}(\cdot) }}
    .
\end{gather}

\subsection{Polynomial operator families in \texorpdfstring{$T$}{T}}

We begin by introducing some nomenclature.

\begin{definition}\label{def:polynomial-degree}
    Let \((A_T)_{T\in \R}\) be a family of operators on $\HH$, and $K_0,K\in\N_0$ with $K_0\leq K$.
    We call \(A_T\) a \emph{polynomial in \(T\) of degree from \(K_0\) to \(K\)}, if there are operators \(\Tcoeff{A}{k}\) (which might all vanish) such that
    \begin{equation}\label{eq:polynomial-degree}
        A_T = \sum_{k=K_0}^{K} T^k \Tcoeff{A}{k}, \quad T\in\R
        .
    \end{equation}
    In other words, $A_T$ is a (possibly trivial) polynomial of degree $K$ whose coefficients up to degree $K_0-1$ vanish.

    For a polynomial \(A_T\) as in \cref{eq:polynomial-degree} and $K_0\leq K_0\primed\leq K\primed\leq K$ we denote the \emph{truncation to degree from \(K_0\primed\) to \(K\primed\)} by
    \begin{equation}
        \Ttrunc{A_T}{K_0\primed}{K\primed}
        = \sum_{k=K_0\primed}^{K\primed} T^k \Tcoeff{A}{k}.
    \end{equation}
\end{definition}

\begin{lemma}\label{thm:Sopeff-T-polynomial}
    Assume \((H(t))_{t\in\R}\) satisfies \cref{assump:A-H-cinf,assump:A-H-cont,assump:A-H-bound}.
    Let \(\Heff{L}{T}\) be any operator of the form \eqref{eq:Heff-polynomial-Ansatz}, with each \(\Heffl{l}\) satisfying \Cref{assump:A-H-cinf,assump:A-H-bound}. Let $\Sopeff{j,L}{T}$ be as defined in \cref{eq:def-Sopeff-0,eq:def-Sopeff}.
    Then, for every \(j\in\N_0\) and $t\in\R$, \(\Sopeff{j,L}{T}(t T)\) is a polynomial of degree from \(j\) to at most \(j(L+1)\):
    \begin{equation}
        \label{eq:Sopeff-T-polynomial}
        \Sopeff{j,L}{T}(t T)
        = \sum _{k = j} ^{j(L+1)} T^k \Tcoeff{ \Sopeff{j,L}{T}} {k} (t),
    \end{equation}
    with the (not necessarily nonzero) coefficients \(\Tcoeff{ \Sopeff{j,L}{T}} {k} (t)\) recursively given by
        \begin{align}
            \label{eq:Sopeff-T-coeffs-base}
        \Tcoeff{ \Sopeff{0,L}{T} }{0} (t)
        &= \one,
        \\
        \label{eq:Sopeff-T-coeffs-recursion}
        \Tcoeff{ \Sopeff{j+1,L}{T} }{k} (t)
        &=
        \begin{cases}
        \displaystyle
        \int _{0} ^{t}  \osc*{t\primed}{1}{\Kopeff{\cdot,0}{1}\of*{ \Tcoeff{\Sopeff{j,L}{T} }{j} (t\primed) }} \dd t\primed
        & k = j+1
        \\
        \\
        \displaystyle
            \int _{0} ^{t}  \Biggl[
                \osc*{t\primed}{1}{\Kopeff{\cdot,0}{1}\of*{ \Tcoeff{\Sopeff{j,L}{T} }{k-1} } }
                + \smashoperator{\sum_{(\sigma,\sigma') \in \mathcal{N}_{L,j,k}}} \osc*{t\primed}{1}{  \Heffl{\sigma'} \Tcoeff{ \Sopeff{j,L}{T} }{\sigma}}
            \Biggr]\dd t\primed
        & \substack{j+2 \leq k \\k\leq j(L+1) + 1}
        \\
        \\
        \displaystyle
        \int _{0} ^{t} \smashoperator[r]{\sum_{(\sigma,\sigma') \in \mathcal{N}_{L,j,k}} } \osc*{t\primed}{1}{  \Heffl{\sigma'} \Tcoeff{ \Sopeff{j,L}{T} }{\sigma}}
        \dd t\primed
        & \substack{j(L+1) + 2 \leq k \\k\leq (j+1)(L+1)}
        \\
        \\
        \displaystyle
        0
        & \mathrm{else}
        ,
        \end{cases}
    \end{align}
    where we use the shorthand notation
    \begin{align}
        \label{eq:def-Deltaeff}
        \osc*{s}{T}{A}
        &= A(s) - \avg*{T}{A}
        \\
        \label{eq:def-NLjk}
        \mathcal{N}_{L,j,k}
        &= \Set*{(\sigma,\sigma')\in\N_0\times\N_0 \given \substack{j \leq \sigma \leq j(L+1),\\ 1 \leq \sigma' \leq L,\\ \sigma+\sigma'+1 = k}}
        .
    \end{align}
    Moreover, for every \(m \in \N_0\) and every compact interval \(I\) containing $0$ there exist \(\tilde{k}_{m,j,L,k}\in \N_0\) and non-negative numbers \(\tilde{a}_{m,j,L,k}, \tilde{b}_{m,j,L,k}\) such that
    \begin{equation}\label{eq:Sopeff-T-coeffs-bound}
        \norm*{H_0^{m} \Tcoeff{ \Sopeff{j,L}{T}} {k} (t) \psi}
        \leq \tilde{a}_{m,j,L,k}\norm*{H_0^{m+\tilde{k}_{m,j,L,k} } \psi} + \tilde{b}_{m,j,L,k} \norm{\psi}, \quad t\in I
        .
    \end{equation}
    Finally, if \Cref{assump:A-H-period} holds, i.e. $H^{(T)}(t)$ is $T$-periodic, the coefficients are $1$-periodic:
    \begin{equation}\label{eq:Sopeff-T-coeffs-periodic}
        \Tcoeff{ \Sopeff{j,L}{T}} {k} (t+1)
        = \Tcoeff{ \Sopeff{j,L}{T}} {k} (t).
    \end{equation}
\end{lemma}
\begin{proof}
    We prove the claim by induction over \(j\).

    We begin by proving \cref{eq:Sopeff-T-polynomial,eq:Sopeff-T-coeffs-base,eq:Sopeff-T-coeffs-recursion}.
    By definition (cf. \cref{eq:def-Sopeff}), \(\Sopeff{0,L}{T}(t T) = \one\) and these properties trivially hold for \(j=0\). Assume the claim holds for some \(j \in \N_0\), i.e. assume \cref{eq:Sopeff-T-polynomial,eq:Sopeff-T-coeffs-base,eq:Sopeff-T-coeffs-recursion}.
Substituting the polynomial expansions for \(\Heff{L}{T}\) and
\(\Sopeff{j,L}{T}(tT)\), rescaling the integration variable, and collecting equal powers of \(T\), we obtain
\begin{align}
        \MoveEqLeft
        \int _{0} ^{t T} \Kopeff{t\primed,L}{T}\of*{ \Sopeff{j,L}{T}(t\primed) }  \dd t\primed
        =
        \nonumber\\
        &= \int _{0} ^{t T} \Big(\Heff{L}{T} \Sopeff{j,L}{T}(t\primed) - \Sopeff{j,L}{T}(t\primed) H^{(T)}(t\primed)\Big) \dd t\primed
        \nonumber\\
        &= T \int _{0} ^{t} \Big(\Heff{L}{T} \Sopeff{j,L}{T}(s\primed T) - \Sopeff{j,L}{T}(s\primed T) H^{(T)}(s\primed T)\Big) \dd s\primed
        \nonumber\\
        &= \sum _{\tilde k = j} ^{j(L+1)} T^{\tilde k+1} \int _{0} ^{t} \Big(\Heff{L}{T} \Tcoeff{ \Sopeff{j,L}{T}} {\tilde k} (s\primed) - \Tcoeff{ \Sopeff{j,L}{T}} {\tilde k} (s\primed) H(s\primed)\Big) \dd s\primed
        \nonumber\\
        &
        = \sum _{\tilde k = j} ^{j(L+1)} \int _{0} ^{t} \Biggl[
            \sum_{l=1}^{L} T^{\tilde k+1+l} \Heffl{l} \Tcoeff{ \Sopeff{j,L}{T}} {\tilde k} (s\primed)
            + T^{\tilde k+1} \underbrace{\of*{\Heffl{0}\Tcoeff{ \Sopeff{j,L}{T}} {\tilde k} (s\primed) - \Tcoeff{ \Sopeff{j,L}{T}} {\tilde k} (s\primed) H(s\primed)}}_{ =\Kopeff{s\primed,0}{1}\of*{ \Tcoeff{ \Sopeff{j,L}{T}} {\tilde k} (s\primed) } }
        \Biggr] \dd s\primed
        \nonumber\\
        &= \int _{0} ^{t} \eof*{
                \sum_{k = j+2}^{(j+1)(L+1)} \sum _{(\sigma,\sigma')\in \mathcal{N}_{L,j,k}} T^{k} \Heffl{\sigma'} \Tcoeff{ \Sopeff{j,L}{T}} {\sigma} (s\primed)
                + \sum _{k = j+1} ^{j(L+1)+1} T^{k} \Kopeff{s\primed,0}{1}\of*{ \Tcoeff{ \Sopeff{j,L}{T}} {k-1} (s\primed) }
            }\dd s\primed
        \nonumber\\
        \label{eq:Sopeff-T-coeffs-recursion-proof-1}
        &=
        \begin{aligned}[t]
            \int_{0}^{t}
            \Bigg[
            &T^{j+1} \Kopeff{s\primed,0}{1}\of*{ \Tcoeff{ \Sopeff{j,L}{T}} {j} (s\primed) }
            \\
            &+ \sum_{k = j+2}^{j(L+1)+1} T^{k}\eof*{\Kopeff{s\primed,0}{1}\of*{ \Tcoeff{ \Sopeff{j,L}{T}} {k-1} (s\primed) } + \sum _{(\sigma,\sigma')\in \mathcal{N}_{L,j,k}} \Heffl{\sigma'} \Tcoeff{ \Sopeff{j,L}{T}} {\sigma} (s\primed)}
            \\
            &+ \sum_{k=j(L+1)+2}^{(j+1)(L+1)} T^{k} \sum _{(\sigma,\sigma')\in \mathcal{N}_{L,j,k}} \Heffl{\sigma'} \Tcoeff{ \Sopeff{j,L}{T}} {\sigma} (s\primed)
            \Bigg]\dd s\primed
            ,
        \end{aligned}
    \end{align}
    where in the 6th line we made the substitution $\sigma=\tilde k$, $\sigma'=l$ and $k=\sigma+\sigma'+1$ and used the notation $\mathcal{N}_{L,j,k}$ introduced in \cref{eq:def-NLjk}. We see that this is a polynomial of degree from $j+1$ to at most $(j+1)(L+1)$. Combining this with
    \begin{equation}
        \Sopeff{j+1,L}{T}(t T)
        = \int _{0} ^{t T} \osc*{t\primed}{T}{\Kopeff{\cdot,L}{T}\of*{ \Sopeff{j,L}{T}(t\primed) }}  \dd t\primed
        = \int _{0} ^{t T} \Kopeff{t\primed,L}{T}\of*{ \Sopeff{j,L}{T}(t\primed) }  \dd t\primed
        - t\int _{0} ^{T} \Kopeff{r\primed,L}{T}\of*{ \Sopeff{j,L}{T}(r\primed) }  \dd r\primed
    \end{equation}
    and reading off the coefficients proves \cref{eq:Sopeff-T-polynomial,eq:Sopeff-T-coeffs-base,eq:Sopeff-T-coeffs-recursion} for $j+1$.

Next we show \cref{eq:Sopeff-T-coeffs-bound}.
    Let \(m \in \N_0\) and \(I\) be an arbitrary compact interval containing $0$.
    For $j=0$, since \(\Sopeff{0,L}{T}(t T)\) is bounded and \(t\)-independent, \(H_0^m\Sopeff{0,L}{T}(t T)\) is \(H_0^m\)-bounded uniformly in \(t\).
    Assume all \(H_0^m \Tcoeff{\Sopeff{j,L}{T}}{k}(t)\) to be relatively bounded with respect to some \(H_0^{m+\tilde{k}_{m,j,L,k}}\) uniformly for \(t \in I\) for some $j\in\N_0$.
    By \Cref{assump:A-H-bound} we know that \(H_0^m\Heffl{k}\) is \(H_0^{m+k_m}\)-bounded, uniformly in \(t \in I\).
    Since the recursion involves only finite sums, compositions with the
\(H_0^{m+k_m}\)-bounded operators \(\Heffl{l}\), averaging operations,
and integration over compact intervals, the same type of relative bounds
is preserved at each induction step. Hence it follows from \cref{eq:Sopeff-T-coeffs-recursion} that the coefficients of \(\Sopeff{j+1,L}{T}(t T)\) also satisfy such bounds with certain $\tilde{k}_{m,j+1,L,k}$, uniformly in \(t\in I\).

Finally we show \cref{eq:Sopeff-T-coeffs-periodic}, i.e., the coefficients \(\Tcoeff{\Sopeff{j,L}{T}}{k}(t)\) are \(1\)-periodic if \(H^{(T)}(t)\) is \(T\)-periodic.
    For \(j=0\), \(\Sopeff{0,L}{T}\) clearly satisfies this.
    Assume that the coefficients of \(\Sopeff{j,L}{T}\) are 1-periodic for some \(j\in \N_0\).
    Then \cref{eq:Sopeff-T-coeffs-recursion} shows that we can write
    \begin{equation}
        \Tcoeff{\Sopeff{j+1,L}{T}}{k}(t) =\int_0^{t} \osc*{s}{1}{A_j}\dd s
        , \quad \forall t\in\R
    \end{equation}
    where $A_j(s)$ is a 1-periodic function because \(H(t)\) and, by the induction hypothesis, the coefficients \(\Tcoeff{\Sopeff{j,L}{T}}{k}\) are \(1\)-periodic. Thus
    \begin{equation}
        \Tcoeff{\Sopeff{j+1,L}{T}}{k}(t+1)
        = \Tcoeff{\Sopeff{j+1,L}{T}}{k}(t) + \int_t^{t+1} \osc*{s}{1}{A_j}\dd s
        = \Tcoeff{\Sopeff{j+1,L}{T}}{k}(t), \quad \forall t\in\R.
    \end{equation}
This proves the claimed periodicity for $j+1$ and completes the proof.
\end{proof}
The previous lemma shows that the recursive construction of the operators
\(\Sopeff{j,L}{T}\) preserves polynomiality in \(T\), and that the order in \(T\) increases by at least one at each iteration.
\begin{corollary}\label{thm:AKS-T-polynomial}
    Under the same assumptions and notation as in \cref{thm:Sopeff-T-polynomial}, the operator $\avg*{T}{\Kopeff{\cdot,L}{T}\of*{\Sopeff{j,L}{T}} }$
    is a polynomial of degree from \(j\) to \(j(L+1)+L\) in \(T\),
    \begin{equation}\label{eq:AKS-T-polynomial-1}
        \avg*{T}{\Kopeff{\cdot,L}{T}\of*{\Sopeff{j,L}{T}} }
        = \sum_{k = j}^{j(L+1)+L} T^k \Tcoeff{ \avg*{T}{\Kopeff{\cdot,L}{T}\of*{\Sopeff{j,L}{T}} } }{k},
    \end{equation}
     with coefficients
    \begin{equation}\label{eq:AKS-T-polynomial-2}
        \Tcoeff{ \avg*{T}{\Kopeff{\cdot,L}{T}\of*{\Sopeff{j,L}{T}} } }{k}
        = \begin{cases}
            \displaystyle
            \avg*{1}{ \Kopeff{\cdot,0}{1}\of*{ \Tcoeff{\Sopeff{j,L}{T} }{j} } }
            & k = j
            \\
            \displaystyle
            \avg*{1}{ \Kopeff{\cdot,0}{1}\of*{ \Tcoeff{\Sopeff{j,L}{T} }{k} } + \sum_{(\sigma,\sigma') \in \mathcal{N}_{L,j,k+1}} \Heffl{\sigma'}\Tcoeff{ \Sopeff{j,L}{T} }{\sigma} }
            & \substack{j+1 \leq k \\ k\leq j(L+1)}
            \\
            \displaystyle
            \avg*{1}{\sum_{(\sigma,\sigma') \in \mathcal{N}_{L,j,k+1}} \Heffl{\sigma'}\Tcoeff{ \Sopeff{j,L}{T} }{\sigma} }
            & \substack{ j(L+1) + 1 \leq k \\ k \leq j(L+1) + L }
            \\
            \displaystyle
            0
            & \mathrm{else}
            .
        \end{cases}
    \end{equation}
\end{corollary}

\begin{proof}
    It follows by applying \cref{eq:Sopeff-T-coeffs-recursion-proof-1} with \(t = 1\):
    \begin{align}
        \avg*{T}{\Kopeff{\cdot,L}{T}\of*{\Sopeff{j,L}{T}} }
        &= \frac{1}{T} \int _{0} ^{T} \Kopeff{t\primed,L}{T}\of*{ \Sopeff{j,L}{T}(t\primed) }  \dd t\primed
        \nonumber\\
        &=
        \begin{aligned}[t]
           \frac{1}{T} \int_{0}^{1}
            &\Bigg[
            T^{j+1} \Kopeff{s\primed,0}{1}\of*{ \Tcoeff{ \Sopeff{j,L}{T}} {j} (s\primed) }
            \\
            &+ \sum_{k = j+2}^{j(L+1)+1} T^{k}\eof*{\Kopeff{s\primed,0}{1}\of*{ \Tcoeff{ \Sopeff{j,L}{T}} {k-1} (s\primed) } + \sum _{(\sigma,\sigma') \in \mathcal{N}_{L,j,k}} \Heffl{\sigma'} \Tcoeff{ \Sopeff{j,L}{T}} {\sigma} (s\primed)}
            \\
            &+ \sum_{k=j(L+1)+2}^{(j+1)(L+1)} T^{k} \sum _{(\sigma,\sigma') \in \mathcal{N}_{L,j,k}} \Heffl{\sigma'} \Tcoeff{ \Sopeff{j,L}{T}} {\sigma} (s\primed)
            \Bigg]\dd s\primed
            ,
        \end{aligned}
        \nonumber\\
        &=
        \begin{aligned}[t]
           \int_{0}^{1}
            &\Bigg[
            T^{j} \Kopeff{s\primed,0}{1}\of*{ \Tcoeff{ \Sopeff{j,L}{T}} {j} (s\primed) }
            \\
            &+ \sum_{k = j+2}^{j(L+1)+1} T^{k - 1}\eof*{\Kopeff{s\primed,0}{1}\of*{ \Tcoeff{ \Sopeff{j,L}{T}} {k-1} (s\primed) } + \sum _{(\sigma,\sigma') \in \mathcal{N}_{L,j,k}} \Heffl{\sigma'} \Tcoeff{ \Sopeff{j,L}{T}} {\sigma} (s\primed)}
            \\
            &+ \sum_{k=j(L+1)+2}^{(j+1)(L+1)} T^{k - 1} \sum _{(\sigma,\sigma') \in \mathcal{N}_{L,j,k}} \Heffl{\sigma'} \Tcoeff{ \Sopeff{j,L}{T}} {\sigma} (s\primed)
            \Bigg]\dd s\primed
            ,
        \end{aligned}
        \nonumber\\
        &=
        \begin{aligned}[t]
           \int_{0}^{1}
            &\Bigg[
            T^{j} \Kopeff{s\primed,0}{1}\of*{ \Tcoeff{ \Sopeff{j,L}{T}} {j} (s\primed) }
            \\
            &+ \sum_{k = j+1}^{j(L+1)} T^{k}\eof*{\Kopeff{s\primed,0}{1}\of*{ \Tcoeff{ \Sopeff{j,L}{T} } {k} (s\primed) } + \sum _{(\sigma,\sigma') \in \mathcal{N}_{L,j,k+1}} \Heffl{\sigma'} \Tcoeff{ \Sopeff{j,L}{T}} {\sigma} (s\primed)}
            \\
            &+ \sum_{k=j(L+1)+1}^{j(L+1) + L} T^{k} \sum _{(\sigma,\sigma') \in \mathcal{N}_{L,j,k+1}} \Heffl{\sigma'} \Tcoeff{ \Sopeff{j,L}{T}} {\sigma} (s\primed)
            \Bigg]\dd s\primed.
        \end{aligned}
    \end{align}
\end{proof}

\subsection{Definition of the effective Hamiltonians\texorpdfstring{ \(\Heff{L}{T}\)}{ }}\label{sec:heff_construction}
Up to this point, the coefficients $\Heffl{l}$ in the ansatz \eqref{eq:Heff-polynomial-Ansatz} for the effective Hamiltonian have been left unspecified.
We now make a concrete choice based on the preceding results:
\begin{definition}\label{def:Heff-definition}
	For every \(L \in \N_0\), the operator $\Heff{L}{T}$ is given by $\domof{\Heff{L}{T}}=\cinftyofh$ and
	\begin{equation}\label{eq:Heff-definition-LT}
		\Heff{L}{T}
		= \Tcoeff{ \avg*{1}{H} - \sum_{j= 1}^{L} (-\iu )^j \avg*{T}{ \Kopeff{\cdot,L-j}{T}\of*{ \Sopeff{j,L-j}{T} } } }{0;L}
		= \sum_{l=0}^{L} T^l \Heffl{l}
		,
	\end{equation}
	where
	\begin{align}
	    \label{eq:Heff-definition-0}
		\Heffl{0}
		&= \avg*{1}{H}
		\\
        \label{eq:Heff-definition-l}
		\Heffl{l}
		&= - \sum_{j=1}^l (-\iu)^j \Tcoeff{ \avg*{T}{\Kopeff{\cdot,L-j}{T}\of*{\Sopeff{j,L-j}{T} } } }{l}, \quad l\geq 1
		.
	\end{align}
\end{definition}

This means each \(\Heff{L}{T}\) is defined recursively by \(\Heff{0}{T}, \dots, \Heff{L-1}{T}\). Note that, at this stage, it is not a priori clear that $\Heff{L}{T}$ is a symmetric operator; this property will be addressed later on.

In order for this definition to be well-posed, we must verify that the operator appearing in square brackets in \cref{eq:Heff-definition-LT} is indeed a polynomial in \(T\). Also, the RHS of \cref{eq:Heff-definition-l} might depend on \(L\), in which case adding an additional order of approximation could also change the previous lower order terms. In the following proof, we see that this is not the case. Owing to the recursive nature of the construction, this verification must itself proceed recursively. For this reason, we first state the definition and only afterwards establish its well-posedness.

\begin{lemma}\label{thm:Heff-definition-well-posed}
    Assume \((H(t))_{t\in\R}\) satisfies \cref{assump:A-H-cinf,assump:A-H-cont,assump:A-H-bound,assump:A-H-period}.
    For every \(L \in \N_0\), the operator
    \begin{equation}
        \avg*{1}{H} - \sum_{j= 1}^{L} (-\iu )^j \avg*{T}{ \Kopeff{\cdot,L-j}{T}\of*{ \Sopeff{j,L-j}{T} } }
    \end{equation}
    is a polynomial in \(T\).
    Furthermore, each coefficient \(\Heffl{l}\), where $l=0,\ldots,L$, as in \Cref{def:Heff-definition} (cf. \cref{eq:Heff-definition-LT}) satisfies \Cref{assump:A-H-cinf,assump:A-H-bound}.
\end{lemma}
\begin{proof}
	Proof by induction over \(L \in \N_0\).

	For $L=0$, \(\Heff{0}{T} = \avg*{1}{H} = \avg*{T}{H^{(T)}}\) is a well-defined operator on \(\cinftyofh\), as \(s\mapsto H(s)\) is strongly continuous and hence integrable on compact intervals.
	In addition, being independent of \(T\), it is a polynomial of degree \(0\).
	Lastly, we have
	\begin{equation}
	    \norm*{ H_0^m \avg*{1}{H} \psi}
		\leq \int_{0}^{1} \norm*{H_0^m H(s) \psi} \dd s
		\leq a_m \norm*{H_0^{m+k_m}\psi } + b_m \norm*{\psi}
	\end{equation}
	for every \(\psi \in \cinftyofh\) and \(m \in \N_0 \).
	So \(\Heffl{0}\) satisfies \Cref{assump:A-H-bound}.

	Now, assume \(L \in \N\) and the claim holds for \(0, \dots, L-1\), so we want show that it also holds for $L$. We have that
	\begin{equation}
		\avg*{1}{H} - \sum_{j= 1}^{L} (-\iu )^j \avg*{T}{ \Kopeff{\cdot,L-j}{T}\of*{ \Sopeff{j,L-j}{T} } }
	\end{equation}
	is a polynomial in \(T\) since, by \Cref{thm:AKS-T-polynomial}, so is each summand. Likewise, we see that this polynomial depends on $\Heffl{l}$ for $l=0,\ldots,L-1$ only and that it satisfies \Cref{assump:A-H-cinf,assump:A-H-bound}. Furthermore, its coefficients are given by
	\begin{align}
	    \Tcoeff{\avg*{1}{H} - \sum_{j= 1}^{L} (-\iu )^j \avg*{T}{ \Kopeff{\cdot,L-j}{T}\of*{ \Sopeff{j,L-j}{T} } }}{0}
		&= \avg*{1}{H}
		\nonumber\\
        \Tcoeff{\avg*{1}{H} - \sum_{j= 1}^{L} (-\iu )^j \avg*{T}{ \Kopeff{\cdot,L-j}{T}\of*{ \Sopeff{j,L-j}{T} } }}{l}
		&= - \sum_{j= 1}^{L} (-\iu )^j \Tcoeff{\avg*{T}{ \Kopeff{\cdot,L-j}{T}\of*{ \Sopeff{j,L-j}{T} } } }{l}
		\nonumber\\
		&=- \sum_{j= 1}^{l} (-\iu )^j \Tcoeff{\avg*{T}{ \Kopeff{\cdot,L-j}{T}\of*{ \Sopeff{j,L-j}{T} } } }{l}
		,\quad 1\leq l \leq L
        ,
	\end{align}
	where we used that \(\Sopeff{j,L}{T}\) is a polynomial of degree from \(j\), meaning that, if \(j > l\), then we have  \(\Tcoeff{\avg*{T}{ \Kopeff{\cdot,L-j}{T}\of*{ \Sopeff{j,L-j}{T} } } }{l} = 0\), and hence, in particular, that we can truncate the sum over \(j\) to \(1, \dots, l\leq L\). This shows that \(\Heffl{l}\) as in \Cref{def:Heff-definition} satisfies \Cref{assump:A-H-cinf,assump:A-H-bound}, for all $l=0,\ldots,L$, which concludes the induction step and hence completes the proof.
\end{proof}

We now motivate our specific construction of \(\Heff{L}{T}\) in \Cref{def:Heff-definition} by showing that it is uniquely characterised by a natural asymptotic property in the high-frequency limit \(T \to 0\) for periodic Hamiltonians \(H(t)\). To this end, we first introduce some standard notation:

\begin{definition}\label{def:big-O}
    Consider two functions \(\psi\colon \R \to \HH\) and \(\phi\colon \R \to \R\).
    We write \(\psi = \bigo(\phi)\) as \(T \to T_0\) if there are \(\delta, c >0\) such that
    \begin{equation}
        \norm{\psi(T)} \leq c |\phi(T)|
        ,\quad \forall T\in \R \colon 0 < \abs{T - T_0} < \delta
        .
    \end{equation}
Furthermore, given \(\psi\colon X\times Y \subseteq \R\times\R \to \HH\), we write \(\psi(\cdot, y) = \bigo(\phi)\) \emph{uniformly in \(y\)} as \(T \to T_0\) if there are \(c,\delta >0 \) such that, for every \(y\in Y\),
    \begin{equation}
        \norm{\psi(T,y)} \leq c |\phi(T)|
        ,\quad \forall T\in \R \colon 0 < \abs{T - T_0} < \delta
        .
    \end{equation}
\end{definition}

The following results connect the asymptotic expansion of the propagator with our polynomial ansatz for the effective Hamiltonian. After deriving the asymptotics of the transformed propagator coefficients in \cref{thm:Sopeff-asymptotics}, \cref{thm:asymptotic-relation-Sop} and \cref{thm:asymptotic-relation-Kop} establish recursive relations between these coefficients and the operators defining the effective Hamiltonian. This correspondence will enable us to justify our construction of effective Hamiltonians in \cref{thm:Heff-condition}.

\begin{lemma}\label{thm:Sopeff-asymptotics}
    Assume \((H(t))_{t\in\R}\) satisfies \cref{assump:A-H-cinf,assump:A-H-cont,assump:A-H-bound,assump:A-H-period}.
    Then, for every  \(j,L \in \N_0\), and \(\psi \in \cinftyofh\),  the function \(t\mapsto \Sopeff{j,L}{T}(t)\psi\) is \(T\)-periodic and satisfies
    \begin{equation}\label{eq:asymptotics-iterated-actions}
        \norm*{\Sopeff{j,L}{T}(t)\psi }
        \leq
        \sum_{k=j}^{j(L+1)} T^k \of*{\tilde{a}_{0,j,L,k} \norm*{H_0^{\tilde{k}_{0,j,L,k}} \psi}
        + \tilde{b}_{0,j,L,k} \norm*{\psi}}
        = \bigo\of{T^j}
        \quad \text{as}\ T \to 0
    \end{equation}
    uniformly for $t$ in compact intervals, with $\tilde{k}_{0,j,L,k},\tilde{a}_{0,j,L,k},\tilde{b}_{0,j,L,k}$ as in \Cref{thm:Sopeff-T-polynomial}.
    \end{lemma}

\begin{proof}
This follows immediately from \Cref{thm:Sopeff-T-polynomial}, since
    \begin{equation}
        \Sopeff{j,L}{T}(t)
        = \Sopeff{j,L}{T}\of*{\textstyle\frac{t}{T}T}
        = \sum_{k=j}^{j(L+1)} T^k \Tcoeff{\Sopeff{j,L}{T}}{k} \of*{\textstyle\frac{t}{T} }
    .\qedhere
    \end{equation}
\end{proof}

\begin{lemma}\label{thm:asymptotic-relation-Sop}
    Assume \((H(t))_{t\in\R}\) satisfies \cref{assump:A-H-cinf,assump:A-H-cont,assump:A-H-bound,assump:A-H-period}.
    Then, for every \(L\in\N\), \(j\in \N_0\) and \(\psi \in \cinftyofh\) we have
    \begin{equation}\label{eq:asymptotic-relation-Sop}
        \of*{\Sopeff{j,L}{T}(s) - \Sopeff{j,L-1}{T}(s)}\psi
        = \bigo\of*{T^{L+j}}
    \end{equation}
    uniformly for \(s\) in compact intervals as \(T \to 0\).
\end{lemma}
\begin{proof}
    Let \(\psi \in \cinftyofh\).
    Again, we prove the statement by induction over \(j\in\N_0\).

    In the case \(j=0\) the statement is trivially true as \(\Sopeff{0,L}{T}(s) = \one = \Sopeff{0,L-1}{T}(s)\). Now, assume \cref{eq:asymptotic-relation-Sop} holds for some \(j \in \N_0\) and let \(t\in \R\).
    Then
    \begin{align}
        \psi(sT)
        &\coloneqq \eof[\Big]{\Kopeff{sT,L}{T}\of*{\Sopeff{j,L}{T}(sT)} - \Kopeff{sT,L-1}{T}\of*{\Sopeff{j,L-1}{T}(sT)}}\psi
        \nonumber\\
        &= \eof[\Big]{ \Heff{L}{T}\Sopeff{j,L}{T}(sT) - \Sopeff{j,L}{T}(sT)H^{(T)}(sT) - \of*{ \Heff{L-1}{T}\Sopeff{j,L-1}{T}(sT) - \Sopeff{j,L-1}{T}(sT)H^{(T)}(sT) } }\psi
        \nonumber\\
        &=
        \begin{aligned}[t]
            &\underbrace{\vphantom{\of*{\Sopeff{j,L}{T}(sT) - \Sopeff{j,L-1}{T}(sT)}}\Heff{L-1}{T}}_{=\bigo\of*{1}} \underbrace{\of*{\Sopeff{j,L}{T}(sT) - \Sopeff{j,L-1}{T}(sT)}}_{=\bigo\of*{T^{L+j}}}\psi
            - \underbrace{\of*{\Sopeff{j,L}{T}(sT) - \Sopeff{j,L-1}{T}(sT)}}_{=\bigo\of*{T^{L+j}}} \underbrace{\vphantom{\of*{\Sopeff{j,L}{T}(sT) - \Sopeff{j,L-1}{T}(sT)}}H(s)}_{=\bigo\of*{1}} \psi
            \nonumber\\
            &+
            T^L \Heffl{L}
            \underbrace{\vphantom{\of*{\Sopeff{j,L}{T}(sT) - \Sopeff{j,L-1}{T}(sT)}}
                \Sopeff{j,L}{T}(sT)\psi
            }_{=\bigo\of*{T^j}}
        \end{aligned}
        \nonumber\\
        &= \bigo\of*{T^{L+j}}
        ,
    \end{align}
    uniformly for \(s\in [0,t]\) as \(T\to 0\) by \cref{thm:Sopeff-asymptotics}, where we used the fact that all operators in the expression above are polynomial in \(T\) and thus their asymptotic behaviour as \(T \to 0\) is given by the lowest order in each of the products, which is multiplicative.
    We indicate the latter by writing \(A_T = \bigo\of*{T^k}\), if \(A_T\) is a polynomial in \(T\) of order from at least \(k\).
    Due to the uniformity in \(s\), \(\avg*{T}{\psi} = \bigo\of*{T^{L+j}}\) and hence
    \begin{equation}
        \psi_{\mathrm{osc}}(sT) \coloneqq  \psi(sT)-\avg*{T}{\psi} = \bigo\of*{T^{L+j}},
    \end{equation}
    where the latter is also uniform for \(s\) in compact intervals.
    Therefore, since by construction
    \begin{equation}
        \int_0^T \psi_{\mathrm{osc}}(s) \dd s = 0,
    \end{equation}
    we get
    \begin{align}
        \norm*{\of*{\Sopeff{j+1,L}{T}(t) - \Sopeff{j+1,L-1}{T}(t)}\psi}
        &= \norm[\bigg]{\int_0^t \psi_{\mathrm{osc}}(s) \dd{s}}
        \nonumber\\
        &\leq \norm[\bigg]{\floorof*{\textstyle \frac{t}{T}}\displaystyle \int_0^T \psi_{\mathrm{osc}}(s) \dd{s} }
        + \norm[\bigg]{ \int_{\floorof*{ \textfrac{t}{T}} T }^t \psi_{\mathrm{osc}}(s) \dd{s} }
        \nonumber\\
        &\leq  T \int_{\floorof*{\textstyle \frac{t}{T}}  }^{\textstyle \frac{t}{T}} \norm*{ \psi_{\mathrm{osc}}(r T)} \dd{r}
        \nonumber\\
        &\leq T\left(\frac{t}{T}-\left \lfloor\frac{t}{T}\right\rfloor\right) c T^{L+j}
        \nonumber\\
        &= \bigo\of*{T^{L+(j+1)}}       ,
    \end{align}
    thus completing the proof.
\end{proof}

\begin{corollary}\label{thm:asymptotic-relation-Kop}
    Assume \(\of{H(t)}_{t\in\R}\) satisfies \Cref{assump:A-H-cinf,assump:A-H-cont,assump:A-H-bound,assump:A-H-period}.
    For every \(j\in \N_0\), \(L \in \N\), and \(\psi \in \cinftyofh\), we have
    \begin{equation}
        \label{eq:asymptotic-relation-Kop}
        \avg*{T}{\Kopeff{\cdot,L}{T}\of*{\Sopeff{j,L}{T}} - \Kopeff{\cdot,L-1}{T}\of*{\Sopeff{j,L-1}{T}} }\psi
        = \bigo\of*{T^{L+j} }
    \end{equation}
     as \(T \to 0\).
\end{corollary}

\begin{proof}
    Let \(\psi \in \cinftyofh\).
    We have
    \begin{align}
        \MoveEqLeft\nonumber
        \avg*{T}{\Kopeff{\cdot,L}{T}\of*{\Sopeff{j,L}{T}} - \Kopeff{\cdot,L-1}{T}\of*{\Sopeff{j,L-1}{T}} }\psi
        \nonumber\\
        &= \frac{1}{T}\int_{0}^T \of*{\Kopeff{t\primed,L}{T}(\Sopeff{j,L}{T}(t\primed)) - \Kopeff{t\primed,L-1}{T}(\Sopeff{j,L-1}{T}(t\primed))}\psi \dd t\primed
        \nonumber\\
        &= \int_{0}^1 \of*{\Kopeff{sT,L}{T}(\Sopeff{j,L}{T}(sT)) - \Kopeff{sT,L-1}{T}(\Sopeff{j,L-1}{T}(sT))}\psi \dd s
        \nonumber\\
        &= \int_{0}^1 \biggl[\Heff{L}{T}\Sopeff{j,L}{T}(sT) - \Heff{L-1}{T}\Sopeff{j,L-1}{T}(sT) \psi
        - \of*{ \Sopeff{j,L}{T}(sT) - \Sopeff{j,L-1}{T}(sT)}H^{(T)}(sT) \psi \biggr]\dd s
        \nonumber\\
        &=
        \begin{aligned}[t]
            \int_{0}^1\biggl[
            &\underbrace{\Heff{L-1}{T}}_{=\bigo(1)}\underbrace{\of*{ \Sopeff{j,L}{T}(sT) -\Sopeff{j,L-1}{T}(sT)}}_{=\bigo\of*{T^{L+j}}}\psi
            - \underbrace{\of*{ \Sopeff{j,L}{T}(sT) - \Sopeff{j,L-1}{T}(sT)}}_{=\bigo\of{T^{L+j}}}\underbrace{H^{(T)}(sT)}_{=\bigo\of*{1}} \psi
            \nonumber\\
            &+ T^L \Heffl{L}\underbrace{\Sopeff{j,L}{T}(sT) }_{=\bigo\of*{T^j}}\psi
            \biggr]\dd s
        \end{aligned}
        \nonumber\\
        & = \bigo\of*{T^{L+j} }
        ,
    \end{align}
    where again we used the fact that all operators in the expression above are polynomials in \(T\), and applied \Cref{thm:asymptotic-relation-Sop} in the last step.
\end{proof}

With the preparatory results in place, we now turn to the characterization of $\Heff{L}{T}$:

\begin{proposition}\label{thm:Heff-condition}
	Assume \((H(t))_{t\in\R}\) satisfies \Cref{assump:A-H-cinf,assump:A-H-cont,assump:A-H-bound,assump:A-H-period}, and let \(\Heff{L}{T}\) be given as in \Cref{def:Heff-definition}.
	Then, for every \(\psi \in \cinftyofh\) and \(L\in \N_0\),
	\begin{equation}\label{eq:Heff-condition}
	    \sum_{j=0}^{L} \of*{-\iu}^j \avg*{T}{\Kopeff{\cdot,L}{T}\of{\Sopeff{j,L}{T}}} \psi
        = \bigo\of*{T^{L+1}}
        \quad \text{as}\ T \to 0.
	\end{equation}
Moreover, $\Heff{L}{T}$ is the \emph{unique} polynomial in $T$ of degree $L$ with coefficients satisfying \Cref{assump:A-H-cinf,assump:A-H-bound} that fulfils this property.
\end{proposition}
\begin{proof}
Consider an operator of the form
\begin{equation}
	\hat{H}_{\mathrm{eff},L}^{(T)} = \sum_{l=0}^{L} T^l \hat{H}_{\mathrm{eff}}^{[l]},
\end{equation}
with some unspecified coefficients \(\hat{H}_{\mathrm{eff}}^{[l]}\) satisfying \Cref{assump:A-H-cinf,assump:A-H-bound}, and let \(\psi \in \cinftyofh\). Let $\Kopeff{\cdot,L}{T}$ and $\Sopeff{j,L}{T}$ be constructed as in \cref{eq:def-Kopeff,eq:def-Sopeff,eq:def-Sopeff-0} but with $\hat{H}_{\mathrm{eff},L}^{(T)}$ in place of $\Heff{L}{T}$. Applying \cref{eq:asymptotic-relation-Kop} iteratively, we obtain in case $L>0$:
    \begin{align}
        \avg*{T}{\Kopeff{\cdot,L}{T}\of{\Sopeff{j,L}{T}}}\psi
        &= \avg*{T}{\Kopeff{\cdot,L-1}{T}\of{\Sopeff{j,L-1}{T}}}\psi + \bigo\of{T^{L+j}}
        \nonumber\\
        &= \avg*{T}{\Kopeff{\cdot,L-2}{T}\of{\Sopeff{j,L-2}{T}}}\psi + \bigo\of{T^{(L-1)+j}} + \bigo\of{T^{L+j}}
        \nonumber\\
        &= \avg*{T}{\Kopeff{\cdot,L-j}{T}\of{\Sopeff{j,L-j}{T}}}\psi + \bigo\of{T^{(L-(j-1))+j}} + \dots + \bigo\of{T^{L+j}}
        \nonumber\\
        &= \avg*{T}{\Kopeff{\cdot,L-j}{T}\of{\Sopeff{j,L-j}{T}}}\psi + \bigo\of{T^{L+1}}
        .
    \end{align}
    This gives us the following chain of equivalences:
    \begin{align}
        &&
        \bigo\of{T^{L+1}}
        &=\sum_{j=0}^{L} \of*{-\iu}^j \avg*{T}{ \Kopeff{\cdot,L}{T}\of{ \Sopeff{j,L}{T}} } \psi
        &
        \nonumber\\
        \iff
        &&
        \bigo\of{T^{L+1}}
        &=\hat{H}_{\mathrm{eff},L}^{(T)}\psi - \avg*{1}{H}\psi
        +\sum_{j=1}^{L} \of*{-\iu}^j \avg*{T}{ \Kopeff{\cdot,L-j}{T}\of{ \Sopeff{j,L-j}{T}} } \psi
        &
        \nonumber\\
        &&
        &=\hat{H}_{\mathrm{eff},L}^{(T)}\psi - \avg*{1}{H}\psi
        + \sum_{l=1}^{L(L+1)+L} T^l \Tcoeff{\sum_{j=1}^{L} \of*{-\iu}^j \avg*{T}{ \Kopeff{\cdot,L-j}{T}\of{ \Sopeff{j,L-j}{T}} }}{l} \psi
        &
        \nonumber\\
        \iff
        &&
        \bigo\of{T^{L+1}}
        &=\hat{H}_{\mathrm{eff},L}^{(T)}\psi - \avg*{1}{H}\psi
        + \sum_{l=1}^{L} T^l \Tcoeff{\sum_{j=1}^{L} \of*{-\iu}^j \avg*{T}{ \Kopeff{\cdot,L-j}{T}\of{ \Sopeff{j,L-j}{T}} }}{l} \psi
        &
        \nonumber\\
        \label{eq:Heff-condition-proof-1}
        &&
        &=
            \hat{H}_{\mathrm{eff}}^{[0]}\psi - \avg*{1}{H}\psi
            + \sum_{l=1}^{L} T^l \of*{\hat{H}_{\mathrm{eff}}^{[l]} + \Tcoeff{\sum_{j=1}^{L} \of*{-\iu}^j \avg*{T}{ \Kopeff{\cdot,L-j}{T}\of{ \Sopeff{j,L-j}{T} } } }{l} } \psi
        .
    \end{align}
    Considering the limit \(T \to 0\), \cref  {eq:Heff-condition-proof-1} implies
    \begin{equation}
        \hat{H}_{\mathrm{eff}}^{[0]}\psi = \avg*{1}{H}\psi = \Heffl{0}\psi
    \end{equation}
    where the second equality follows from \cref{eq:Heff-definition-0}. Dividing \cref  {eq:Heff-condition-proof-1} by \(T\) and again considering the limit \(T \to 0\) afterwards, we get
    \begin{align}
        \hat{H}_{\mathrm{eff}}^{[1]}\psi
        &= - \Tcoeff{\sum_{j=1}^{L} \of*{-\iu}^j \avg{T}{ \Kopeff{\cdot,L-j}{T}\of{ \Sopeff{j,L-j}{T} } } }{1}\psi
        = \Heffl{1}\psi
    \end{align}
    where the middle term depends on $\hat{H}_{\mathrm{eff}}^{[l]}$ only for $l=0$, which was already shown to be equal to $\Heffl{0}$ in the previous step, and hence the second equality follows from \cref{eq:Heff-definition-l}.
    Iterating this procedure proves \(\hat{H}_{\mathrm{eff}}^{[l]} = \Heffl{l}\) order by order, for \(l = 0, \dots, L\). In other words, \cref{eq:Heff-condition} implies the uniqueness $\hat{H}_{\mathrm{eff},L}^{(T)}=\Heff{L}{T}$. The reverse implication follows easily by going backwards through the above construction.
\end{proof}

Having established a precise characterization of the effective Hamiltonian introduced in \cref{def:Heff-definition}, we now proceed to the asymptotic analysis of the associated dynamics by exploiting the results of \cref{sec:iterated-integration-by-parts}.
This requires an additional structural assumption, namely that the effective Hamiltonian is symmetric and admits self-adjoint extensions.
The following theorem constitutes the first main result of this section.

\begin{theorem}\label{thm:general-error-ip}
    Assume \((U^{(T)}(t))_{t\in\R}\), for all $T>0$, and \((H(t))_{t\in\R}\) satisfy \cref{assump:A-H-cinf,assump:A-U-cinf,assump:A-H-cont,assump:A-U-diff,assump:A-H-bound,assump:A-H-period,assump:A-U-bound}, and let \(\Heff{L}{T}\) be given as in \cref{def:Heff-definition}.
    Moreover, assume that \(\Heff{L}{T}\) admits at least one self-adjoint extension \(\Heffext{L}{T}\).
    Then, for every \(L \in \N_0\), \(t\in \R\), and \(\psi \in \cinftyofh\),
    \begin{equation}
     \label{eq:general-error-ip-2}
       \of*{\e^{-\iu t \Heffext{L}{T} } - U^{(T)}(t)}\psi
        = \bigo\of*{T};
\end{equation}
moreover, for every \(q\in \Z\),
\begin{equation}
    \label{eq:general-error-ip-3}
        \of*{\e^{-\iu qT \Heffext{L}{T} } - U^{(T)}(qT)}\psi
        = \bigo\of*{T^{L+2}}.
\end{equation}
\end{theorem}
\begin{proof}
    The existence of the self-adjoint extension $\Heffext{L}{T}$ means that we can apply Stone's theorem and \cref{eq:assump-diff-U1} holds for $H_1(t)=\Heffext{L}{T}$. Applying then \Cref{thm:iterated-integration-by-parts}, we obtain
    \begin{align}
        &
        \norm*{\of*{\e^{-\iu t \Heffext{L}{T} } - U^{(T)}(t)}\psi}
        \nonumber\\
        &\leq
        \begin{aligned}[t]
            &\sum_{j=1}^{L+1} \norm*{\Sopeff{j,L}{T}\of*{t}U^{(T)}(t) \psi}
            + \sgn(t)\int_0^t \norm*{
                \sum_{j=0}^L (-\iu)^j \avg*{T}{\Kopeff{\cdot,L}{T}\of*{\Sopeff{j,L}{T}}}
                U^{(T)}(s) \psi\,
            }
            \dd s
            \\
            &+\sgn(t)\int_0^t
            \norm*{
                \Kopeff{s,L}{T}\of*{\Sopeff{L+1,L}{T}(s)}
            U^{(T)}(s) \psi\,} \dd s
            .
        \end{aligned}
    \end{align}
    We treat each of these three terms individually with a similar strategy: we write the relevant operator (up to \(U^{(T)}\)) as a polynomial in \(T\), apply the triangle inequality and relative boundedness of the coefficients, and finally make use of \cref{assump:A-U-bound}.
    By this procedure, we obtain an upper bound in which the dependence on \(T\) and \(t\) is explicit and informs us about the desired asymptotics.

    For the first term on the right-hand side, we use \Cref{thm:Sopeff-asymptotics,assump:A-U-bound} to obtain, for every compact interval $I$, some \(\mathfrak{m}_1\geq 0\), \(0 \leq \mathfrak{a}_{1}^{(T)}= \bigo(1)\) and \(0 \leq \mathfrak{b}_{1}^{(T)} = \bigo(1)\)  as \(T\to 0\), such that
    \begin{align}
        \norm*{\Sopeff{j,L}{T}\of*{t}U^{(T)}(t) \psi}
        &\leq
        \sum_{k=j}^{j(L+1)} T^k \of*{
             \mathfrak{a}_{1}^{(T)}\norm*{H_0^{\mathfrak{m}_1} \psi}
            + \mathfrak{b}_{1}^{(T)} \norm*{\psi}
        }
        ,
    \end{align}
    for all $t\in I$, $\psi\in\cinftyofh$ and $j\in\N_0$.

    For the second term, note that, by \cref{thm:Heff-condition,thm:AKS-T-polynomial},
    \begin{align}
        \sum_{j=0}^{L} (-\iu)^j \avg*{T}{\Kopeff{\cdot,L}{T}\of*{\Sopeff{j,L}{T}} }
        &= \sum_{j=0}^{L} (-\iu)^j \sum_{k = j}^{j(L+1)+L}
            T^k \Tcoeff{ \avg*{T}{\Kopeff{\cdot,L}{T}\of*{\Sopeff{j,L}{T}} } }{k}
        \nonumber\\
        & = \sum_{j=1}^{L} (-\iu)^j \sum_{k = L+1}^{j(L+1)+L}
            T^k \Tcoeff{ \avg*{T}{\Kopeff{\cdot,L}{T}\of*{\Sopeff{j,L}{T}} } }{k}
    		.
    \end{align}
    We then use \cref{assump:A-U-bound} and the relative boundedness of each coefficient \(\Tcoeff{ \avg*{T}{\Kopeff{}{L}\of*{\Sopeff{j,L}{T}} } }{k}\) to obtain, for every compact interval $I$, coefficients $\mathfrak{m}_2\geq 0$, \(0 \leq \mathfrak{a}_{2}^{(T)}= \bigo(1)\) and \(0 \leq \mathfrak{b}_{2}^{(T)} = \bigo(1)\), \ as \(T\to 0\), such that
    \begin{align}
        &\norm*{
            \sum_{j=0}^{L} (-\iu)^{j} \avg*{T}{\Kopeff{\cdot,L}{T}\of*{\Sopeff{j,L}{T}} } U^{(T)}(s)\psi
        }
        \leq
        \sum_{j=1}^{L} \sum_{k = L+1}^{j(L+1)+L}
            T^k
            \of*{\mathfrak{a}_{2}^{(T)}\norm*{H_0^{\mathfrak{m}_{2}} \psi} + \mathfrak{b}_{2}^{(T)} \norm*{\psi}}
    \end{align}
    for all $s\in I$ and $\psi\in\cinftyofh$.

    For the third term, we use \cref{thm:Sopeff-T-polynomial,def:Heff-definition} to obtain, for every compact interval $I$, some \(\mathfrak{m}_3\geq 0\), \(0 \leq \mathfrak{a}_{3}^{(T)}= \bigo(1)\) and \(0 \leq \mathfrak{b}_{3}^{(T)} = \bigo(1)\), as \(T\to 0\), such that
    \begin{align}
        \MoveEqLeft\nonumber
            \norm*{
                \Kopeff{s,L}{T}\of*{\Sopeff{L+1,L}{T}(s)}
        U^{(T)}(s) \psi}
        \nonumber\\
        &\leq
        \sum_{k=L+1}^{(L+1)^2}\sum_{l=0}^{L} T^{k+l}
        \norm*{
            \of*{
                \Heffl{l} \Tcoeff{\Sopeff{L+1,L}{T}}{k}(\textfrac{s}{T})
                - \delta_{l,0} \Tcoeff{\Sopeff{L+1,L}{T}}{k}(\textfrac{s}{T}) H(\textfrac{s}{T})
            }
        U^{(T)}(s) \psi}
        \nonumber\\
        &\leq
        \sum_{k=L+1}^{(L+1)^2}\sum_{l=0}^{L} T^{k+l}
        \of*{\mathfrak{a}_{3}^{(T)}\norm*{H_0^{\mathfrak{m}_{3}} \psi} + \mathfrak{b}_{3}^{(T)} \norm*{\psi}}
        ,
    \end{align}
    for all $s\in I$ and $\psi\in\cinftyofh$. Here, we used the fact that the relative boundedness coefficients of \(\Tcoeff{\Sopeff{j,L}{T}}{k}(\textfrac{s}{T})\) and $H(\textfrac{s}{T})$ can be chosen independently of \(I\) and \(T\) due to periodicity  (cf. \cref{thm:Sopeff-asymptotics}).

    Choosing the compact interval $I$ above such that it contains $0$ and $t$, we obtain
    \begin{align}
        \MoveEqLeft\nonumber
        \norm*{\of*{\e^{-\iu t \Heffext{L}{T} } - U^{(T)}(t)}\psi}
        \nonumber\\
        &\leq
        \begin{aligned}[t]
            &\sum_{j=1}^{L+1}\sum_{k=j}^{j(L+1)} T^k \of*{
                 \mathfrak{a}_{1}^{(T)}\norm*{H_0^{\mathfrak{m}_1} \psi}
                + \mathfrak{b}_{1}^{(T)} \norm*{\psi}
            }
            + |t|\sum_{l=1}^{L} \sum_{k = L+1}^{l(L+1)+L}
                T^k
                \of*{\mathfrak{a}_{2}^{(T)}\norm*{H_0^{\mathfrak{m}_{2}} \psi} + \mathfrak{b}_{2}^{(T)} \norm*{\psi}}
            \\
            &+ |t|\sum_{k=L+1}^{(L+1)^2}\sum_{l=0}^{L} T^{k+l}
            \of*{\mathfrak{a}_{3}^{(T)}\norm*{H_0^{\mathfrak{m}_{3}} \psi} + \mathfrak{b}_{3}^{(T)} \norm*{\psi}}
        \end{aligned}
    \end{align}
    which implies \cref{eq:general-error-ip-2} since the lowest-order contribution in this \(T\)-polynomial is of order \(\bigo(T)\).

    In particular, if \(t = qT\) for some $q\in\Z$, we have
    \begin{equation}
        \norm*{\Sopeff{j,L}{T}\of*{t}U^{(T)}(t) \psi}
        = 0
    \end{equation}
    for every \(j \geq 1\) by definition of \(\Sopeff{j,L}{T}\of*{t}\) (cf.~\cref{eq:def-Sopeff}) and its periodicity (cf. \cref{thm:Sopeff-asymptotics}).
    Therefore, we get
    \begin{align}
        &
        \norm*{\of*{\e^{-\iu qT \Heffext{L}{T} } - U^{(T)}(qT)}\psi}
        \nonumber\\
        &\leq
        \sgn(q)
        \of*{
        \int_0^{qT} \norm[\bigg]{
            \sum_{j=0}^L (-\iu)^{j} \avg*{T}{\Kopeff{\cdot,L}{T}\of*{\Sopeff{j,L}{T}}}
            U^{(T)}(s) \psi
        }
        \dd s
        +\int_0^{qT}
        \norm*{
            \Kopeff{s ,L}{T}\of*{\Sopeff{L+1,L}{T}(s)}
        U^{(T)}(s) \psi} \dd s
        }
        \nonumber\\
        &\leq
        \lvert q\rvert T\sum_{l=1}^{L} \sum_{k = L+1}^{l(L+1)+L}
            T^k
            \of*{\mathfrak{a}_{2}^{(T)}\norm*{H_0^{\mathfrak{m}_{2}} \psi} + \mathfrak{b}_{2}^{(T)} \norm*{\psi}}
        + \lvert q\rvert T\sum_{k=L+1}^{(L+1)^2}\sum_{l=0}^{L} T^{k+l}
        \of*{\mathfrak{a}_{3}^{(T)}\norm*{H_0^{\mathfrak{m}_{3}} \psi} + \mathfrak{b}_{3}^{(T)} \norm*{\psi}}
    \end{align}
    which implies \cref{eq:general-error-ip-3} and thus concludes the proof.
\end{proof}
\begin{remark}\label{rem:error-bounds}
    The proof of \cref{thm:general-error-ip} is based on providing an upper bound for the norm difference of the original dynamics and the \(L\)th order approximation.
    Following the proof, it is possible to derive explicit, computable error bounds by calculating the factors \(\mathfrak{a}_{i}^{(T)}\), \(\mathfrak{b}_{i}^{(T)}\), and \(\mathfrak{m}_{i}\) for \(i = 1,2,3\) explicitly in terms of the coefficients in \cref{assump:A-H-bound,assump:A-U-bound}.
    We do not state these bounds here, as their expressions become increasingly involved with growing \(L\).
\end{remark}

At this stage, we do not yet have concrete conditions ensuring that the effective Hamiltonians are symmetric and admit self-adjoint extensions, and hence that the assumptions of \cref{thm:general-error-ip} are satisfied.
In the remainder of this section, we provide such conditions in an explicit form. Specifically, we replace the abstract requirement that \(\Heff{L}{T}\) is symmetric and admits a self-adjoint extension by more concrete hypotheses guaranteeing both properties.
The verification of these conditions is, however, technically involved.
To this end, we proceed in two steps.
We first analyse the simpler case of bounded Hamiltonians, where the relevant properties can be established more directly and where the connection with the Floquet--Magnus expansion can be shown.
This bounded setting will serve as a key intermediate step for the treatment of the general case in \cref{sec:symmetry-sa-extensions}.

\subsection{Relation to Floquet\texorpdfstring{--}{–}Magnus expansion and symmetry (bounded case)}
\label{sec:relation-FM-bounded-case}

In this section, we consider the case where the time-dependent Hamiltonian \((H(t))_{t\in\R}\) is bounded, and prove that, in such cases, \(\Heff{L}{T}\) coincides with the \(L\)th order truncation of the Floquet--Magnus expansion \(\HFM{L}{T}\) introduced in \cref{sec:assumption-a} (cf. \cref{def:floquet-magnus}).
In particular, this implies that \(\Heff{L}{T}\) is symmetric in this setting. This observation will later be extended to the unbounded case in \cref{sec:symmetry-sa-extensions}.
The identity \(\Heff{L}{T} = \HFM{L}{T}\) in the bounded case was already observed in \cite{deyErrorBoundsFloquetMagnus2025}, whose approach we follow closely in this subsection.
The present argument provides a fully self-contained derivation of this relation, avoiding any \emph{a priori} assumptions on self-adjointness and thereby making the underlying structure more transparent.

\begin{remark}\label{rem:Blanes-infinite-dim}
    In this section we will often refer to \cite{Blanes_Casas_Oteo_Ros_2009} as the main reference point for the Floquet--Magnus expansion. There the Floquet--Magnus expansion is mainly treated in the context of finite dimensions; the unwritten assumption of (strong) continuity of \((H(t))_{t\in\R}\) should be added there. A careful check of the constructions shows that the same results hold true in the context of bounded Hamiltonians in infinite dimensions and hence we can apply \cite{Blanes_Casas_Oteo_Ros_2009} henceforth.
\end{remark}

\begin{definition}[\cite[Eqs.~(51) and~(136)]{Blanes_Casas_Oteo_Ros_2009}]\label{def:floquet-magnus}
    For every $T>0$, let \((H^{(T)}(t))_{t\in\R}\) be a strongly continuous family of bounded and symmetric (hence self-adjoint) operators on \(\HH\). For every \(L \in \N_0\), we define the operator \(\HFM{L}{T}\) by
    \begin{equation}
        \HFM{L}{T} = \sum_{l=0}^{L} \HFMl{l}{T},
    \end{equation}
    where
    \begin{equation}\label{eq:FM-expansion-full-expression00}
        \HFMl{l}{T}
        = \frac{\iu}{T}\, \Omega_{l+1}^{(T)}(T),
    \end{equation}
    with
    \begin{equation}\label{eq:FM-expansion-full-expression0}
        \Omega_1^{(T)}(t)
        = -\iu \int_0^t H^{(T)}(\tau)\, \dd{\tau},
    \end{equation}
    and, for \(l = 2,3,4,\dots\),
    \begin{equation}\label{eq:FM-expansion-full-expression}
        \Omega_l^{(T)}(t)
        = \sum_{j=1}^{l} \frac{B_j}{j!}
        \sum_{\substack{
            k_1 + \dots + k_j = l - 1 \\
            k_1 \geq 1, \dots, k_j \geq 1
        }}
        \int_{0}^{t}
        \operatorname{ad}_{\Omega_{k_1}^{(T)}(\tau)}
        \cdots
        \operatorname{ad}_{\Omega_{k_j}^{(T)}(\tau)}
        \bigl(-\iu H^{(T)}(\tau)\bigr)\, \dd{\tau}.
    \end{equation}
    Here, \(\operatorname{ad}_A = [A,\cdot]\), and \(B_j\) denotes the \(j\)th Bernoulli number.
\end{definition}

As noted in \cref{sec:assumption-a}, if \cref{eq:FM-convergence-condition} holds, the sequence \((\HFM{L}{T})_{L \in \N_0}\) converges in the operator norm as \(L \to \infty\). Hereafter, however, we do not require this condition.

By construction, \(\HFMl{l}{T}\) is a symmetric operator for any \(l \in \N_0\). Indeed, for every \(l \in \N\), \(\Omega_l^{(T)}(t)\) is skew-adjoint: \(-\iu H^{(T)}(t)\) is skew-adjoint as \(H^{(T)}(t)\) is symmetric and bounded, and the commutator of two skew-adjoint operators is skew-adjoint.

A straightforward substitution $\sigma=\tau/T$ and induction argument show:

\begin{lemma}\label{thm:FM-expansion-full-expression-rewritten}
    The \(l\)th coefficient of the Magnus expansion in \cref{eq:FM-expansion-full-expression} can be rewritten as
    \begin{equation}
        \Omega_{1}^{(T)}(t)
        = -\iu T \int_{0}^{t/T} H(\sigma)\, \dd{\sigma},
    \end{equation}
    and, for \(l \geq 2\),
    \begin{equation}
        \label{eq:FM-expansion-full-expression-rewritten}
        \Omega_{l}^{(T)}(t)
        = T^{l} \sum_{j=1}^{l-1} \frac{B_j}{j!}
        \sum_{\substack{
            k_1 + \dots + k_j = l - 1 \\
            k_1 \geq 1, \dots, k_j \geq 1
        }}
        \int_{0}^{t/T}
        \operatorname{ad}_{\Omega_{k_1}^{(1)}(\sigma)}
        \cdots
        \operatorname{ad}_{\Omega_{k_j}^{(1)}(\sigma)}
        \bigl(-\iu H(\sigma)\bigr)\, \dd{\sigma}.
    \end{equation}
    In particular, \(\HFMl{l}{T} = T^l \HFMl{l}{}\), where \(\HFMl{l}{}\) is independent of \(T\), hence
    \begin{equation}\label{eq:HFM-poly-expansion}
    \HFM{L}{T}
    = \sum_{l=0}^{L} \HFMl{l}{T}
    = \sum_{l=0}^{L} T^l \HFMl{l}{}
    \end{equation}
    is a polynomial in $T$ of degree at most $L$.
\end{lemma}

For example, for the first two coefficients we have the following explicit expressions:
\begin{align}
    \label{eq:HFM-poly-expansion-l-0}
    \HFMl{0}{}
    &= \int_{0}^{1} H(\sigma_1)\, \dd{\sigma_1},
    \\
    \HFMl{1}{}
    &= -\frac{\iu}{2} \int_{0}^{1}
    \left( \int_{0}^{\sigma_1}
    \comm{H(\sigma_1)}{H(\sigma_2)}\, \dd{\sigma_2} \right)
    \dd{\sigma_1}.
\end{align}

We claim the following result concerning the relation between $\HFM{L}{T}$ and the effective Hamiltonians introduced in the previous section, cf. \cref{def:Heff-definition}:

\begin{theorem}[{Equivalence to the Floquet--Magnus expansion (bounded case) \cite{deyErrorBoundsFloquetMagnus2025}}]\label{thm:equivalence-FM-bounded}
	For every $T>0$, let \(\of*{H^{(T)}(t)}_{t\in \R}\) be a strongly continuous family of bounded operators satisfying \cref{assump:A-H-period,assump:A-H-symmetric}, and let \(\Heff{L}{T}\) and \(\HFM{L}{T}\) be as in \cref{def:Heff-definition} and \cref{def:floquet-magnus}, respectively, for any \(L \in \N_0\).
    Then
\begin{equation}
    \Heff{L}{T} = \HFM{L}{T}.
\end{equation}
In particular, \(\Heff{L}{T}\) is symmetric and hence, being bounded, it is self-adjoint.
\end{theorem}
The remainder of this section is devoted to the proof of \cref{thm:equivalence-FM-bounded}.
To this end, since it is not known \emph{a priori} whether the effective Hamiltonians defined in \cref{def:Heff-definition} are symmetric, we must also consider the dynamics generated by bounded, but not necessarily symmetric operators.
Accordingly, we establish bounded (non-unitary) analogues of \cref{thm:iterated-integration-by-parts} and of the second statement in \cref{thm:general-error-ip}, adapted to general bounded generators.

Despite working with bounded operators in this section, we continue to work primarily in the strong topology. This is motivated by the way we will apply \cref{thm:equivalence-FM-bounded} later in \cref{sec:symmetry-sa-extensions}. Nevertheless, boundedness significantly simplifies the situation. For instance, the following property can be obtained as a direct consequence of the uniform boundedness principle:
\begin{lemma}\label{thm:product-rule-bounded}
    Let \((A(t))_{t\in\R}\) and \((B(t))_{t\in\R}\) be strongly differentiable families of bounded operators on \(\HH\).
    Then the product \(A(t)B(t)\) is strongly differentiable as well and satisfies
    \begin{equation}\label{eq:product-rule-bounded}
        \sodv{}{t} \of*{A(t)B(t)  }
        = \of*{\sodv{}{t} A(t)} B(t)
        + A(t) \of*{\sodv{}{t} B(t)}
        .
    \end{equation}
\end{lemma}

The following theorem can be viewed as a bounded-operator analogue of \cref{thm:iterated-integration-by-parts}, but it is formulated in a more general, non-unitary setting. The proof follows closely the arguments of \cref{thm:iterated-integration-by-parts}, as boundedness ensures that all relevant limiting processes are compatible with multiplicative operations (cf.~\cref{thm:product-rule-bounded}). We present it here for completeness.

\begin{theorem}[Iterated integration-by-parts (bounded, non-unitary case)]
\label{thm:iterated-integration-by-parts-bounded}
    Let \(\of*{H_1(t)}_{t\in \R}\) and \(\of*{H_2(t)}_{t\in \R}\) be two strongly continuous families of bounded operators and let \((\Lambda_1(t))_{t\in\R}\) and \((\Lambda_2(t))_{t\in\R}\) be the respective solutions of the initial value problems
    \begin{equation}
        \begin{cases}
            \odv{}{t} \Lambda_i(t) \psi
            = - \iu H_i(t)\, \Lambda_i(t)\psi
            & t \in \R
            \\
            \Lambda_i(0)\psi = \psi
            &
        \end{cases}
        \qquad i = 1,2.
    \end{equation}
    Then we have
    \begin{equation}
    \label{eq:iterated-integration-by-parts-bounded}
        \Lambda_1(t) - \Lambda_2(t)
        =
        \begin{multlined}[t]
        \sum_{l=1}^{L+1} (-\iu)^l \Sop{l}(t)\, \Lambda_2(t)
        \\
        - \iu \int_0^t \Lambda_1(t)\, \Lambda_1(s)\inv \Bigg(
            \sum_{l=0}^L (-\iu)^l \avg*{T}{\Kop{}\bigl(\Sop{l}\bigr)}
            + (-\iu)^{L+1} \Kop{s}\bigl(\Sop{L+1}(s)\bigr)
        \Bigg) \Lambda_2(s)\, \dd s .
        \end{multlined}
    \end{equation}
\end{theorem}
\begin{proof}
    As \(t \mapsto H_i(t)\psi\) is continuous for any \(\psi \in \HH\), it is integrable in the sense of \cref{def:integral-of-operator}.
    Furthermore, the theory of ordinary differential equations in Banach spaces shows that for \(i = 1,2\) the solution \((\Lambda_i(t))_{t\in\R}\) exists, is invertible in the sense that there is \((\Lambda_i(t)\inv)_{t\in\R}\) with \(\Lambda_i(t) \Lambda_i(t) \inv  = \one = \Lambda_i(t)\inv \Lambda_i(t)\) for any \(t \in \R\), and is uniformly bounded on compact intervals (cf. \cite[Ch.~III.1]{daleckiiStabilitySolutionsDifferential2002}).

    By the product rule,
    \begin{gather}
        \label{eq:derivative-sop-bounded}
        \sodv{}{t} \Sop{l}(t)
        = \Kop{t}\of*{\Sop{l-1}(t)} - \avg*{T}{\Kop{\cdot}\of*{\Sop{l-1}(\cdot)}}
        \\
        \label{eq:derivative-U-sop-U-bounded}
        \sodv{}{t} \Lambda_1(t)\inv \Sop{l}(t)\Lambda_2(t)
        = \iu \Lambda_1(t)\inv \Kop{t}\of*{\Sop{l}(t)}\Lambda_2(t) + \Lambda_1(t)\inv \of*{\sodv{}{t}\Sop{l}(t)} \Lambda_2(t)
        .
    \end{gather}
    Thus, for every \(l \in \N_0\) we have
    \begin{align}
        \Kop{s}\of*{\Sop{l}(s)}
        &= \avg*{T}{\Kop{s}\of*{\Sop{l}(s)} } + \Kop{s}\of*{\Sop{l}(s)} - \avg*{T}{\Kop{s}\of*{\Sop{l}(s)} }
        \nonumber\\
        &= \avg*{T}{\Kop{s}\of*{\Sop{l}(s)}} + \sodv{}{s} \Sop{l+1}(s)
    \label{eq:iterated-integration-by-parts-proof-1-bounded}
    \end{align}
    and
    \begin{equation}
    \label{eq:iterated-integration-by-parts-proof-2-bounded}
        \Lambda_1(s)\inv \of*{\sodv{}{s}\Sop{l}(s)} \Lambda_2(s)
        = \sodv{}{s} \of*{\Lambda_1(s)\inv \Sop{l}(s) \Lambda_2(s)} - \iu \Lambda_1(s)\inv \Kop{s}\of*{\Sop{l}(s)} \Lambda_2(s)
        .
    \end{equation}
    This allows us to rewrite certain integrals of \(\Sop{l}(s)\) in terms of similar integrals of \(\Sop{l+1}(s)\):
    \begin{align}
    \MoveEqLeft
        \int_0^t \Lambda_1(t)\Lambda_1(s)\inv \Kop{s}\of*{\Sop{l}(s)} \Lambda_2(s) \psi \dd s
        \nonumber\\
        &= \int_0^t \Lambda_1(t) \Lambda_1(s)\inv \avg*{T}{\Kop{}\of*{\Sop{l}}} \Lambda_2(s) \psi \dd s
        + \int_0^t \Lambda_1(t)\Lambda_1(s)\inv \of*{\sodv{}{s}\Sop{l+1}(s)} \Lambda_2(s) \psi \dd s
    \label{eq:iterated-integration-by-parts-proof-3-bounded}
        \\
        &=
        \begin{multlined}[t]
            \Sop{l+1}(t) \Lambda_2(t) \psi
            + \int_0^t \Lambda_1(t)\Lambda_1(s)\inv \avg*{T}{\Kop{}\of*{\Sop{l}}} \Lambda_2(s)\psi \dd s
            \\
            + (-\iu) \int_0^t \Lambda_1(t) \Lambda_1(s)\inv \Kop{s}\of*{\Sop{l+1}(s)} \Lambda_2(s) \psi \dd s
            ,
        \end{multlined}
    \label{eq:iterated-integration-by-parts-proof-4-bounded}
    \end{align}
    where we used \cref{eq:iterated-integration-by-parts-proof-1-bounded} in the first step, and \cref{eq:iterated-integration-by-parts-proof-2-bounded} in the second step.
    The third summand in \cref{eq:iterated-integration-by-parts-proof-4-bounded} has the same structure as the expression in \cref{eq:iterated-integration-by-parts-proof-3-bounded}, so that the same procedure can be iterated up to the desired order.
    This will be made precise below in terms of an induction over \(L \in \N_0\).
    For \(L = 0\), we have
    \begin{align}
        \Lambda_1(t) - \Lambda_2(t)
        &= - \Lambda_1(t) \of*{ \Lambda_1(t)\inv \Sop{0}(t)  \Lambda_2(t) - \one }
        \nonumber\\
        &= - \Lambda_1(t) \int_0^t \sodv{}{s} \of*{\Lambda_1(s)\inv \Sop{0}(s) \Lambda_2(s)} \dd s
        \nonumber\\
        \label{eq:IIBP-iteration--1-duhamel}
        &\annotateeqn{\eqref{eq:derivative-U-sop-U-bounded}}{=}
        - \iu \int_0^t \Lambda_1(t) \Lambda_1(s)\inv \Kop{s}\of*{\Sop{0}(s)} \Lambda_2(s) \dd s
        \nonumber\\
        &= - \iu \Sop{1}(t) \Lambda_2(t)
        - \iu \int_0^t \Lambda_1(t) \Lambda_1(s)\inv \of*{\avg*{T}{\Kop{}\of*{\Sop{0}}} - \iu \Kop{s}\of*{\Sop{1}(s)} } \Lambda_2(s) \dd s
        \nonumber\\
        &=
        \begin{aligned}[t]
            &\sum_{l=1}^{L+1} (-\iu)^l \Sop{l}(t) \Lambda_2(t)
            \\
            &- \iu \int_0^t \Lambda_1(t) \Lambda_1(s)\inv \of*{ \sum_{l=0}^L (-\iu)^l \avg*{T}{\Kop{}\of*{\Sop{l}}} + (-\iu )^{L+1} \Kop{s}\of*{\Sop{L+1}(s)} }\Lambda_2(s)  \dd s,
        \end{aligned}
    \end{align}
    which coincides with \cref{eq:iterated-integration-by-parts-bounded} for \(L=0\).

    Now suppose \cref{eq:iterated-integration-by-parts-bounded} holds for some \(L \in \N_0\). Then
    \begin{align}
        \Lambda_1(t) - \Lambda_2(t)
        &=
        \begin{aligned}[t]
            &\sum_{l=1}^{L+1} (-\iu)^l \Sop{l}(t) \Lambda_2(t)
            \\
            &- \iu \int_0^t \Lambda_1(t) \Lambda_1(s)\inv \Bigg(
                \sum_{l=0}^L (-\iu)^l \avg*{T}{\Kop{}\of*{\Sop{l}}}
                + (-\iu)^{L+1} \Kop{s}\of*{\Sop{L+1}(s)}
            \Bigg) \Lambda_2(s) \dd s
        \end{aligned}
        \nonumber\\
        &=
        \begin{aligned}[t]
            &\sum_{l=1}^{L+1} (-\iu)^l \Sop{l}(t) \Lambda_2(t)
            - \iu \int_0^t \Lambda_1(t) \Lambda_1(s)\inv \sum_{l=0}^L (-\iu)^l \avg*{T}{\Kop{}\of*{\Sop{l}}} \Lambda_2(s)\dd s
            \\
            &+ (-\iu)^{L+2} \int_0^t \Lambda_1(t) \Lambda_1(s)\inv \Kop{s}\of*{\Sop{L+1}(s)} \Lambda_2(s) \dd s
        \end{aligned}
        \nonumber\\
        &=
        \begin{aligned}[t]
            &\sum_{l=1}^{L+1} (-\iu)^l \Sop{l}(t) \Lambda_2(t)
            - \iu \int_0^t \Lambda_1(t) \Lambda_1(s)\inv \sum_{l=0}^L (-\iu)^l \avg*{T}{\Kop{}\of*{\Sop{l}}} \Lambda_2(s)\dd s
            \\
            &+ (-\iu)^{L+2} \Sop{L+2}(t) \Lambda_2(t)
            -\iu \int_0^t \Lambda_1(t)\Lambda_1(s)\inv (-\iu)^{L+1} \avg*{T}{\Kop{}\of*{\Sop{L+1}}} \Lambda_2(s) \dd s
            \\
            &-\iu \int_0^t \Lambda_1(t) \Lambda_1(s)\inv (-\iu)^{L+2} \Kop{s}\of*{\Sop{L+2}(s)} \Lambda_2(s) \dd s
        \end{aligned}
        \nonumber\\
        &=
        \begin{aligned}[t]
            &\sum_{l=1}^{L\primed+1} (-\iu)^l \Sop{l}(t) \Lambda_2(t)
            \\
            &- \iu \int_0^t \Lambda_1(t) \Lambda_1(s)\inv \Bigg(
                \sum_{l=0}^{L\primed} (-\iu)^l \avg*{T}{\Kop{}\of*{\Sop{l}}}
                + (-\iu)^{L\primed+1} \Kop{s}\of*{\Sop{L\primed+1}(s)}
            \Bigg) \Lambda_2(s)  \dd s
            ,
        \end{aligned}
    \end{align}
    which is \cref{eq:iterated-integration-by-parts-bounded} for \(L\primed = L+1\).
    Thus, by induction, \cref{eq:iterated-integration-by-parts-bounded} holds for every \(L\in \N_0\).
\end{proof}

Next, in analogy with the unitary, unbounded case of \cref{thm:iterated-integration-by-parts}, we use \cref{thm:iterated-integration-by-parts-bounded} to construct an effective generator whose associated dynamics approximates a given evolution to the desired order, at least at integer multiples of the period.
As we are not assuming symmetry of the effective generator, its exponential will be defined via its power series rather than through the spectral theorem.
Recall that the exponential of a bounded operator \(A\) is given by
    \begin{equation}
        \label{eq:bounded-exponential-power-series}
        \exp\of{A}
        = \sum_{n=0}^{\infty} \frac{1}{n!} A^n
        ,
    \end{equation}
the series above being absolutely convergent with $\opnorm{\exp\of{A}}\leq\e^{\opnorm{A}}$.
\begin{proposition}\label{thm:asymptotics-dynamics-stroboscopic-bounded}
    Let \(\of*{H(t)}_{t\in \R}\) be a strongly continuous family of bounded operators satisfying \cref{assump:A-H-period,assump:A-H-symmetric} and let \((\Lambda(t))_{t\in\R}\) be the corresponding unitary solution of the initial value problem
    \begin{equation}\label{eq:non-symmetric-initial-value-problem}
        \begin{cases}
            \odv{}{t} \Lambda(t) \psi
            = - \iu H(t)\, \Lambda(t) \psi
            & t \in \R
            \\
            \Lambda(0)\psi = \psi
            .
            &
        \end{cases}
    \end{equation}
	Denote the rescaled Hamiltonian and its evolution family by \(H^{(T)}(t) = H(t/T)\) and \(\Lambda^{(T)}(t)\), respectively.
    Let \(L \in \N_0\) and \(q\in \N_0\).
    Then    \begin{gather}
        \of*{\exp\of{-\iu qT \Heff{L}{T} } - \Lambda^{(T)}(qT)}
        = \bigo\of*{T^{L+2} }
        ,
    \end{gather}
    with $\Heff{L}{T}$ as per \cref{def:Heff-definition}.
\end{proposition}
\begin{proof}
    Since \(H^{(T)}(t)\) is bounded, \cref{assump:A-H-bound,assump:A-H-cinf} are satisfied.
    By \cref{thm:Heff-condition,thm:Sopeff-asymptotics}, we have
    \begin{equation}
    \label{eq:asymptotics-dynamics-stroboscopic-bounded-proof-1}
        \Bigg(
            \sum_{l=0}^L (-\iu)^l \avg*{T}{\Kopeff{\cdot,L}{T}\of*{\Sopeff{l,L}{T}}}
            + (-\iu)^{L+1} \Kopeff{sT,L}{T}\of*{\Sopeff{L+1,L}{T}(sT)}
        \Bigg) \psi
        = \bigo\of*{T^{L+1}}
        ,\quad \text{as}\ T \to 0,
    \end{equation}
    for any \(\psi \in \HH\).
    Due to \cref{thm:AKS-T-polynomial,thm:Sopeff-T-polynomial} this means that the operator on the left-hand side of \cref{eq:asymptotics-dynamics-stroboscopic-bounded-proof-1} is a polynomial in \(T\) of order at least \(L+1\) and hence
    there exist \(c_L ,T_0 > 0 \) such that
    \begin{equation}
    \opnorm*{\Bigg(
        \sum_{l=0}^L (-\iu)^l \avg*{T}{\Kopeff{\cdot,L}{T}\of*{\Sopeff{l,L}{T}}}
        + (-\iu)^{L+1} \Kopeff{sT,L}{T}\of*{\Sopeff{L+1,L}{T}(sT)}
    \Bigg)}
    \leq c_L T^{L+1}
    \end{equation}
    for every \(0 < T < T_0\).
    Due to periodicity, \(c_L\) can be chosen independently of \(s\).

    Denote \(\Lambda_{L}^{(T)}(t) = \exp\of{-\iu t \Heff{L}{T}}\).
    For \(0 < q T < T_0\), we apply \cref{thm:iterated-integration-by-parts-bounded} to obtain
    \begin{align}
        &
        \opnorm*{\Lambda_{L}^{(T)}(qT) - \Lambda^{(T)}(qT) }
        \nonumber\\
        &\leq
        \begin{aligned}[t]
            &\opnorm*{\sum_{l=1}^{L+1} (-\iu)^l \Sopeff{l,L}{T}(qT) \Lambda^{(T)}(qT) }
            \\
            &+ \sgn(q)\int_0^{qT} \opnorm*{
            \Lambda_{L}^{(T)}(qT-s)
            \Bigg(
                \sum_{l=0}^L (-\iu)^l \avg*{T}{\Kopeff{\cdot,L}{T}\of*{\Sopeff{l,L}{T}}}
                + (-\iu)^{L+1} \Kopeff{s,L}{T}\of*{\Sopeff{L+1,L}{T}(s)}
            \Bigg) } \dd s
        \end{aligned}
        \nonumber\\
        &\leq
            \sgn(q)T\int_0^{q}
            \opnorm*{\Lambda_{L}^{(T)}(qT-sT)}
            \norm*{\Bigg(
                \sum_{l=0}^L (-\iu)^l \avg*{T}{\Kopeff{\cdot,L}{T}\of*{\Sopeff{l,L}{T}}}
                + (-\iu)^{L+1} \Kopeff{sT,L}{T}\of*{\Sopeff{L+1,L}{T}(sT)}
            \Bigg)}_{\mathrlap{\mathrm{op}}} \dd s
        \nonumber\\
        &\leq \sgn(q)c_{L} T^{L+1} \int_0^{qT} \opnorm*{\exp\of{-\iu (qT-s) \Heff{L}{T} }} \dd s
        \nonumber\\
        &\leq \sgn(q) c_{L} T^{L+1} \int_0^{qT} \e^{\abs{qT - s} \opnorm*{\Heff{L}{T}}} \dd s
        \nonumber\\
        &= c_{L} T^{L+1} \frac{1}{\opnorm*{\Heff{L}{T}}}\of*{ \e^{\abs{q} T\opnorm*{\Heff{L}{T}}} -1 }
        \nonumber\\
        &= c_{L} T^{L+1} \sum_{n=1}^{\infty} \frac{1}{n!} (\abs{q}T)^n \opnorm*{\Heff{L}{T}}^{n-1}
        \nonumber\\
        &= \abs{q} c_{L} T^{L+2} + \bigo\of*{\abs{q}^{2}T^{L+3} \opnorm*{\Heff{L}{T}}}
        \nonumber\\
        &= \bigo\of*{T^{L+2}}
        ,
    \end{align}
where in the second step we substituted \(s \mapsto s/T\), and in the third step we reverted the substitution.
Here, we assumed \(\opnorm{\Heff{L}{T}} \neq 0\). If \(\opnorm{\Heff{L}{T}} = 0\), then the integral \(\int_0^{qT} \e^{\abs{qT - s} \opnorm{\Heff{L}{T}}} \dd s\) evaluates to \(qT\) and we obtain an estimate by \(\abs{q} c_L T^{L+2}\) without higher order terms.
This concludes the proof.
\end{proof}

In addition to \cref{thm:iterated-integration-by-parts-bounded} and \cref{thm:asymptotics-dynamics-stroboscopic-bounded}, we will also need the following result:
\begin{lemma}\label{thm:continuity-logarithm}
    Consider two norm-continuous $\mathcal{B}(\mathcal{H})$-valued functions \(\R \ni t \mapsto A(t)\) and  \(\R \ni t \mapsto B(t)\), with \(A(0) = B(0) = 0\).
    Define \(\Lambda_{A}(t) = \exp\of*{A(t)}\) and \(\Lambda_{B}(t) = \exp\of*{B(t)}\).
    For any \(C > 1\) there is \(r > 0\) such that
    \begin{equation}
        \opnorm{A(t) - B(t)}
        \leq C \opnorm*{\Lambda_A(t) - \Lambda_B(t)}
        \quad \text{for every}\ \abs{t} \leq r
        .
    \end{equation}
\end{lemma}
\begin{proof}
    For any bounded operator \(Z\) with \(\opnorm*{Z - \one} < 1\), we define \cite[Eq.~(20) of Ch.~2,\S.7]{Bourbaki_1-3_1998}
    \begin{equation}
        \ln\of{Z} = \sum_{n=1}^{\infty} \frac{(-1)^{n+1}}{n} (Z - \one)^n
        ,
    \end{equation}
    where the series on the right converges absolutely.
    If \(\opnorm{Z} < \ln\of{2}\), \(\opnorm{\exp\of{Z}-\one} < 1\) and
    \begin{equation}
        \label{eq:continuity-logarithm-proof-1}
        \ln\of{\exp\of{Z}}= Z
    \end{equation}
    with \cref{eq:bounded-exponential-power-series} \cite[Eq.~(22) of Ch.~2,\S.7]{Bourbaki_1-3_1998}.

    Let \(C>1\) and \(\delta = \ln\of{2-\frac{1}{C}}\).
    By continuity and \(A(0)=B(0) = 0\) there is \(r>0\) s.t.
    \begin{equation}
        \label{eq:continuity-logarithm-proof-2}
        \opnorm*{A(t)} \leq \delta
        ,\quad
        \opnorm*{B(t)} \leq \delta
        \quad \text{for every}\ \abs{t} \leq r
        .
    \end{equation}

    Having fixed such choice of \(r\) and \(\abs{t} \leq r\), one also easily checks that
    \begin{equation}
        \label{eq:continuity-logarithm-proof-3}
        \opnorm*{\Lambda_A(t) - \one} \leq \e^{\delta} - 1
        ,\quad
        \opnorm*{\Lambda_B(t) - \one} \leq \e^{\delta} - 1
    \end{equation}
    by \cref{eq:bounded-exponential-power-series}.
    As \(C >1\), we have \(0 < \delta < \ln\of{2}\) and \cref{eq:continuity-logarithm-proof-1} holds with \(Z = A(t)\) and \(Z = B(t)\), respectively.

    This means,
    \begin{align}
        \opnorm{A(t) - B(t)}
        &= \opnorm*{ \ln\of{\Lambda_A(t)} - \ln\of{\Lambda_B(t) }}
        \nonumber\\
        &\leq \sum_{n=1}^{\infty} \frac{1}{n} \opnorm*{ \of*{\Lambda_A(t) - \one}^n - \of*{\Lambda_B(t) - \one}^n }
        \nonumber\\
        &\leq \sum_{n=1}^{\infty} \frac{1}{n} \sum_{j=0}^{n-1} \opnorm*{ \of*{\Lambda_A(t) - \one}^j \of*{\Lambda_A(t) - \Lambda_B(t)} \of*{\Lambda_B(t) - \one}^{n-1-j} }
        \nonumber\\
        &\leq \opnorm*{\Lambda_A(t) - \Lambda_B(t)} \sum_{n=1}^{\infty} \frac{1}{n} n (\e^\delta -1)^{n-1}
        \nonumber\\
        &= \frac{1}{1-\of*{\e^\delta - 1}} \opnorm*{\Lambda_A(t) - \Lambda_B(t)}
        \nonumber\\
        &= C \opnorm*{\Lambda_A(t) - \Lambda_B(t)}
        .
    \end{align}
\end{proof}
    With these results at hand, we can finally proceed with the proof of \cref{thm:equivalence-FM-bounded}.

\begin{proof}[Proof of \cref{thm:equivalence-FM-bounded}]
    We follow the proof idea in Ref.~\cite[Sec.~C]{deyErrorBoundsFloquetMagnus2025}.
    If $H(s)\equiv 0$ then clearly \(\HFM{L}{T}= 0= \Heff{L}{T}\). Consider the case $H(s)\not\equiv 0$. Then we may choose \(T>0\) such that
    \begin{equation}
        T < \pi \of*{\textstyle\int_{0}^{1} \opnorm{H(s)} \dd s}^{-1}
        .
    \end{equation}
    Then
    \begin{equation}
        \int_{0}^{T} \opnorm{H^{(T)}(t)} \dd t
        = T \int_{0}^{1} \opnorm{H(s)} \dd s
        < \pi
        .
    \end{equation}
    Consequently, the Floquet--Magnus series converges to some self-adjoint bounded operator \(\HFM{}{T}\) \cite[Theorem~9]{Blanes_Casas_Oteo_Ros_2009} and we have
    \begin{equation}
        \HFM{}{T} - \HFM{L}{T} = \bigo\of*{T^{L+1}}
        .
    \end{equation}
    Moreover, by construction of \(\HFM{}{T}\), we have
    \begin{equation}
        \Lambda^{(T)}(T) = \e^{-\iu T \HFM{}{T}}
    \end{equation}
    where \(\Lambda^{(T)}(t)\) is the solution of \cref{eq:non-symmetric-initial-value-problem}  for \(H^{(T)}(t)\).
    Hence, applying \cref{eq:IIBP-iteration--1-duhamel} to \(\Lambda_1(s) = \e^{-\iu T \HFM{}{T}s}\) and \(\Lambda_2(s) = \e^{-\iu T \HFM{L}{T} s }\), taking norms, applying the triangle inequality and submultiplicativity of \(\opnorm{}\), and evaluating at \(s=1\) in the end, we get
    \begin{equation}
        \opnorm*{ \Lambda^{(T)}(T) -  \e^{-\iu T \HFM{L}{T} }}
        = \opnorm*{ \e^{-\iu T \HFM{}{T}} -  \e^{-\iu T\HFM{L}{T}  }}
        \leq \int_0^1 T\opnorm*{ \HFM{}{T} - \HFM{L}{T} } \dd s
        = \bigo\of*{T^{L+2}}
        .
    \end{equation}
    Together with \cref{thm:asymptotics-dynamics-stroboscopic-bounded}, this allows us to compare \(\e^{-\iu \HFM{L}{T} T }\) and \(\e^{-\iu \Heff{L}{T} T }\):
    \begin{equation}
        \e^{-\iu T\HFM{L}{T}  } - \e^{-\iu T\Heff{L}{T}  }
        = \of*{\e^{-\iu T\HFM{L}{T} } - \Lambda^{(T)}(T)} - \of*{\e^{-\iu T\Heff{L}{T} } - \Lambda^{(T)}(T) }
        = \bigo\of*{T^{L+2}}
        .
    \end{equation}
    Now, using \cref{thm:continuity-logarithm} with  \(A(t) = -\iu t\HFM{L}{T} \) and \(B(t) = -\iu t\Heff{L}{T} \), which in contrast to \(H^{(T)}(t)\) are continuous in the operator norm,  and for \(t = T>0\) small enough, gives us
    \begin{equation}
        \label{eq:equivalence-FM-bounded-proof-1}
        -\iu T \of*{\HFM{L}{T} - \Heff{L}{T}}
        = \bigo\of*{T^{L+2}}
        .
    \end{equation}
    On the other hand, \(\HFM{L}{T}\) and \(\Heff{L}{T}\) are both polynomials of order \(L\) in \(T\) (cf. \cref{thm:FM-expansion-full-expression-rewritten}) and, therefore, so is their difference \(\HFM{L}{T} - \Heff{L}{T}\).
    Thus, \cref{eq:equivalence-FM-bounded-proof-1} can only be true if \(\HFM{L}{T} - \Heff{L}{T} = 0\), which concludes the proof.
\end{proof}

\subsection{Symmetry and self-adjoint extensions of \texorpdfstring{\(\Heff{L}{T}\)}{effective Hamiltonians}}
\label{sec:symmetry-sa-extensions}

In this subsection, we return to the general (possibly unbounded) setting.
We first show that \(\Heff{L}{T}\), when defined on \(\cinftyofh\), is symmetric when assuming that \cref{assump:A-H-symmetric} holds.
We then prove that \(\Heff{L}{T}\) admits self-adjoint extensions under the further assumption of \cref{assump:A-H-conj}.

We begin with the following structural observation concerning the effective Hamiltonians constructed in \cref{sec:heff_construction}, cf.~\cref{def:Heff-definition}:
\begin{proposition}\label{thm:Heff-integral-polynomial}
    Assume \((H(t))_{t\in\R}\) satisfies \cref{assump:A-H-cinf,assump:A-H-cont,assump:A-H-bound,assump:A-H-period,assump:A-H-symmetric}.
    For every \(L \in \N_0\), there exist $v_L,d_L\in\N_0$, a polynomial \(\mathcal{P}_L\) in \(v_L\) variables of total degree \(d_L\), and a compact region \(\mathcal{M}_L \subseteq \R^{v_L}\) such that
    \begin{equation}\label{eq:heff_poly}
        \Heff{L}{T}\psi
        = \int_{\mathcal{M}_L} \mathcal{P}_L(H(s_1),\dots,H(s_{v_L}))\psi \dd{\underline{s}}
        ,
    \end{equation}
    for every \(\psi \in \cinftyofh\).
\end{proposition}
\begin{proof}
    We prove the claim by induction over \(L\in \N_0\).
    We have
    \begin{equation}
        \Heff{0}{T} = \avg{1}{H} = \int_0^1 H(s) \dd s
        ,
    \end{equation}
    hence the claim holds with \(\mathcal{P}_0(x) = x\) and \(\mathcal{M}_0 = [0,1]\).

    Now, assume the claim holds up to some \(L \in \N\).
    Being symmetric, \(H(s)\) is closable; moreover, the map
    \begin{equation}
        \underline{s}\mapsto H(s) \mathcal{P}_L(H(s_1),\dots,H(s_{v_L}))\psi
    \end{equation}
    is integrable by the relative \(H_0^{k_0}\)-boundedness of \(H(s)\) and \cref{thm:strong-continuity-monomials}.
    Therefore,
    \begin{align}
        H(s) \int_{\mathcal{M}_L} \mathcal{P}_L(H(s_1),\dots,H(s_{v_L}))\psi\dd{\underline{s}}
        &= \close{H(s)} \int_{\mathcal{M}_L} \mathcal{P}_L(H(s_1),\dots,H(s_{v_L}))\psi\dd{\underline{s}}
        \nonumber\\
        &= \int_{\mathcal{M}_L} \close{H(s)}\mathcal{P}_L(H(s_1),\dots,H(s_{v_L}))\psi\dd{\underline{s}}
        \nonumber\\
        &= \int_{\mathcal{M}_L} H(s) \mathcal{P}_L(H(s_1),\dots,H(s_{v_L}))\psi\dd{\underline{s}}
        ,
    \end{align}
    for every \(\psi \in \cinftyofh\).
    Furthermore, it follows from \cref{eq:def-Kopeff,eq:def-Deltaeff,eq:def-Sopeff-0,eq:def-Sopeff} that both \(\avg{T}{\Kopeff{\cdot,L'}{T}(\cdot)}\) and \(\Sopeff{j,L'}{T}(t)\) can be written as integrals of polynomials in \(H(s)\), for all \(j \in \N_0\) and all \(L' \leq L\). Invoking \cref{eq:Heff-definition-LT}, we then conclude that \(\Heff{L+1}{T}\) is itself given by an integral of polynomials, since it depends only on terms of the form \(\avg{T}{\Kopeff{\cdot,L'}{T}(\cdot)}\) and \(\Sopeff{j,L'}{T}(t)\) with \(L' \leq L\). This completes the induction step.
\end{proof}

We now establish an important property of (not necessarily symmetric) operators on \(\cinftyofh\) of the form appearing on the right-hand side of \cref{eq:heff_poly}:
\begin{lemma}\label{thm:polynomial_adjoint}
    Assume \((H(t))_{t\in\R}\) satisfies \cref{assump:A-H-cinf,assump:A-H-cont,assump:A-H-bound,assump:A-H-symmetric}.
    Consider an operator \(A\) on \(\cinftyofh\) of the form
    \begin{equation}
        A\psi = \int_{\mathcal{M} } \of{\mathcal{P}(H(s_1),\dots,H(s_v))} \psi \dd{\underline{s}}
    \end{equation}
    for some polynomial \(\mathcal{P}\) in \(v\in \N_0\) variables and of degree \(d\in \N_0\), and \(\mathcal{M}\subseteq \R^{v}\) being a compact region.
    Then we have
    \begin{equation}
        \cinftyofh \subseteq \domof{A\adj}
    \end{equation}
    and
    \begin{equation}
        A\adj\psi  = \int_{\mathcal{M}} \of{\mathcal{P}(H(s_1),\dots,H(s_v))}\adj \psi \dd{\underline{s}}
    \end{equation}
    for all \(\psi \in \cinftyofh\).
\end{lemma}
\begin{proof}
    Let \(\psi,\phi \in \cinftyofh\) and
    \begin{equation}
        \mathcal{P}(H(s_1),\dots,H(s_v))
        = \sum_{n=0}^{d} \sum_{\underline{i}\in \Set{1,\dots,v}^{\times n}} c_{\underline{i}} H(s_{i_1})\dots H(s_{i_n})
        ,
    \end{equation}
    with certain $c_{\underline{i}}\in\C$. Then
    \begin{align}
        \innerp{\phi}{A\psi}
        &= \sum_{n=0}^{d} \sum_{\underline{i}\in \Set{1,\dots,v}^{\times n}} c_{\underline{i}} \int_{\mathcal{M}} \innerp{\phi}{H(s_{i_1})\dots H(s_{i_n})\psi} \dd{\underline{s}}
        \nonumber\\
        &= \sum_{n=0}^{d} \sum_{\underline{i}\in \Set{1,\dots,v}^{\times n}} c_{\underline{i}} \int_{\mathcal{M}} \innerp{H(s_{i_n})\dots H(s_{i_1})\phi}{\psi} \dd{\underline{s}}
        \nonumber\\
        &= \innerp*{\int_{\mathcal{M}}\sum_{n=0}^{d} \sum_{\underline{i}\in \Set{1,\dots,v}^{\times n}}\overline{c_{\underline{i}}}H(s_{i_n})\dots H(s_{i_1})\phi \dd{\underline{s}}}{\psi}
        \nonumber\\
        &= \innerp*{\int_{\mathcal{M}} \of{\mathcal{P}(H(s_1),\dots,H(s_v))}\adj \dd{\underline{s}}\,\phi}{\psi}
        ,
    \end{align}
    as \(\cinftyofh\) is invariant under \(H(s)\), and \(H(s)\) is symmetric.
    This proves the desired claim.
\end{proof}

In particular, by \cref{thm:Heff-integral-polynomial}, \cref
{thm:polynomial_adjoint} applies to our effective Hamiltonians:
\begin{corollary}
    Assume \((H(t))_{t\in\R}\) satisfies \cref{assump:A-H-cinf,assump:A-H-cont,assump:A-H-bound,assump:A-H-period,assump:A-H-symmetric}.
    We have
    \begin{equation}
        \cinftyofh \subseteq \domof*{\of*{\Heff{L}{T}}\adj}
    \end{equation}
    and
    \begin{equation}
        \of*{\Heff{L}{T}}\adj\psi  = \int_{\mathcal{M}_L} \tilde{\mathcal{P}}_L(H(s_1),\dots,H(s_{v_L})) \psi \dd{\underline{s}},
    \end{equation}
    where
    \begin{equation}
        \tilde{\mathcal{P}}_L(H(s_1),\dots,H(s_{v_L}))\psi
        = \of*{\mathcal{P}_L(H(s_1),\dots,H(s_{v_L}))}\adj\psi
    \end{equation}
    for all \(\psi \in \cinftyofh\) and \(\tilde{\mathcal{P} }\) is a polynomial of the same order as \(\mathcal{P}\).
\end{corollary}

The following result provides a general mechanism for transferring algebraic identities established in the bounded setting to the present unbounded framework. We will then use this property to establish symmetry of the effective Hamiltonians.

\begin{proposition}\label{thm:integral-polynomials-and-spectral-approximations}
    Let \(\mathcal{P}\) be a polynomial in \(v \in \N_0\) variables of degree \(d\in\N_0\), and consider
    \begin{equation}
        A(H)\psi = \int_{\mathcal{M}} \mathcal{P}\bigl(H(s_1), \dots, H(s_v)\bigr)\, \psi \,\dd{\underline{s}},
    \end{equation}
    where \(\mathcal{M} \subseteq \R^v\) is a compact region. Suppose \(A(H) = 0\) for every choice of strongly continuous family of self-adjoint bounded Hamiltonians \((H(t))_{t\in\R}\).
    Then \(A(H) = 0\) for every (possibly unbounded) \((H(t))_{t\in\R}\) satisfying \cref{assump:A-H-symmetric,assump:A-H-cinf,assump:A-H-bound,assump:A-H-cont}.
\end{proposition}
\begin{proof}
    For simplicity of notation, we will work with the norms
    \begin{equation}
        \norm{\psi}_m = \norm{\of{\one + H_0^2}^{m/2} \psi}
        ,\quad m \in \N_0
    \end{equation}
    on \(\domof{H_0^m}\) which are equivalent to the respective graph norm of \(H_0^m\).
    Thus, by \cref{assump:A-H-bound}, for any \(m\in \N_0\) there is \(\gamma_m \geq 0\) s.t.
    \begin{equation}
        \label{eq:integral-polynomials-bounded-to-unbounded-proof-1}
        \norm{H(t)\psi}_{m}
        \leq \gamma_m \norm{\psi}_{m+k_m}
    \end{equation}
    for any \(\psi \in \cinftyofh\) uniformly for \(t\) in compact intervals.

    Let \(E_0\) be the projector-valued measure associated with \(H_0\), and for any \(N \in \N\), let \(P_N = E_0([-N,N])\) and \(Q_N = \one - P_N\).
    Then, we have
    \begin{equation}
        \label{eq:integral-polynomials-bounded-to-unbounded-proof-3}
        \norm{Q_N \psi}_{m} \leq \norm{\psi}_{m}
        ,\quad
        \norm{P_N \psi}_{m} \leq \norm{\psi}_{m}
    \end{equation}
    as well as
    \begin{equation}
        \label{eq:integral-polynomials-bounded-to-unbounded-proof-2}
        \norm{Q_N \psi}_m
        \leq (1+N^2)^{-\frac{\tilde{m}-m}{2}} \norm{\psi}_{\tilde{m}}
    \end{equation}
    for any \(0 \leq m \leq \tilde{m}\) and \(\psi \in \cinftyofh\).
    In fact, the latter estimate is readily shown by spectral calculus:
    \begin{align}
        \norm{Q_N \psi}^2_{m}
        &= \int_{\R\setminus [-N,N]} (1+x^2)^m \dd{\langle \psi, E_0\psi\rangle(x)}
        \nonumber\\
        &\leq (1+ N^2)^{-(\tilde{m}-m)} \int_{\R\setminus [-N,N]} (1+x^2)^{\tilde{m}} \dd{\mu_{\psi}(x)}
        \nonumber\\
        &\leq (1+ N^2)^{-(\tilde{m}-m)} \norm{\psi}^2_{\tilde{m}}
        .
    \end{align}

    Now, let \(n,k \in \N\), \(\psi \in \cinftyofh\), \(\underline{s}\in\R^n\), and \(X_j \in \cof{P_N, Q_N}\) for \(j = 1,\dots, n+1\) such that there is \(j' \in \cof{1,\dots, n+1}\) with \(X_{j\primed} = Q_N\).
    Based on \cref{eq:integral-polynomials-bounded-to-unbounded-proof-1}, we adopt the following notation: \(K_0 = 0\), \(K_n = K_{n-1} + k_{K_{n-1}}\), \(\Gamma_n = \prod_{j=0}^{n-1}\gamma_{K_j}\) for \(n \in \N\).
    We have the following estimate:
    \begin{align}
        \norm{X_1 H(s_{1}) \dots X_n H(s_{n}) X_{n+1} \psi}
        &\leq \Gamma_{j' - 1} \norm{Q_N H(s_{j\primed}) \dots X_n H(s_n) X_{n+1} \psi }_{K_{j'-1}}
        \nonumber\\
        &\leq (1+N^2)^{-k/2} \Gamma_{j' - 1} \norm{H(s_{j\primed}) \dots X_n H(s_{n}) X_{n+1} \psi }_{K_{j'-1} + k}
        \nonumber\\
        &\leq (1+N^2)^{-k/2} \tilde{\Gamma}_n \norm{\psi }_{\tilde{K}_n}
        ,
    \end{align}
    where in the first step we have applied \cref{eq:integral-polynomials-bounded-to-unbounded-proof-1,eq:integral-polynomials-bounded-to-unbounded-proof-3} repeatedly starting from the left until we reach \(X_{j'} = Q_N\), in the second step we have applied \cref{eq:integral-polynomials-bounded-to-unbounded-proof-2} with \(m = K_{j'}\) and \(\tilde{m} = K_{j'} + k\), and in the third step we have again applied \cref{eq:integral-polynomials-bounded-to-unbounded-proof-1,eq:integral-polynomials-bounded-to-unbounded-proof-3} repeatedly for the remainder to obtain the final estimate for some \(\tilde{\Gamma}_n>0\) and \(\tilde{K}_n\in \N\) that are independent of \(N\).

    Now, consider the operator \(A(H)\).
    We can write it as
    \begin{equation}
        A(H)
        = \int_{\mathcal{M}} \of[\bigg]{ \sum_{n=0}^{d} \sum_{\underline{i}\in \cof{1,\dots, v}^{\times n}} c_{n,\underline{i}} H(s_{i_1})\dots H(s_{i_n}) }\dd{\underline{s}}
        ,
    \end{equation}
    insert \(\one = P_N + Q_N\), expand, and obtain
    \begin{align}
        A(H)\psi
        = A(P_N H(\cdot) P_N)\psi
            + \int_{\mathcal{M}} \of[\bigg]{\sum_{n=0}^{d} \sum_{\underline{i}\in \cof{1,\dots, v}^{\times n}} c_{n,\underline{i}} \sum_{\substack{ \mathclap{\underline{X}\in \cof{P_N,Q_N}^{\times n+1}} \\ \abs{\cof{X_j = Q_N}}\geq 1}}  X_1 H(s_{i_1})\dots X_n H(s_{i_n})X_{n+1} }\dd{\underline{s}}
            .
    \end{align}
    The coefficient \(c_0\) must vanish by assumption because of the corresponding term being independent of \(H(t)\).
    As \((P_N H(t) P_N)_{t\in \R}\) is bounded, strongly continuous and self-adjoint due to \cref{assump:A-H-bound,assump:A-H-cont,assump:A-H-symmetric}, we have \(A(P_N H(\cdot) P_N)\psi = 0\) by assumption and we obtain
    \begin{align}
        \norm{A(H)\psi}
        &\leq \int_{\mathcal{M}} \of[\bigg]{\sum_{n=1}^{d} \sum_{\underline{i}\in \cof{1,\dots, v}^{\times n}} \abs*{c_{n,\underline{i}}} \sum_{\substack{ \mathclap{\underline{X}\in \cof{P_N,Q_N}^{\times n+1}} \\ \abs{\cof{X_j = Q_N}}\geq 1}}  \norm*{X_1 H(s_{i_1})\dots X_n H(s_{i_n})X_{n+1} \psi} }\dd{\underline{s}}
        \nonumber\\
        &\leq (1+N^2)^{-k/2} \int_{\mathcal{M}} \of[\bigg]{\sum_{n=1}^{d} \sum_{\underline{i}\in \cof{1,\dots, v}^{\times n}} \abs*{c_{n,\underline{i}}} \sum_{\substack{ \mathclap{\underline{X}\in \cof{P_N,Q_N}^{\times n+1}} \\ \abs{\cof{X_j = Q_N}}\geq 1}}  \tilde{\Gamma}_n \norm{\psi}_{\tilde{K}_{n}} }\dd{\underline{s}}
        \nonumber\\
        &\leq (1+N^2)^{-k/2} \int_{\mathcal{M}} \of[\bigg]{ \sum_{n=1}^{d}2^{n+1} \sum_{\underline{i}\in \cof{1,\dots, v}^{\times n}} \abs{c_{n,\underline{i}}} \tilde{\Gamma}_n \norm{\psi}_{\tilde{K}_{n}} } \dd{\underline{s}}
        \nonumber\\
        &\to 0
    \end{align}
    as \(N \to \infty\), because \(-k/2 < 0\).
    As \(\norm{A(H)\psi}\) is independent of \(N\), this means \(\norm{A(H)\psi} = 0\) and hence \(A(H)\psi = 0\).
    The claim follows, as \(\psi \in \cinftyofh\) was arbitrary.
\end{proof}

\begin{corollary}\label{thm:effective-Hamiltonian-symmetric}
    Let \(L \in \N_0\), and assume \((H(t))_{t\in\R}\) satisfies \cref{assump:A-H-symmetric,assump:A-H-cinf,assump:A-H-bound,assump:A-H-cont,assump:A-H-period}. Then  \(\Heff{L}{T}\) is symmetric on \(\cinftyofh\).
\end{corollary}
\begin{proof}
	This follows directly by applying \cref{thm:integral-polynomials-and-spectral-approximations} to \(A = \Heff{L}{T} - \of*{\Heff{L}{T}}\adj\) and by \cref{thm:equivalence-FM-bounded}.
\end{proof}
We have established that, by assuming \cref{assump:A-H-cinf,assump:A-U-cinf,assump:A-H-cont,assump:A-U-diff,assump:A-H-bound,assump:A-U-bound,assump:A-H-period,assump:A-H-symmetric} to be satisfied, our construction yields symmetric effective Hamiltonians. Since these operators are, in general, unbounded, it remains to show that they admit self-adjoint extensions. To this end, we require the additional \cref{assump:A-H-conj}, which ensures the existence of a conjugation $J$ such that $JH(t) \subseteq H(-t)J$. For the sake of completeness we recall:
\begin{definition}
    A \emph{conjugation} is an anti-unitary involution on \(\HH\), i.e. a map \(J\colon \HH \to \HH\) with
    \begin{align}
        \label{eq:anti-linear}
        J(\alpha \phi + \beta \psi)
        &= \conj{\alpha} J \phi + \conj{\beta} J \psi
        ,\\
        \label{eq:anti-isometric}
        \innerp{J \phi}{J \psi}
        &= \innerp{\psi}{\phi}
        ,\\
        \label{eq:involution}
        J J
        &= \one
        .
    \end{align}
\end{definition}

\begin{theorem}[{von Neumann's criterion~\cite[Thm.~5.43]{morettiSpectralTheoryQuantum2017}}]\label{thm:vonNeumanncriterion}
    Let \(A\) be a symmetric operator on \(\HH\).
    If there exists a conjugation \(J\) such that
    \begin{equation}
        J A \subseteq A J
        ,
    \end{equation}
    then \(A\) admits at least one self-adjoint extension.
\end{theorem}

Our goal is to show that, if \cref{assump:A-H-conj} is satisfied, then $J \Heff{L}{T} \subseteq \Heff{L}{T} J$, thus establishing the existence of self-adjoint extensions of $\Heff{L}{T}$. We begin with a preliminary lemma:

\begin{lemma}\label{thm:conjugation-Sop}
    Assume \((H(t))_{t\in\R}\) satisfies \cref{assump:A-H-cinf,assump:A-H-cont,assump:A-H-bound,assump:A-H-period,assump:A-H-conj} and \(J\Heffl{l}J = \Heffl{l}\) for all \(l=0,\dots,L\) and some \(L \in \N_0\).
    For every \(j \in \N_0\)
    \begin{equation}
        \label{eq:conjugation-Heff-proof-1}
        J\Sopeff{j,l}{T}(s)J
        = (-1)^j \Sopeff{j,l}{T}(-s)
        .
    \end{equation}
\end{lemma}
\begin{proof}
    We prove the claim by induction over \(j\in\N_0\).
    We have
    \begin{equation}
        J\Sopeff{0,l}{T}(s)J
        = J J
        = \one
        = (-1)^0 \Sopeff{0,l}{T}(-s)
        ;
    \end{equation}
    so the claim holds for \(j=0\).
    Now assume the claim holds for some \(j\).
    Then, for \(F_j(s) = \Kopeff{s,l}{T}\of*{\Sopeff{j,l}{T}(s)}\),
    \begin{equation}
        J F_j(s) J
        = J\Heff{l}{T} JJ \Sopeff{j,l}{T}(s)J - J\Sopeff{j,l}{T}(s)JJ H^{(T)}(s)J
        = (-1)^j F_j(-s)
    \end{equation}
    and therefore
    \begin{equation}
        J \avg*{T}{F_j} J
        = (-1)^j \avg*{T}{F_j}
        ,
    \end{equation}
    as \(F_j\) is \(T\)-periodic.
    In conclusion, we have
    \begin{align}
        J \Sopeff{j+1,l}{T} (t) J
        &= J \int_{0}^{t} \of*{F_j(s) - \avg*{T}{F_j} } \dd s J
        \nonumber\\
        &= - (-1)^j \int_{0}^{-t} \of*{F_j(s) - \avg*{T}{F_j} } \dd s
        \nonumber\\
        &= (-1)^{j+1} \Sopeff{j+1,l}{T}(-t)
    \end{align}
    which proves the claim for \(j+1\), thus completing the proof.
\end{proof}

\begin{proposition}\label{thm:effective-Hamiltonian-conjugation}
    Assume \((H(t))_{t\in\R}\) satisfies \cref{assump:A-H-cinf,assump:A-H-cont,assump:A-H-bound,assump:A-H-period,assump:A-H-symmetric,assump:A-H-conj} and let \(\Heff{L}{T}\) be given according to \cref{eq:Heff-definition-0,eq:Heff-definition-l} for every \(L \in \N_0\). Then
    \begin{equation}
        J \Heff{L}{T} \subseteq \Heff{L}{T} J
    \end{equation}
    for every \(L \in \N_0\), and therefore $\Heff{L}{T}$ admits at least one self-adjoint extension.
\end{proposition}
\begin{proof}

It suffices to prove
\begin{equation}
    J \Heffl{l} J = \Heffl{l}
\end{equation}
for all $l=0,\ldots,L$ and $L\in\N_0$ and we do this by induction again.

For \(L=0\), we have
\begin{equation}
    J\Heffl{0}J
    = J \avg*{1}{H} J
    = \avg*{1}{H}
    = \Heffl{0},
\end{equation}
since \(JH(s)J = H(-s)\) and \(J \cinftyofh \subseteq \cinftyofh\) with \(J\) being a conjugation implies \(J \cinftyofh = \cinftyofh\), so the claim holds for $L=0$.

Assume now that the claim holds for all \(l=0,\dots,L\) and some $L\in\N_0$, i.e., we have $J \Heffl{l} J = \Heffl{l}$ for $l=0,\ldots,L$, and we need to show this for $l=L+1$. By definition of $\Heffl{L+1}$ (\cref{eq:Heff-definition-l}) and anti-linearity of $J$, we have
\begin{align}
    J \Heffl{L+1} J
    &= - \sum_{j=1}^{L+1} (-1)^j (-\iu)^j
    \Tcoeff{ J \avg*{T}{\Kopeff{\cdot,L+1-j}{T}\bigl(\Sopeff{j,L+1-j}{T}\bigr)} J }{L+1},
\end{align}
where the right-hand side depends on $\Heffl{l}$ only for $l=0,\ldots,L$, for which our claim is already assumed to be true. Then it follows from \cref{thm:conjugation-Sop} that
\begin{equation}
    J \Kopeff{s,L+1-j}{T}\bigl(\Sopeff{j,L+1-j}{T}(s)\bigr)J = (-1)^j \Kopeff{-s,L+1-j}{T}\bigl(\Sopeff{j,L+1-j}{T}(-s)\bigr)
\end{equation}
and hence
\begin{equation}
    J \avg*{T}{\Kopeff{\cdot,L+1-j}{T}\bigl(\Sopeff{j,L+1-j}{T}\bigr)}J = (-1)^j \avg*{T}{\Kopeff{\cdot,L+1-j}{T}\bigl(\Sopeff{j,L+1-j}{T}\bigr)}.
\end{equation}
Substituting this into the previous expression yields
\begin{align}
    J \Heffl{L+1} J
    &= - \sum_{j=1}^{L+1} (-\iu)^j
    \Tcoeff{ \avg*{T}{\Kopeff{\cdot,L+1-j}{T}\bigl(\Sopeff{j,L+1-j}{T}\bigr)} }{L+1}
    = \Heffl{L+1},
\end{align}
which proves the claim for $L+1$ and concludes the induction step.

 This proves the existence of a conjugation \(J\) such that \(J \Heff{L}{T} \subseteq \Heff{L}{T} J\). On the other hand, by \cref{assump:A-H-symmetric} and \cref{thm:effective-Hamiltonian-symmetric}, \(\Heff{L}{T}\) is symmetric. Therefore, by von Neumann's criterion for the existence of self-adjoint extensions, \cref{thm:vonNeumanncriterion}, \(\Heff{L}{T}\) admits at least one self-adjoint extension.
   \end{proof}

We have shown that, under our assumptions, all effective Hamiltonians are symmetric and admit self-adjoint extensions. We are now in a position to complete the argument initiated in \cref{sec:relation-FM-bounded-case} and to establish that our construction indeed reproduces the Floquet--Magnus expansion:

\begin{theorem}\label{thm:FM-Eff-equal-coefficients-unbounded}
    Let \((H(t))_{t\in\R}\) satisfy \cref{assump:A-H-cinf,assump:A-H-cont,assump:A-H-bound,assump:A-H-period,assump:A-H-symmetric,assump:A-H-conj} and \(\HFM{L}{T}\) be given as in \cref{eq:FM-expansion-full-expression} for every \(L \in \N_0\) and \(T>0\). Then \(\HFM{L}{T}\) is well-defined on \(\cinftyofh\) and coincides with \(\Heff{L}{T}\).
\end{theorem}
\begin{proof}
    Due to \cref{thm:FM-expansion-full-expression-rewritten,thm:Heff-integral-polynomial}, both \(\Heff{L}{T}\) and \(\HFM{L}{T}\) are of the form
    \begin{equation}
        \int_{\mathcal{M}} \mathcal{P}(H(s_1),\dots,H(s_v)) \dd{\underline{s}}
        .
    \end{equation}
    Therefore, so is \(A \coloneqq \Heff{L}{T} - \HFM{L}{T}\).
    Due to \cref{assump:A-H-cinf,assump:A-H-cont,def:integral-of-operator}, this implies \(\cinftyofh \subseteq \domof{\HFM{L}{T}}\) and \(\cinftyofh \subseteq \domof{A}\): \(H(s_j)\) leaves \(\cinftyofh\) invariant, so it is contained in the domain of \(\mathcal{P}(H(s_1),\dots,H(s_v))\) for any \(\underline{s}\); as \(\underline{s}\mapsto \mathcal{P}(H(s_1),\dots,H(s_v))\psi\) is continuous by \cref{thm:strong-continuity-monomials}, it is integrable and \(\cinftyofh\) is contained in the domains of \(A\) and \(\HFM{L}{T}\).

    By \cref{thm:equivalence-FM-bounded}, \(A\) would vanish if \((H(t))_{t\in\R}\) were a strongly continuous family of bounded self-adjoint operators. Therefore, by \cref{thm:integral-polynomials-and-spectral-approximations}, \(A= 0\) also for unbounded \((H(t))_{t\in\R}\) satisfying \cref{assump:A-H-symmetric,assump:A-H-cinf,assump:A-H-bound,assump:A-H-cont}.
\end{proof}

\section{Interaction picture of a time-independent Hamiltonian}\label{sec:interaction-picture}
In this section, we specialise our discussion to time-independent Hamiltonians of the form $H_0+V$, with $
	\domof{H_0 + V} = \domof{H_0}
    $,
and consider the associated interaction-picture dynamics, see \cref{sec:assumption-b}.
In this framework, the time-dependent Hamiltonian is given by
\begin{equation}
\label{eq:sec-effective-Hamiltonian-H-in-interatction-picture}
    H(t)
    := \e^{\iu t H_{0}}\, V \, \e^{-\iu t H_{0}},
\end{equation}
acting on $\domof{H_0}$.
The corresponding unitary evolution in the interaction picture is given by
\begin{equation}
	U(t)
	= \e^{\iu t H_0} \e^{-\iu t \of{H_0 + V}}.
\end{equation}
The operator $H(t)$ is well-defined for all $t \in \R$, since $\e^{-\iu t H_0}$ leaves $\domof{H_0}$ invariant and $\domof{H_0} \subseteq \domof{V}$ whenever $V$ is $H_0$-bounded (and symmetric) \cite[Lemma~6.2]{teschlMathematicalMethodsQuantum2009}.

We work under assumption of \assumpB{} for $H_0$ and $V$, which, in particular, ensures that $H(t)$ is $1$-periodic.
As in the general setting, a $T$-periodic Hamiltonian can then be obtained by rescaling time (cf.~\cref{eq:rescaling}).
In the present case, this corresponds, at the level of the original Schrödinger picture Hamiltonian, to replacing $H_0 + V$ by
\begin{equation}\label{eq:rescaling_interactionpicture}
	\textfrac{1}{T} H_0 + V
    .
\end{equation}
The corresponding unitary propagator will then be denoted by $U^{(T)}(t)=\e^{\iu \frac{t}{T} H_0} \e^{-\iu t \of{\frac{1}{T}H_0 + V}}$.

We will prove that this setting is a special case of the more general one considered in the previous sections: specifically, we show that the operators $H_0,(H(t))_{t\in\R},(U(t))_{t\in\R}$ introduced above satisfy \assumpA{}, thereby allowing us to specialise the main results of \cref{sec:Heff} to the interaction-picture dynamics; more precisely, we show the following:
\begin{center}
\begin{tabular}{rcl}
     \cref{assump:B-V-sym,assump:B-V-bound}& \(\implies\) &\cref{assump:A-H-cinf,assump:A-U-bound,assump:A-H-cont,assump:A-U-diff,assump:A-H-bound,assump:A-U-cinf} ,
     \\
     \cref{assump:B-V-period} & \(\iff\) & \cref{assump:A-H-period},
     \\
    \cref{assump:B-V-sym} & \(\implies\) & \cref{assump:A-H-symmetric},
    \\
    \cref{assump:B-V-conj} & \(\implies\) & \cref{assump:A-H-conj}.
\end{tabular}
\end{center}

We start with \cref{assump:A-H-cinf}.

\begin{lemma}\label{thm:invariance-under-V}
    Suppose that, for every \(m \in \N\), there exists an integer \(k_m \in \N_0\) such that \(H_0^m V\) is \(H_0^{m + k_m}\)-bounded. Then
    \begin{equation}\label{eq:V-energy-limited}
        V \domof{H_0^{m+k_m}} \subseteq \domof{H_0^m}
        ,
    \end{equation}
    and hence
    \begin{equation}
        V \cinfty\of{H_0} \subseteq \cinfty\of{H_0}
        .
    \end{equation}
    In particular, this is the case if \cref{assump:B-V-bound} is satisfied.
\end{lemma}
\begin{proof}
    If \(H_0^m V\) is \(H_0^{m+k_m}\)-bounded, then $\domof{H_0^{m+k_m}} \subseteq \domof{H_0^m V}$, so $V\domof{H_0^{m+k_m}} \subseteq V\domof{H_0^m V} \subseteq \domof{H_0^m}$. Hence,
    \begin{equation}
        V \cinfty\of{H_0}
        = V \bigcap_{m=1}^\infty \domof{H_0^{m + k_m}}
        \subseteq \bigcap_{m=1}^\infty V \domof{H_0^{m + k_m}}
        \subseteq \bigcap_{m=1}^\infty \domof{H_0^m}
        = \cinfty\of{H_0}.\qedhere
    \end{equation}
\end{proof}

As an immediate consequence of \cref{thm:invariance-under-V}, we get the following:

\begin{corollary}
    Suppose \cref{assump:B-V-sym,assump:B-V-bound} hold.
    Then \((H(t))_{t\in\R}\) satisfies \cref{assump:A-H-cinf}.
\end{corollary}

Our next goal is to show that \cref{assump:A-U-cinf} holds.
To this end, we first prove an additional lemma:

\begin{lemma}\label{thm:dom-H-m}
    Suppose \cref{assump:B-V-sym,assump:B-V-bound} hold. Then, for every $T>0$, we have
    \begin{equation}
        \domof*{\of*{\textfrac{1}{T}H_0 + V}^m} = \domof*{H_0^m}\quad \forall m \in \N.
    \end{equation}
    In particular, $\cinfty\of{\textfrac{1}{T}H_0+V}=\cinfty\of{H_0}$.
\end{lemma}
\begin{proof}
    Let \(m \in \N\) be arbitrary.
    Denote the restriction of \(\of*{\textfrac{1}{T}H_0}^m\) to \(\cinftyofh\), which is essentially self-adjoint, by \(A\) and write
    \begin{gather}
        \evalat{\of*{\textfrac{1}{T}H_0 + V}^m}{\cinftyofh}
        = A + B
        \shortintertext{where}
        B
        = \sum_{n = 1}^{m} \frac{1}{T^{m-n}} \sum_{\substack{\underline{i}\in \cof{1,2}^{\times m}\\ \abs{\cof{i_j=1}}=n}} \evalat{A_{i_1}\dots A_{i_m}}{\cinftyofh}
        ,
    \end{gather}
    and with \(A_1 = V\) and \(A_2 = H_0\).
    Because of \cref{assump:B-V-bound}, \(B\) is infinitesimally relatively bounded with respect to \(A\).
    By the Kato-Rellich theorem, \(A + B\) is essentially self-adjoint and
    \begin{equation}
        \domof*{\of*{\textfrac{1}{T}H_0 + V}^m}
         = \domof*{\close{A+B}}
         = \domof*{\close{A}}
        = \domof*{H_0^m}
    \end{equation}
    which concludes the proof.
\end{proof}
As $\e^{\iu t \of{\frac{1}{T}H_0 + V}}$ leaves $\cinfty\of*{\frac{1}{T}H_0+V}$ invariant by spectral calculus, this immediately implies the following:
\begin{corollary}\label{thm:invariance-under-evolution}
    Suppose \cref{assump:B-V-sym,assump:B-V-bound} hold. Then, for every $T>0$,
    \begin{equation}
        \e^{\iu t \of{\frac{1}{T}H_0 + V}} \cinfty\of{H_0}
        \subseteq \cinfty\of{H_0}
        \quad \forall t \in \R
        .
    \end{equation}
    In particular, $U^{(T)}(t)\cinftyofh \subseteq \cinftyofh$, and \cref{assump:A-U-cinf} holds.
\end{corollary}

We turn to \cref{assump:A-H-cont,assump:A-U-diff,assump:A-H-bound}.

\begin{proposition}
    \cref{assump:B-V-bound} implies \cref{assump:A-H-bound}. \cref{assump:B-V-sym,assump:B-V-bound} imply \cref{assump:A-H-cont}.
\end{proposition}
\begin{proof}
	Assume \cref{assump:B-V-bound} holds, and let \(\psi \in \cinftyofh\subseteq\domof*{H_0^{m+1}}\).
	Then, for every $t\in\R$, we have
	\begin{equation}
        \begin{aligned}
	    \norm*{H_0^m H(t)\psi}
		&= \norm*{H_0^m \e^{+\iu t H_0} V \e^{-\iu t H_0}\psi}\\
		&\leq a_m \norm*{H_0^{m + 1}\e^{-\iu t H_0}\psi} + b_m \norm*{\e^{-\iu t H_0}\psi}
		= a_m \norm*{H_0^{m + 1}\psi} + b_m \norm*{\psi}
		,
        \end{aligned}
	\end{equation}
    hence \cref{assump:A-H-bound} holds with $k_m=1$. Suppose \cref{assump:B-V-sym} also holds, and let \(s_0 \in \R\), \(m \in \N_0\), and \(\psi \in \cinftyofh\). Then
    \begin{align}
        \norm{H_0^m H(s)\psi - H_0^m H(s_0)\psi}
        &\leq
            \norm{H_0^m \of{\e^{\iu s H_0} - \e^{\iu s_0 H_0}} V \e^{-\iu s_0 H_0}\psi }
            + \norm{H_0^m \e^{\iu s H_0} V \of{\e^{-\iu s H_0} - \e^{-\iu s_0 H_0}}\psi }
        \nonumber\\
        &\leq\begin{multlined}[t]
            \norm{ \of{\e^{\iu s H_0} - \e^{\iu s_0 H_0}} H_0^m V \e^{-\iu s_0H_0}\psi }
        \\
            + a_m \norm{\of{\e^{-\iu s H_0} - \e^{-\iu s_0 H_0}} H_0^{m + 1} \psi }
            + b_m \norm{\of{\e^{-\iu s H_0} - \e^{-\iu s_0 H_0}} \psi }
        \end{multlined}
        \nonumber\\
        &\to 0
    \end{align}
    as \(s \to s_0\), so \cref{assump:A-H-cont} holds.
\end{proof}

\begin{lemma}
    \label{thm:relative-boundedness-with-dynamics}
    Let \(T>0\), and suppose \(\frac{1}{T}H_0 + V\) and \(H_0\) are self-adjoint operators with \(\domof{H_0^m} = \domof*{\of*{\frac{1}{T}H_0 + V}^m}\) for every \(m\in\N\).
    Then, for every \(s,\,t\in\R\), \(H_0^m \e^{\iu \frac{s}{T} H_0} \e^{\iu t \of{\frac{1}{T}H_0 + V}}\) is \(H_0^m\)-bounded and \(\of{\frac{1}{T}H_0 + V}^m\)-bounded.
\end{lemma}
\begin{proof}
    By Lemma \ref{thm:dom-H-m} we have \(\e^{\iu t (\frac{1}{T}H_0+V)}\domof{H_0^m} = \domof{H_0^m}\) and \(\e^{\iu \frac{s}{T} H_0}\domof{H_0^m} = \domof{H_0^m}\).
    Hence,
    \begin{equation}
        \domof{H^m_0 \e^{\iu \frac{s}{T} H_0} \e^{\iu t (\frac{1}{T}H_0+V)}} = \domof{H_0^m} = \domof{(\textfrac{1}{T}H_0+V)^m}
        .
    \end{equation}
    As \(H^m_0\) is self-adjoint on \(\domof{H_0^m}\) and \(\e^{\iu t (\frac{1}{T}H_0+V)}\) as well as \(\e^{\iu \frac{s}{T} H_0}\) are bounded, \(H^m_0 \e^{\iu \frac{s}{T} H_0} \e^{\iu t (\frac{1}{T}H_0+V)}\) is closed.
    Due to
    ~\cite[Lemma~6.2]{teschlMathematicalMethodsQuantum2009}, \(H^m_0 \e^{\iu s H_0} \e^{\iu t (\frac{1}{T}H_0+V)}\) is \(H_0^m\)- and \((\frac{1}{T}H_0+V)^m\)-bounded.
\end{proof}

\begin{proposition}
    Suppose \cref{assump:B-V-sym,assump:B-V-bound} hold.
    Then \cref{assump:A-U-diff} holds, i.e. \(t \mapsto H_0^m U^{(T)}(t)\psi\) is differentiable and
    \begin{equation}
        \odv{}{t}H_0^m U^{(T)}(t) \psi
        = - \iu H_0^m H^{(T)}(t) U^{(T)}(t) \psi
    \end{equation}
    for every \(T>0\), $m\in\N_0$, and \(\psi \in \cinftyofh\).
\end{proposition}
\begin{proof}
	For every \(T>0\), $m\in\N_0$, and $t\in\R$, we have
    \begin{align}
        \MoveEqLeft\nonumber
        \frac{1}{h} \of*{ H_0^m U^{(T)}(t+h) - H_0^m U^{(T)}(t) } \psi
        =
        \\
        &= \frac{1}{h}\of*{ \e^{\iu  \frac{t+h}{T} H_0}H_0^m \e^{-\iu (t+h)\of{\frac{1}{T}H_0 + V}} - \e^{\iu \frac{t}{T} H_0}H_0^m \e^{-\iu t \of{\frac{1}{T}H_0 + V}} }\psi
        \nonumber\\
        &=
            \frac{1}{h}\of*{ \e^{\iu \frac{t+h}{T} H_0} - \e^{\iu \frac{t}{T} H_0} } H_0^m \e^{-\iu t \of{\frac{1}{T}H_0 + V}} \psi
            + \e^{\iu \frac{t+h}{T} H_0} H_0^m \frac{1}{h} \of*{ \e^{-\iu (t+h)\of{\frac{1}{T}H_0 + V}} - \e^{-\iu t \of{\frac{1}{T}H_0 + V}} }\psi
        \nonumber\\
        &\to \of*{\iu \textfrac{1}{T} H_0} \e^{\iu \frac{t}{T} H_0}H_0^m \e^{-\iu t \of{\frac{1}{T}H_0 + V}}\psi + \e^{\iu \frac{t}{T} H_0} H_0^m \of*{-\iu \of{\textfrac{1}{T}H_0 + V}}\e^{-\iu t \of{\frac{1}{T}H_0 + V}} \psi
        \nonumber\\
        &= -\iu H_0^m H^{(T)}(t) U^{(T)}(t) \psi
    \end{align}
    as \(h \to 0\).
    Concerning the second summand above, we have
    \begin{equation}
        \e^{\iu \frac{t+h}{T} H_0}H_0^m \frac{1}{h} \of*{ \e^{-\iu (t+h)\of{\frac{1}{T}H_0 + V}} - \e^{-\iu t \of{\frac{1}{T}H_0 + V}} }\psi
        \to \e^{\iu \frac{t}{T} H_0}H_0^m \of*{-\iu \of{\textfrac{1}{T}H_0 + V}} \e^{-\iu t \of{\frac{1}{T}H_0 + V}}\psi
        ,
    \end{equation}
    because \(\e^{\iu \frac{t+h}{T} H_0}\psi \to \e^{\iu \frac{t}{T} H_0}\psi\) and \(\e^{\iu \frac{t}{T} H_0}\) is uniformly bounded, \(H_0^m\) is \(\of{\frac{1}{T}H_0 + V}^m\)-bounded by \cref{thm:relative-boundedness-with-dynamics}, and
    \begin{equation}
    \frac{1}{h} \of*{ \e^{-\iu (t+h)\of{\frac{1}{T}H_0 + V}} - \e^{-\iu t \of{\frac{1}{T}H_0 + V}} }\psi \to \of*{-\iu \of{\textfrac{1}{T}H_0 + V}} \e^{-\iu t \of{\frac{1}{T}H_0 + V}}\psi.\qedhere
    \end{equation}
\end{proof}

We now come to \cref{assump:A-U-bound}.

\begin{proposition}\label{thm:V-sym-V-bound-implies-U-bound}
	Suppose \cref{assump:B-V-sym,assump:B-V-bound} hold.
	Then,  \((U^{(T)}(t))_{t\in\R}\) satisfies \cref{assump:A-U-bound}.
\end{proposition}
\begin{proof}
    Let \(m \in N\), \(\psi \in \cinftyofh\), and \(t\in \R\). It is clear from \cref{assump:B-V-sym,assump:B-V-bound} and \cref{thm:dom-H-m,thm:relative-boundedness-with-dynamics} that, for every \(T > 0\), there exist \(a_m^{(T)},\, b_m^{(T)} \geq 0\) such that
    \begin{equation}
        \norm*{H_0^m\psi} \leq a_m^{(T)} \norm*{(\textfrac{1}{T}H_0 + V)^m \psi}+ b_m^{(T)}\norm*{\psi}
        ,
    \end{equation}
whence
\begin{align}
        \norm*{H_0^m U^{(T)}(t)\psi}
        &= \norm*{H_0^m \e^{+\iu \frac{t}{T} H_0}\e^{-\iu t (\frac{1}{T}H_0+V)} \psi}
        \nonumber\\
        &
        = \norm*{H_0^m \e^{-\iu t (\frac{1}{T}H_0 + V)} \psi}
        \nonumber\\
        &
        \leq a_m^{(T)} \norm*{\of*{\textfrac{1}{T}H_0 + V}^m \psi}
        + b_m^{(T)} \norm*{\psi}.
    \end{align}
    Thus, it suffices to show that \(a_m^{(T)}\) and \(b_m^{(T)}\) can be chosen so that
    \begin{equation}\label{eq:V-bound-implies-U-bound-1}
        a_m^{(T)} = \bigo(T^m),
        \quad
        b_m^{(T)} = \bigo(T^0)
        \quad
        \text{as}\ T \to 0.
    \end{equation}
    Indeed, in this case we would have
    \begin{equation}
        \norm*{(\textfrac{1}{T}H_0 + V)^m\psi} \leq c_m\frac{1}{T^m}\norm*{H_0^m\psi} + d_m\frac{1}{T^m}\norm*{\psi},
    \end{equation}
    with certain $c_m,d_m\geq 0$ independent of $T$ and $\psi$, and hence we would get
    \begin{equation}
        \norm*{H_0^m U^{(T)}(t)\psi} \leq
        a'^{(T)}_m \norm*{H_0^m\psi}
        + b'^{(T)}_m \norm*{\psi}
    \end{equation}
    with
    \begin{equation}
        a'^{(T)}_m \coloneqq a_m^{(T)} c_m \frac{1}{T^m}
        = \bigo(T^0)
        ,\quad
        b'^{(T)}_m \coloneqq a_m^{(T)} d_m\frac{1}{T^m} + b_m^{(T)}
        = \bigo(T^0)
        , \quad
        \text{as}\ T \to 0,
    \end{equation}
    thereby proving \cref{assump:A-U-bound}.

    In order to prove \cref{eq:V-bound-implies-U-bound-1}, we write
    \begin{align}
        \of*{\textfrac{1}{T}H_0 + V}^m
        = \of*{ \textfrac{1}{T}H_0 }^m
        + \sum_{n=0}^{m-1} \frac{1}{T^n} \sum_{\substack{
            \underline{i} \in \Set{1,2}^{\times m}
            \\
           \abs{\Set{i_j = 2}} =n
        } }
        A_{i_1} \dots A_{i_m}
        ,
    \end{align}
    where \(A_1 = V\) and \(A_2 = H_0 \).
    Thus
    \begin{align}
        \norm*{\of*{\textfrac{1}{T}H_0}^m \psi}
        &\leq
        \norm*{\of*{ \textfrac{1}{T}H_0 + V }^m \psi}
        + \sum_{n=0}^{m-1} \frac{1}{T^n} \sum_{\substack{
            \underline{i} \in \Set{1,2}^{\times m}
            \\
           \abs{\Set{i_j = 2}} =n
        } }
        \norm*{A_{i_1} \dots A_{i_m}\psi}
        \nonumber\\
        &\leq \norm*{\of*{ \textfrac{1}{T}H_0 +V}^m \psi}
        +
        \sum_{n=0}^{m-1} \frac{1}{T^n} \binom{m}{n}
        \of*{ a \norm*{H_0^{m}\psi} + b \norm*{\psi}}
        \nonumber\\
        \label{eq:V-sym-V-bound-implies-U-bound-proof-1}
        &\leq \norm*{\of*{ \textfrac{1}{T}H_0 +V}^m \psi}
        +
        \frac{\of*{1+T}^m - 1}{T^m}
        \of*{ a \norm*{H_0^{m}\psi} + b \norm*{\psi}}
    \end{align}
    for some \(a,b \geq 0\) that are independent of \(T\).
    Here, we have used repeatedly the fact that \(H_0^{m'} A_{i}\) is \(H_0^{m'+1}\)-bounded for \(i=1,2\)  due to \cref{assump:B-V-bound}.

    Now, for $T< \of*{\frac{1}{a}+1}^{1/m} - 1$, we have \(1 - a\of*{\of*{1+T}^m-1} > 0\) and we can rearrange \cref{eq:V-sym-V-bound-implies-U-bound-proof-1}  to obtain
    \begin{align}
        \norm*{H_0^m\psi}
        \leq
        \frac{T^m}{1 - a\of*{\of*{1+T}^m-1}} \norm*{\of*{\textfrac{1}{T}H_0 + V}^m \psi}
        +  \frac{b \of*{\of*{1+T}^m-1}}{1 - a\of*{\of*{1+T}^m-1}}  \norm*{\psi}
        .
    \end{align}
    Hence, we can choose
    \begin{equation}
        a_m^{(T)}
        = \frac{T^m}{1 - a\of*{\of*{1+T}^m-1}}
        , \quad
        b_m^{(T)}
        =  \frac{b \of*{\of*{1+T}^m-1}}{1 - a\of*{\of*{1+T}^m-1}}
        \quad
        \text{whenever}\ T< \of*{\frac{1}{a}+1}^{1/m} - 1
        ,
    \end{equation}
    and therefore
    \begin{equation}
        a_m^{(T)} = \bigo(T^m),
        \quad
        b_m^{(T)} = \bigo(T)
        \quad
        \text{as}\ T \to 0
        .\qedhere
    \end{equation}
\end{proof}

\cref{assump:A-H-symmetric} is an immediate consequence of \cref{assump:B-V-sym} and \(H_0\) being self-adjoint. Regarding \cref{assump:A-H-conj}, we have:

\begin{proposition}\label{thm:V-conj-implies-H-conj}
	 \Cref{assump:B-V-conj} implies \cref{assump:A-H-conj}.
\end{proposition}
Before proving this, we mention a general fact about conjugations and functional calculus, whose proof we explicitly provide here for completeness:
\begin{lemma}\label{thm:commutation-conjugation-dynamics}
    Let \(A\) be a self-adjoint operator and \(J\) an \emph{anti-linear} bounded operator on \(\HH\) with \(J A \subseteq A J\). Then
    \begin{equation}
        J \e^{\iu t A} = \e^{- \iu t A} J
    \end{equation}
    for every \(t\in \R\).
\end{lemma}
\begin{proof}
    We adapt the proof of the implication \(\text{(e)} \implies \text{(h)}\) in~\cite[Thm.~9.41~(iii)]{morettiSpectralTheoryQuantum2017} to the case of anti-linear bounded \(J\).
    By~\cite[Prop.~9.25(d,f)]{morettiSpectralTheoryQuantum2017}, there exists a dense set \(\mathcal{D}_a \subseteq \cinfty\of*{A}\) such that, for every \(t\in \R\) and \(\psi \in \mathcal{D}_a\),
    \begin{equation}
        \e^{\iu t A}\psi
        = \sum_{n=0}^{\infty} \frac{(\iu t)^n}{n!} A^n \psi
        .
    \end{equation}
    Moreover, $JA\subseteq AJ$ implies $J\domof{A^n}\subseteq\domof{A^n}$ for every $n\in\N$.
    Let \(\psi \in \mathcal{D}_a\).
    Then, due to boundedness of \(J\),
    \begin{equation}
        J \e^{\iu t A}\psi
        = \sum_{n=0}^{\infty} J \frac{(\iu t)^n}{n!} A^n \psi
        = \sum_{n=0}^{\infty} \frac{(-\iu t)^n}{n!} A^n J \psi
        = \e^{-\iu t A} J \psi,\quad \forall t\in \R
    \end{equation}
    and \(J \psi \in \mathcal{D}_a\).
    As both \(\e^{\iu t A}\) and \(J\) are bounded and \(\mathcal{D}_a\) is dense, the claimed property follows.
\end{proof}
\begin{proof}[{Proof of \cref{thm:V-conj-implies-H-conj}}]
    By \cref{assump:B-V-conj}, \(JH_0 \subseteq H_0 J\) which implies \(J \cinftyofh \subseteq \cinftyofh\).
    Furthermore, by \cref{thm:commutation-conjugation-dynamics} and \cref{assump:B-V-conj}, \(J\) implements time reversal in the dynamics of \(H_0\):
    \begin{equation}
        J \e^{\pm \iu t H_0} J
        = \e^{\mp \iu t H_0}
        .
    \end{equation}
    Thus,
    \begin{equation}
        J H(s) J
        = J \e^{+\iu s H_0 } \,V\,\e^{-\iu s H_0 } J
        = \e^{-\iu s H_0 } J \,V J\,\e^{+\iu s H_0 }
        = \e^{-\iu s H_0 } \,V \,\e^{+\iu s H_0 }
        = H(-s)
        ,
    \end{equation}
    which is \cref{assump:A-H-conj} for \((H(t))_{t\in\R}\).
\end{proof}

\section{Evolution systems for hyperbolic time-dependent Schrödinger equations}\label{sec:hyperbolic-setting}
In this section, we specialise our discussion to time-dependent Hamiltonians $(H(t))_{t\in\R}$ satisfying \assumpC{}, cf.~\cref{sec:assumption-c}.
As we will see, these assumptions fall into the setting usually referred to as the \emph{hyperbolic case} of time-dependent Schrödinger equations, see e.g.,~\cite[Sec. 5.3]{pazySemigroupsLinearOperators1983} or \cite[Sec. 6.9]{Schaubelt_VI_9_Engel_Nagel_2000}.
This terminology is due to work by Kato \cite{kato1953integration,kato1970linear} and should not be confused with the terminology that is commonly used in the general theory of partial differential equations.
Our goal is to show that \assumpC{} imply \assumpA{}. Here is our roadmap:

\begin{center}
    \begin{tabular}{rcl}
        \cref{assump:C-H-self-adjoint,assump:C-H-bound,assump:C-H-diff,assump:C-H-periodic} & \(\implies\) &
        \cref{assump:pazy:H1,assump:pazy:H2-plus,assump:pazy:H3} below
        \\
        \cref{assump:C-H-self-adjoint,assump:C-H-bound} & \(\implies\) & \cref{assump:A-H-cinf}
        \\
        \(
        \left.
        \begin{tabular}{@{}r@{}}
            \cref{assump:pazy:H1,assump:pazy:H2-plus,assump:pazy:H3},
            \\
            \cref{thm:pazy-regular-solutions}
        \end{tabular}
        \right\}\!
        \)
        & \(\implies\) & \cref{assump:A-U-cinf}
        \\
        \cref{assump:C-H-self-adjoint,assump:C-H-bound,assump:C-H-diff} & \(\implies\) & \cref{assump:A-H-cont}
        \\
        \cref{assump:C-H-self-adjoint,assump:C-H-bound,assump:C-H-diff,assump:C-H-periodic} & \(\implies\) & \cref{assump:A-U-diff}
        \\
        \cref{assump:C-H-bound} & \(\implies\) & \cref{assump:A-H-bound}
        \\
        \cref{assump:C-H-self-adjoint,assump:C-H-bound,assump:C-H-diff,assump:C-H-periodic,assump:C-commutator-bound} & \(\implies\) & \cref{assump:A-U-bound}
        \\
        \cref{assump:C-H-periodic} & \(\iff\) & \cref{assump:A-H-period}
        \\
        \cref{assump:C-H-self-adjoint} & \(\implies\) & \cref{assump:A-H-symmetric}
        \\
        \cref{assump:C-H-conj} & \(\iff\) & \cref{assump:A-H-conj}.
    \end{tabular}
\end{center}

\subsection{\texorpdfstring{Proof of \cref{assump:A-H-cinf,assump:A-H-cont}}{Proof of Properties A.1 and A.3}}

\cref{assump:A-H-cinf,assump:A-H-cont} are quite immediate consequences of \cref{assump:C-H-self-adjoint,assump:C-H-bound,assump:C-H-diff}:

\begin{lemma}\label{thm:C-H-bound-without-H0}
    \Cref{assump:C-H-bound} implies that for every \(m \in \N_0\) and \(t \in \R\), \(H_0^m H(t)\) is \(H_0^{m+1}\)-bounded, i.e.
    \begin{equation}
        \domof{H_0^{m+1}} \subseteq \domof{H_0^m H(t)}
    \end{equation}
    and for any compact interval \(I\) there are \(a_m, b_m \geq 0\) s.t.
    \begin{equation}
        \label{eq:C-H-bound-without-H0}
        \norm*{H_0^m H(t) \psi } \leq a_m \norm*{H_0^{m+1}\psi} + b_m \norm{\psi}
    \end{equation}
    for any \(t \in I\) and \(\psi \in \domof{H_0^{m+1}}\).
\end{lemma}
\begin{proof}
    Let \(\psi \in \domof{H_0^{m+1}}\), i.e. \(\psi \in \domof{H_0}\) with \(H_0\psi \in \domof{H_0^m}\).
    By \cref{assump:C-H-bound}, \(\of*{H(t) - H_0 }\psi \in \domof{H_0^m}\) for any \(t \in \R\) and hence
    \begin{equation}
        H(t) \psi  = \of*{H(t) - H_0} \psi + H_0 \psi \in \domof{H_0^m}
    \end{equation}
    for any \(t \in \R\).
    Therefore, \(\psi \in \domof{H_0^m H(t)}\) and \(\domof{H_0^{m+1}} \subseteq \domof{H_0^m H(t)}\) as \(\psi\) was arbitrary.

    For showing the estimate in \cref{eq:C-H-bound-without-H0}, let \(I\) be a compact interval and let \(\epsilon>0\) and \(\delta \geq 0\) be such that the estimate in \cref{eq:assump:C-H-bound} holds.
    Then,
    \begin{equation}
        \norm*{H_0^m H(t)\psi}
        \leq \norm*{H_0^m \of*{H(t) - H_0}\psi} + \norm{H_0^{m}H_0 \psi}
        \leq (\epsilon +1) \norm{H_0^{m+1}\psi} + \delta \norm{\psi}
    \end{equation}
    for any \(\psi \in \domof{H_0^{m+1}}\) and \(t \in I\).
\end{proof}

\begin{proposition}[{cf.~\cref{thm:invariance-under-V}}] \label{thm:invariance-under-V-2}
    Suppose that \cref{assump:C-H-self-adjoint,assump:C-H-bound} hold. Then, for all \(t \in \R\), \(\cinftyofh \subseteq \domof{H(t)}\) and
    \begin{equation}
        H(t)\cinftyofh \subseteq \cinftyofh,
    \end{equation}
    that is, \cref{assump:A-H-cinf} is satisfied.
\end{proposition}

\begin{proof}
  The claim follows by repeating the proof of \cref{thm:invariance-under-V} with $V$ replaced by $H(t)$, using the statement of \cref{thm:C-H-bound-without-H0} in place of \cref{assump:B-V-bound}.
\end{proof}

\begin{proposition}\label{thm:differentiability-implies-continuity}
    \cref{assump:C-H-self-adjoint,assump:C-H-bound,assump:C-H-diff} imply \cref{assump:A-H-cont}.
\end{proposition}
\begin{proof}
    Due to \cref{assump:C-H-self-adjoint,assump:C-H-bound,thm:invariance-under-V-2}, \(\cinftyofh \subseteq \domof{H(t)}\).
    Moreover, by \cref{assump:C-H-diff}, the map \(t\mapsto H_0^m H(t)\psi\) is differentiable, and therefore continuous, for every \(\psi\in\cinftyofh\) and \(m\in \N_0\).
    Thus \cref{assump:A-H-cont} holds.
\end{proof}

\subsection{Proof of \texorpdfstring{\cref{assump:A-U-cinf}}{Property A.2}}
Next, we verify \cref{assump:A-U-cinf}.
To this end, we appeal to the following theorem, formulated in the general framework of strongly continuous semigroups on Banach spaces and originating in work of Kato~\cite{kato1953integration,kato1970linear}, stated here in the form given in~\cite{pazySemigroupsLinearOperators1983}:

\begin{theorem}[{\cite[Thms.~3.1, 4.3 and 4.6]{pazySemigroupsLinearOperators1983}}]
\label{thm:pazy-regular-solutions}
   Let \(X\) be a Banach space, and let \(A(t)\), for \(0 \leq t \leq t_{\max}\), be the infinitesimal generator of a strongly continuous semigroup \(S_t(s)\), \(s \geq 0\), on \(X\). Let \(Y \subseteq X\) be a dense subspace equipped with a norm \(\norm{\cdot}_Y\) such that
\begin{equation}\label{eq:pazy-Y-dominates-X}
    \norm{\psi} \leq C \norm{\psi}_Y,
    \quad \forall \psi \in Y,
\end{equation}
and such that \(Y\), equipped with the norm \(\norm{\cdot}_Y\), is a Banach space. Assume that the family \((A(t))_{t\in[0,\tmax]}\) satisfies the following conditions (hyperbolic setting):
    \begin{assumptionlist}[label={H.\arabic*},format=\bfseries\upshape]
        \item\label{assump:pazy:H1} The family \((A(t))_{t\in[0,\tmax]}\) is stable, i.e. there exist \(\omega,M \geq 0\) such that
        \begin{gather}
            \label{eq:stable-1}
            \rho\of{A(t)} \supseteq (\omega, \infty),
            \\
            \label{eq:stable-2}
            \norm*{\prod_{j=1}^{k} S_{t_j}(s_j) } \leq M \exp\of*{\omega \sum_{j=1}^{k}s_j}
            ,\quad s_j \geq 0
        \end{gather}
        for every sequence \(0\leq t_1 \leq t_2 \leq \dots \leq t_k \leq \tmax\) and \(k = 1,2,\dots\), and where the product in \cref{eq:stable-2} is ordered from right to left.
        \item\label{assump:pazy:H2-plus} There exists a family \((Q_t)_{t\in[0,\tmax]}\) of Banach space isomorphisms \(Q_t \colon Y \to X\) such that \(t\mapsto Q_t \psi\) is continuously differentiable in \(X\) for every \(\psi \in Y\), and such that
        \begin{equation}
            Q_t A(t) Q_t\inv = A(t) + B(t),
        \end{equation}
        where \(B(t)\) is bounded for every $t$, and the map $t\mapsto B(t)$ is strongly continuous.
        \item\label{assump:pazy:H3} For every \(t\in [0,\tmax]\), \(Y \subseteq \domof{A(t)}\), \(A(t)\) is bounded as an operator from \(Y\) to \(X\), and \(t \mapsto A(t)\) is continuous in the \(B(Y,X)\) norm.
    \end{assumptionlist}
    Then, there exists a unique family of bounded operators \((U(t,s))_{0\leq s \leq t \leq \tmax}\) on \(X\) such that the following properties hold:
    \begin{statements}[format=\bfseries \upshape, label={E.\arabic*}]
        \item\label{thm:pazy:E1} \(\displaystyle \norm*{U(t,s)} \leq M \exp\of*{ \omega (t-s) }\),\quad for \(0\leq s \leq t \leq \tmax\)
        \item\label{thm:pazy:E2} \(\displaystyle \evalat{ \lim_{h \downarrow 0}\frac{U(t+h,s) \psi-U(t,s) \psi}{h}}{t = s} = A(s)\psi \),\quad for \(\psi \in Y\), \(0\leq s \leq \tmax\)
        \item\label{thm:pazy:E3} \(\displaystyle \pdv{}{s} U(t,s)\psi = - U(t,s) A(s)\psi\),\quad for \(\psi \in Y\), \(0\leq s \leq t \leq \tmax\)
        \item\label{thm:pazy:E4} \(\displaystyle U(t,s) Y \subseteq Y\),\quad for \(0\leq s \leq t \leq \tmax\)
        \item\label{thm:pazy:E5} For every \(\psi \in Y\), \((t,s)\mapsto U(t,s)\psi\) is continuous in \(Y\) for \(0\leq s \leq t \leq \tmax\).
    \end{statements}
and such that, for every \(\psi_0 \in Y\), \(\psi(t) = U(t,s)\psi_0\) is the unique solution to the initial value problem
    \begin{equation}\label{eq:non-autonomous-schrödinger-equaiton}
        \begin{cases}
            \odv{}{t} \psi(t) = A(t)\psi(t),
            & \text{for}\ 0 \leq s < t\leq \tmax
            \\
            \psi(s) = \psi_0
        \end{cases}
    \end{equation}
    such that \(t \mapsto \psi(t) \in Y\) is continuous in \(Y\) and continuously differentiable in \(X\) for every \(t\in(s,\tmax]\).

\end{theorem}

Our \cref{assump:pazy:H2-plus} is called $(H_2)^+$ in \cite{pazySemigroupsLinearOperators1983} and is not to be confused with assumption $(H_2)$ there.

\begin{remark}\label{rem:pazy-regular-solutions}
(1) The restriction to compact time domains in Theorem \ref{thm:pazy-regular-solutions} is not essential, and the interval $[0,\tmax]$ can be replaced by $\R$. What matters is that, throughout the proof, ``uniform convergence'' is replaced by ``uniform convergence on compact intervals''. This is also the setting considered in \cite[Sec.~6.9]{Schaubelt_VI_9_Engel_Nagel_2000}.

(2) We will apply the above theorem to \(A(t) = - \iu H^{(T)}(t)\) for \(T>0\) and \(Y = \domof{H_0^m}\), equipped with the graph norm of \(H_0^m\):
\begin{equation}\label{eq:Y-graphnorm}
    \norm{\psi}_Y
    = \norm{\psi}_{H_0^m}
    \coloneqq \norm{H_0^m \psi} + \norm{\psi},
    \quad \psi \in \domof{H_0^m},
\end{equation}
for \(m \in \N\).
As \(H_0^m\) is closed, \(Y\) is a Banach space, and
\begin{equation}
	\norm{\psi}
	\leq \norm{\psi} + \norm{H_0^m \psi}
	= \norm{\psi}_{Y}.
\end{equation}
Initially, this yields an \emph{evolution system} (or \emph{propagator}) \((U^{(T)}(t,s))_{s\leq t}\), i.e.,
\begin{equation}
    \begin{aligned}
        U^{(T)}(t,r) &= U^{(T)}(t,s) U^{(T)}(s,r), \quad \forall r \leq s \leq t, \\
        U^{(T)}(t,t) &= \one, \quad \forall t \in \R, \\
        (s,t) \mapsto U^{(T)}(t,s) &\text{ is jointly strongly continuous.}
    \end{aligned}
\end{equation}

(3) Since each \(H^{(T)}(t)\) is self-adjoint, the propagators \(U^{(T)}(t,s)\) are unitary. Applying the theorem to the family $(-H^{(T)}(-t))_{t\in\R}$ yields propagators for backward time evolution, allowing the extension to $s>t$ and thus obtaining \(U^{(T)}(t,s)^* = U^{(T)}(s,t)\) for all \(s,t \in \R\). Hence, we obtain a unitary propagator \((U^{(T)}(t,s))_{s,t \in \R}\) as in \cite[Def.~p.~282]{reedIIFourierAnalysis1975}, which allows for going forward and backward in time and which is the setting we adopt in the remainder of this section. The unitary family \((U^{(T)}(t))_{t \in \R}\) is then defined by \(U^{(T)}(t) \coloneqq U^{(T)}(t,0)\) for all \(t \in \R\).
\end{remark}

Throughout this subsection we consider arbitrary \(T>0\).
We make the following immediate observation:
\begin{proposition}\label{thm:self-adjoint-implies-stable-C-H1}
    \Cref{assump:C-H-self-adjoint} implies \cref{assump:pazy:H1} of \cref{thm:pazy-regular-solutions} for \(A(t) = -\iu H^{(T)}(t)\) with $\omega=0$ and $M=1$.
\end{proposition}

To apply \cref{thm:pazy-regular-solutions}, we first establish some intermediate lemmas, beginning with a consequence of \cref{assump:C-H-self-adjoint,assump:C-H-bound} that is analogous to the implication of \cref{assump:B-V-sym,assump:B-V-bound} in \cref{thm:dom-H-m}.
\begin{lemma}\label{thm:C-equal-domains-powers}
    Suppose \cref{assump:C-H-self-adjoint,assump:C-H-bound} hold.
    Then
    \begin{equation}
        \domof*{H^{(T)}(t)^m} = \domof*{H_0^m}
        ,\quad \forall t \in \R
        ,\, m\in \N_0
        .
    \end{equation}
\end{lemma}
\begin{proof}
	Let \(t \in \R\).
	We write
	\begin{equation}\label{eq:equal-domains-powers-proof-1}
		V = H^{(T)}(t) - H_0
	\end{equation}
	and by \cref{assump:C-H-self-adjoint} we have
	\begin{equation}
	    \domof{V}
		= \domof{H^{(T)}(t)} \cap \domof{H_0}
		= \domof{H_0}
	\end{equation}
    and \(V\) is symmetric; consequently, $\domof{H^{(T)}(t)}= \domof{H_0+V}$.
    Furthermore, due to \cref{assump:C-H-bound}, \(H_0^m V\) is infinitesimally \(H_0^{m+1}\)-bounded.
	In conclusion, \(H_0\) and \(V\) satisfy \cref{assump:B-V-sym,assump:B-V-bound}.
	Therefore, following the same steps as in \cref{thm:dom-H-m},
	\begin{equation}
		\domof{H^{(T)}(t)^m}
	    = \domof{\of*{H_0+V}^m }
		= \domof{H_0^m}.\qedhere
	\end{equation}
\end{proof}

\begin{lemma}\label{thm:differentiability-powers-H}
        Suppose \Cref{assump:C-H-self-adjoint,assump:C-H-bound,assump:C-H-diff,assump:C-H-periodic} hold true. Then, for every \(m, n \in \N_0\) and \(\psi \in \domof{H_0^{m+n}}\), the map \(t\mapsto\odv{}{t} H_0^m H^{(T)}(t)^n \psi\) exists, is continuous, and equals \(H_0^m \odv{}{t}H^{(T)}(t)^n \psi\).
    Furthermore, \(H_0^m \odv{}{t} H^{(T)}(t)^n\) is \(H_0^{m+n}\)-bounded.

    In particular, \(t \mapsto H^{(T)}(t)^n\psi\) is continuously differentiable for every \(\psi \in \domof{H_0^{n}}\).
\end{lemma}
\begin{proof}
    We prove the claim by induction over \(n \in \N_0\). For \(n = 0\), the claim follows trivially.

    Now, let \(n \in \N_0\) be arbitrary and assume the claim holds for \(n\), i.e., \(\odv{}{t} H_0^m H^{(T)}(t)^n \psi\) exists, is continuous, equals \( H_0^m \odv{}{t}H^{(T)}(t)^n \psi\), and \(H_0^m (\sodv{}{t}H^{(T)}(t)^n)\) is \(H_0^{m+n}\)-bounded for every \(\psi \in \domof{H_0^{m+n}}\) and \(m\in \N_0\).
    Let \(\psi \in \domof{H_0^{m+n+1}}\).
    Then \(H^{(T)}(t)^n \psi \in \domof{H_0^{m+1}} \) by \cref{thm:invariance-under-V}.
    Hence, we have    \begin{equation}\label{eq:differentiability-powers-H-proof-1}
        H_0^m \frac{H^{(T)}(t+h) - H^{(T)}(t)}{h} H^{(T)}(t)^n \psi
        \to H_0^m \of[\Big]{\sodv{}{t}H^{(T)}(t)} H^{(T)}(t)^n \psi
        ,\quad \text{as}\ h \to 0
    \end{equation}
    by \cref{assump:C-H-diff}.

    Next, we show
    \begin{equation}
        H_0^m H^{(T)}(t+h) \frac{H^{(T)}(t+h)^{n} - H^{(T)}(t)^{n}}{h}\psi
        \to H_0^m H^{(T)}(t) \odv{}{t}H^{(T)}(t)^{n}\psi
    \end{equation}
    as \(h \to 0\):
    \begin{align}
        \MoveEqLeft \nonumber
        \norm*{
            H_0^m H^{(T)}(t+h) \frac{H^{(T)}(t+h)^{n} - H^{(T)}(t)^{n}}{h}\psi
            - H_0^m H^{(T)}(t) \odv{}{t}H^{(T)}(t)^{n}\psi
        }\leq
        \nonumber\\
        &\leq
        \begin{aligned}[t]
            &\norm*{ H_0^m \of{ H^{(T)}(t+h) - H^{(T)}(t) } \odv{}{t}H^{(T)}(t)^n \psi }
            \\
            &+ a_m \norm*{H_0^{m+1} \of*{\frac{H^{(T)}(t+h)^{n} - H^{(T)}(t)^{n}}{h} - \sodv{}{t}H^{(T)}(t)^{n} } \psi }
            \\
            &+ b_m \norm*{\of*{\frac{H^{(T)}(t+h)^{n} - H^{(T)}(t)^{n}}{h} - \sodv{}{t}H^{(T)}(t)^{n} } \psi }
        \end{aligned}
        \nonumber\\
        \label{eq:differentiability-powers-H-proof-2}
        &\to 0
        ,
    \end{align}
    where we used that, by assumption and \cref{thm:invariance-under-V}, \(\odv{}{t}H^{(T)}(t)^n \psi \in \domof{H_0^{m+1}}\); moreover, we used the relative boundedness constants from \cref{eq:C-H-bound-without-H0}, which in the current case are uniform in $t\in\R$ and $T>0$ due to periodicity of $t\mapsto H^{(T)}(t)$ as in \cref{assump:C-H-periodic}.
    The first summand converges due to \cref{thm:differentiability-implies-continuity}, while the second and third summand converge due to the assumption of the induction step.
    Combining \cref{eq:differentiability-powers-H-proof-1,eq:differentiability-powers-H-proof-2} we see that \(H_0^m  H^{(T)}(t)^{n+1}\psi\) is differentiable with
    \begin{equation}\label{eq:differentiability-powers-H-proof-3}
        \odv{}{t} H_0^m H^{(T)}(t)^{n+1} \psi
        = H_0^m \of*{\of[\Big]{\sodv{}{t}H^{(T)}(t)} H^{(T)}(t)^n \psi + H^{(T)}(t) \odv{}{t} H^{(T)}(t)^n\psi}
        .
    \end{equation}

    We need to verify that the derivative is continuous. Since, by \cref{assump:C-H-diff}, \(t\mapsto H_0^m H^{\prime(T)}(t)\psi\) is continuous and $H_0^m H^{\prime(T)}(t)\psi$ is $H_0^{m+1}$-bounded with uniform relative boundedness constants $a_m'$ and $b_m'$, we obtain for the first summand
    \begin{align}
        \MoveEqLeft
        \norm*{
            H_0^m H^{\prime(T)}(t) H^{(T)}(t)^n  \psi
            - H_0^m H^{\prime(T)}(s) H^{(T)}(s)^n  \psi
        }\leq
        \nonumber\\
        & \leq
        \begin{aligned}[t]
            &\norm*{ H_0^m \of*{ \of[\Big]{\sodv{}{t}H^{(T)}(t)} - \of[\Big]{\sodv{}{s}H^{(T)}(s)} } H^{(T)}(s)^n \psi}
            \\&
            + a_m' \norm*{ H_0^{m+1}  \of*{ H^{(T)}(t)^n -H^{(T)}(s)^n } \psi}
            + b_m' \norm*{ \of*{ H^{(T)}(t)^n -H^{(T)}(s)^n } \psi}
        \end{aligned}
        \nonumber\\
        \label{eq:differentiability-powers-H-proof-4}
        &\to 0
        ,\quad \text{as}\ t \to s
        .
    \end{align}
    An analogous argument using the induction assumption works for the second summand in the derivative \eqref{eq:differentiability-powers-H-proof-3}.

    Finally, the \(H_0^{m+n+1}\)-boundedness of the derivative \eqref{eq:differentiability-powers-H-proof-3} is proven separately for the two summands and follows also immediately from the induction assumption and \cref{assump:C-H-bound,assump:C-H-diff,thm:C-H-bound-without-H0}.
\end{proof}

We can now prove \cref{assump:pazy:H2-plus,assump:pazy:H3}:
\begin{proposition}\label{thm:C-H2}
    Suppose \cref{assump:C-H-self-adjoint,assump:C-H-bound,assump:C-H-diff,assump:C-H-periodic} hold true.
    Let \(Y = \domof{H_0^m}\) be equipped with the graph norm of \(H_0^m\) in \cref{eq:Y-graphnorm}, and \(A(t) = - \iu H^{(T)}(t)\).
    Then \cref{assump:pazy:H2-plus} of \cref{thm:pazy-regular-solutions} is satisfied with \(Q_t = H^{(T)}(t)^m-\iu\).
\end{proposition}
\begin{proof}
    Let \(m \in \N_0\) and \(Q_t\colon \domof{H^{(T)}(t)^m} \to \HH\) given by \(Q_t \coloneqq H^{(T)}(t)^m - \iu\) for every \(t \in \R\).
    As \(H^{(T)}(t)\) is self-adjoint, so is \(H^{(T)}(t)^m\), hence its spectrum is real and \(Q_t\) is bijective.
    By \cref{thm:C-equal-domains-powers}, \(\domof{H_0^m} = \domof{H^{(T)}(t)^m}\); besides, \(H_0^m\) and \(H^{(T)}(t)^m\) are both self-adjoint, therefore they are relatively bounded with respect to each other \cite[Lem.~6.2]{teschlMathematicalMethodsQuantum2009}, that is, their graph norms are equivalent. We equip $Y=\domof{H_0^m}$ with the graph norm in \cref{eq:Y-graphnorm}.
    Then \(Q_t\colon Y\to\HH\) and \(Q_t\inv\colon \HH \to Y\) are clearly bounded.

    By construction, \(Q_t\) and \(A(t)\) commute strongly and hence
    \begin{equation}
        Q_t A(t) Q_t\inv
        = A(t)Q_t Q_t\inv
        = A(t)
        .
    \end{equation}

    By \cref{thm:differentiability-powers-H}, \(t\mapsto H^{(T)}(t)^m\psi\) is continuously differentiable for every $\psi\in Y$, and so is \(t\mapsto Q_t\psi =  H^{(T)}(t)^m\psi -\iu \psi\).
\end{proof}

\begin{proposition}\label{thm:C-unform-continuity}
    Let \(m \in \N_0\).
    Suppose \cref{assump:C-H-self-adjoint,assump:C-H-bound,assump:C-H-periodic} hold true, and \(t\mapsto H(t)\psi\) is continuously differentiable in the graph norm of \(H_0^m\) for every \(\psi \in \domof{H_0^{m+1}}\).
    Then, \(-\iu H^{(T)}(t)\) is (Lipschitz) continuous in the \(B\of*{\domof{H_0^{m+1}},\domof{H_0^{m}}}\) norm, when \(\domof{H_0^{m}}\) and \(\domof{H_0^{m+1}}\) are equipped with the graph norms of \(H_0^m\) and \(H_0^{m+1}\), respectively.
\end{proposition}
\begin{proof}
    Let \(t \in \R\) and \(\psi \in \domof{H_0^{m+1}}\).
    By \cref{assump:C-H-self-adjoint}, \(\domof{H_0^{m+1}} \subseteq \domof{H_0} = \domof{-\iu H^{(T)}(t)}\).
    Hence, it follows from \cref{assump:C-H-bound,thm:C-H-bound-without-H0} that \(A(t) = -\iu H^{(T)}(t)\colon \domof{H_0^{m+1}}\to \domof{H_0^{m}}\) is bounded, when \(\domof{H_0^{m}}\) and \(\domof{H_0^{m+1}}\) are equipped with the graph norms of \(H_0^m\) and \(H_0^{m+1}\), respectively.

    As \(t\mapsto A(t)\) is continuously strongly differentiable, its strong derivative \(\sodv{}{t}A(t)\) is bounded, as follows from a corollary of the uniform boundedness principle \cite[Cor.~2 in Ch.~II.1,~p.68]{yosidaFunctionalAnalysis1995}.
    Moreover, since the map \(t\mapsto \of{\sodv{}{t}A(t)}\psi\) is continuous and periodic by \cref{assump:C-H-periodic}, we have
    \begin{equation}
        \sup_{t \in \R} \norm*{\of*{\sodv{}{t}A(t)}\psi}_{H_0^{m}} < \infty,
    \end{equation}
    for every \(\psi \in \domof{H_0^{m+1}}\).
    Combining these two facts, another application of the uniform boundedness principle yields
    \begin{equation}
        M\coloneqq \sup_{t \in \R} \norm*{\sodv{}{t}A(t)}_{m+1,m}
        <\infty
        ,
    \end{equation}
    where \(\norm{}_{m+1,m}\) denotes the norm of \(B(\domof{H_0^{m+1}},\domof{H_0^{m}})\).
    In conclusion,
    \begin{align}
        \norm*{\of*{A(t)-A(s)}\psi}_{H_0^{m}}
        &\leq \int_s^t \norm*{\of*{\sodv{}{\tau}A(\tau)}\psi}_{H_0^{m}} \dd \tau
        \nonumber\\
        &\leq \int_s^t \norm*{\sodv{}{\tau}A(\tau)}_{m+1,m} \norm{\psi}_{H_0^{m+1}} \dd \tau
        \nonumber\\
        &\leq \abs{t-s} M \norm{\psi}_{H_0^{m+1}}
    \end{align}
    for \(t \geq s\) implying that \(A(t)\) is (Lipschitz) continuous in the \(B(\domof{H_0^{m+1}},\domof{H_0^{m}})\) norm.
\end{proof}
As a direct corollary, we get \cref{assump:pazy:H3} of \cref{thm:pazy-regular-solutions}:
\begin{corollary}\label{thm:C-H3}
    Let \(m \in \N\), \(Y = \domof{H_0^m}\) equipped with the graph norm of \(H_0^m\) in \cref{eq:Y-graphnorm}, and \(A(t) = - \iu H^{(T)}(t)\).
	Suppose \cref{assump:C-H-self-adjoint,assump:C-H-bound,assump:C-H-periodic} hold true, and \(t\mapsto H(t)\psi\) is continuously differentiable for every \(\psi \in Y \subseteq \dom(H(t))\).
	Then \cref{assump:pazy:H3} of \cref{thm:pazy-regular-solutions} is satisfied.
\end{corollary}
To summarise, \cref{thm:self-adjoint-implies-stable-C-H1,thm:C-H2,thm:C-H3} show that \cref{assump:C-H-self-adjoint,assump:C-H-bound,assump:C-H-diff,assump:C-H-periodic} imply that all assertions of \cref{thm:pazy-regular-solutions} are satisfied by setting \(A(t) = - \iu H^{(T)}(t)\) and \(Y = \domof{H_0^m}\) for arbitrary \(m \in \N\). This will allow us to use \cref{thm:pazy-regular-solutions} to prove \cref{assump:A-U-cinf}.

\begin{proposition}\label{thm:C-existence-invariance}
    Suppose \cref{assump:C-H-self-adjoint,assump:C-H-bound,assump:C-H-diff,assump:C-H-periodic} hold true.
    Then the statement of \cref{thm:pazy-regular-solutions,rem:pazy-regular-solutions} holds true with \(A(t) = -\iu H^{(T)}(t)\) and \(Y = \domof{H_0^m}\) for every \(m \in \N\).
    In particular, there exists a unique family of unitary operators \((U^{(T)}(t,s))_{s,t\in\R}\) such that \(\psi(t) = U^{(T)}(t,s)\psi_0 \in \domof{H_0}\) is the unique \(\domof{H_0}\)-valued solution of \cref{eq:non-autonomous-schrödinger-equaiton},
    \begin{equation}\label{eq:C-invariance-dom-H0-m}
        U^{(T)}(t,s)\domof{H_0^m} \subseteq \domof{H_0^m}
        ,\quad \forall m \in \N,\, \forall s,t\in\R
        ,
    \end{equation}
    and \(t\mapsto H_0^m \psi(t)\) is continuous in the graph norm of \(H_0^m\) for any \(m \in \N\).
    Consequently, \cref{eq:non-autonomous-schrödinger-equaiton} holds for \(\psi_0\in \cinftyofh\) and \cref{assump:A-U-cinf} is satisfied by \(U^{(T)}(t)\coloneqq U^{(T)}(t,0)\).
\end{proposition}
\begin{proof}
    Assume \cref{assump:C-H-self-adjoint,assump:C-H-bound,assump:C-H-diff,assump:C-H-periodic} hold true, and let \(m \in \N\).
	By \cref{thm:self-adjoint-implies-stable-C-H1,thm:C-H2,thm:C-H3}, the statement of \cref{thm:pazy-regular-solutions,rem:pazy-regular-solutions} holds true with \(A(t) = -\iu H^{(T)}(t)\) and \(Y = \domof{H_0^m}\) for every \(m \in \N\): for every \(m\) there exists a unique family \((U^{(T)}_m(t,s))_{s,t\in\R}\) of unitary operators on $\HH$ generating the unique solution of the initial value problem \cref{eq:non-autonomous-schrödinger-equaiton} on \(\domof{H_0^m}\) and satisfying \cref{thm:pazy:E1,thm:pazy:E2,thm:pazy:E3,thm:pazy:E4,thm:pazy:E5}.

    Let \(m' \geq m\).
    If $\psi_0 \in \domof{H_0^{m'}}$ then $\psi_0\in\domof{H_0^m}$ and $t\mapsto \psi_m(t) \coloneqq  U^{(T)}_m(t,s)\psi_0$ is the unique $\domof{H_0^{m}}$-valued solution of the initial value problem \eqref{eq:non-autonomous-schrödinger-equaiton} in \cref{thm:pazy-regular-solutions} with $\psi_m(s)=\psi_0$, and $t\mapsto \psi_{m'}(t) \coloneqq  U^{(T)}_{m'}(t,s)\psi_0$ is the unique $\domof{H_0^{m'}}$-valued solution of the same initial value problem.
    Since \(\domof{H_0^{m'}} \subseteq \domof{H_0^m}\), any \(\domof{H_0^{m'}}\)-valued solution is also \(\domof{H_0^m}\)-valued, and continuity with respect to the stronger norm on \(\domof{H_0^{m'}}\) implies continuity in \(\domof{H_0^m}\).
    Hence $\psi_{m'}$ is also \(\domof{H_0^m}\)-valued and therefore coincides with $\psi_m$, and both of them are actually $\domof{H_0^{m'}}$-valued.
    Thus $U^{(T)}_m(t,s)$ preserves $\domof{H_0^{m'}}$ and coincides with $U^{(T)}_{m'}(t,s)$ on this subspace whenever $m\leq m'$, and we have a unique unitary family of operators \((U^{(T)}(t,s))_{s,t\in\R}\) on $\HH$, where \(U^{(T)}(t,s)\coloneqq U^{(T)}_1(t,s)\), generating the unique solution of the initial value problem \eqref{eq:non-autonomous-schrödinger-equaiton} on every \(\domof{H_0^m}\), \(m\in\N\), and satisfying \cref{thm:pazy:E1,thm:pazy:E2,thm:pazy:E3,thm:pazy:E4,thm:pazy:E5}.

    Consequently,
    \begin{equation}
        U^{(T)}(t,s)\cinftyofh
        = \bigcap_{m \in \N} U^{(T)}(t,s)\domof{H_0^m}
        \subseteq \bigcap_{m \in \N} \domof{H_0^m}
        = \cinftyofh, \quad \forall s,t\in\R.
    \end{equation}
    and choosing $U^{(T)}(t)\coloneqq U^{(T)}(t,0)$ concludes the proof of \cref{assump:A-U-cinf}.
\end{proof}

\subsection{Proof of \texorpdfstring{\cref{assump:A-U-diff}}{Property A.4}}

We can upgrade the differentiability of \(U(t)\psi\) in \(\HH\) for any \(\psi\in\cinftyofh\), which is guaranteed by \cref{thm:C-existence-invariance}, to differentiability of \(H_0^m U(t)\psi\) by using that the derivative of \(U(t)\psi\) is continuous in the graph norm of \(H_0^m\).
\begin{proposition}
    Under assumption of \cref{assump:C-H-self-adjoint,assump:C-H-bound,assump:C-H-diff,assump:C-H-periodic}, \cref{assump:A-U-diff} holds.
\end{proposition}
\begin{proof}
    Let \(\psi \in \cinftyofh \subseteq \domof{H_0^{m+1}}\) and \(T > 0\).
    Endow \(\domof*{H_0^{m+1}}\) and \(\domof{H_0^{m}}\) with the graph norms of \(H_0^{m+1}\)and \(H_0^{m}\), respectively.
    By \cref{thm:C-existence-invariance,thm:C-unform-continuity}, \(s \mapsto -\iu H^{(T)}(s) \in B\of*{\domof{H_0^{m+1}}, \domof{H_0^{m}}}\) and \(s \mapsto U^{(T)}(s)\psi \in \domof{H_0^{m+1}}\) are continuous.
    Therefore, also \(s\mapsto -\iu H^{(T)}(s) U^{(T)}(s) \psi \in \domof{H_0^m}\) is continuous.

    Let \(\epsilon > 0\).
    By continuity, we can choose \(\delta > 0\) s.t.
    \begin{equation}
        \norm*{H^{(T)}(s) U^{(T)}(s) \psi - H^{(T)}(t) U^{(T)}(t) \psi}_{H_0^m} < \epsilon
        ,\quad \text{for } \abs{s-t}<\delta
        .
    \end{equation}
    If \(0 < \abs{h} < \delta\),
    \begin{align}
        \MoveEqLeft
        \norm*{ \of*{ \frac{U^{(T)}(t+h) - U^{(T)}(t)}{h} - (-\iu H^{(T)}(t) U^{(T)}(t)) }\psi }_{H_0^m}
        \nonumber\\
        &
        \leq \frac{1}{h} \int_{t}^{t+h} \norm*{\of*{ H^{(T)}(s) U^{(T)}(s) - H^{(T)}(t) U^{(T)}(t) } \psi }_{H_0^m} \dd{s}
        \nonumber\\
        &< \epsilon
        ,
    \end{align}
    where we used that \(\odv{}{s} U^{(T)}(s)\psi = -\iu H^{(T)}(s) U^{(T)}(s)\psi\).
    In conclusion, \(\frac{U^{(T)}(t+h) - U^{(T)}(t)}{h} \psi \to -\iu H^{(T)}(s) U^{(T)}(s)\psi\) converges in the \(H_0^m\)-graph norm and, in particular, \(\odv{}{t} H_0^m U^{(T)}(t)\psi\) exists and equals \(-\iu H_0^m H^{(T)}(t) U^{(T)}(t)\psi\).
\end{proof}

\subsection{Proof of \texorpdfstring{\cref{assump:A-U-bound}}{Property A.6}}
It remains to verify that \cref{assump:A-U-bound} holds assuming that \assumpC{} hold.
Given a compact interval $I$, we need to show that
\begin{equation}
    \norm{H_0^m U^{(T)}(t) \psi }
    \leq a'^{(T)}_m \norm{H_0^{m+k'_m}\psi} + b'^{(T)}_m \norm{\psi},
\end{equation}
holds for all \(\psi \in \cinftyofh\) and \(t \in I\), with constants satisfying \(a'^{(T)}_m = \bigo(1)\) and \(b'^{(T)}_m = \bigo(1)\) as \(T \to 0\). As we will see, the commutator estimate (\cref{assump:C-commutator-bound}) plays a crucial role in establishing the asymptotic behaviour of these coefficients.

\begin{proposition}\label{thm:C-U-bound}
    Suppose \(\of*{H(t)}_{t \in \R}\) satisfies \cref{assump:C-H-self-adjoint,assump:C-H-bound,assump:C-H-diff,assump:C-H-periodic,assump:C-commutator-bound}.
    Then \(\of*{U^{(T)}(t)}_{t \in \R}\) satisfies \cref{assump:A-U-bound}.
\end{proposition}
\begin{proof}
    Due to \cref{thm:C-existence-invariance} we have
    \begin{equation}
        U^{(T)}(t) \domof{H_0^m} \subseteq \domof{H_0^m}, \quad \forall m\in\N_0,\, \forall t\in\R,
    \end{equation}
    and accordingly,
    \begin{equation}
        \domof{H_0^m} \subseteq \domof{H_0^m U^{(T)}(t)}, \quad \forall m\in\N_0,\, \forall t\in\R
        .
    \end{equation}
    Fix a compact interval $I\subseteq\R$. As \(U^{(T)}(t)\) is bounded and \(H_0^m\) is self-adjoint, \(H_0^m U^{(T)}(t)\) and \(H_0^m\) are both closed; thus \(H_0^m U^{(T)}(t)\) is \(H_0^m\)-bounded and, by strong continuity, there exist \(a^{\prime (T)}_{m},b^{\prime (T)}_{m}\geq 0\) such that
    \begin{equation}
        \norm*{H_0^m U^{(T)}(t)\psi}
        \leq a^{\prime (T)}_{m} \norm*{H_0^m \psi} + b^{\prime (T)}_{m} \norm*{\psi},\quad \forall t\in I
        .
    \end{equation}
    We now verify that we can choose these coefficients in such a way that \(a^{\prime (T)}_{m} = \bigo(1)\) and \(b^{\prime (T)}_{m}=\bigo(1)\) as \(T \to 0\).
    For this purpose, we write
    \begin{align}
        \iu \odv{}{t} H_0^m U^{(T)}(t) \psi
        &= H_0^m H^{(T)}(t) U^{(T)}(t) \psi
        \nonumber\\
        \label{eq:thm:C-U-bound-proof-1}
        &= H^{(T)}(t) H_0^m U^{(T)}(t) \psi + \comm{H_0^m}{H^{(T)}(t)} U^{(T)}(t)\psi,
    \end{align}
    which is a linear differential equation in \(t\mapsto H_0^mU^{(T)}(t)\psi\), with a homogeneous term generated by \(H^{(T)}(t)\) plus an inhomogeneous term which is controlled by \(\comm{H_0^m}{H^{(T)}(t)}\).
    Applying the variation-of-constants formula~\cite[Eq.~(5.2) in Ch.~5]{pazySemigroupsLinearOperators1983} to
\cref{eq:thm:C-U-bound-proof-1} and adding the initial-value condition $H_0^m U^{(T)}(0) \psi=H_0^m\psi$, we obtain
    \begin{equation}
        H_0^m U^{(T)}(t) \psi
        = U^{(T)}(t) H_0^m \psi -i \int_0^t U^{(T)}(t)U^{(T)}(s)^* \comm{H_0^m}{H^{(T)}(s)} U^{(T)}(s) \psi \dd{s}
        .
    \end{equation}
    Taking norms and using \cref{assump:C-commutator-bound},
    \begin{align}
        \norm{H_0^m U^{(T)}(t)\psi}
        &\leq \norm{H_0^m\psi} + \sgn(t)\int_0^t \norm*{ \comm{H_0^m}{H^{(T)}(s)} U^{(T)}(s) \psi } \dd{s}
        \nonumber\\
        &\leq \norm{H_0^m\psi} + \sgn(t)\int_0^t \of*{ \tilde{a}_{m} \norm{H_0^mU^{(T)}(s)\psi} + \tilde{b}_{m} \norm{\psi} } \dd{s}
        ,
    \end{align}
    where the relative boundedness constants $\tilde{a}_{m},\,\tilde{b}_{m}$ can be chosen uniformly on all $\R$ due to the periodicity assumption \cref{assump:C-H-periodic}. To make the notation clearer, we define \(y_m^{(T)}(t) \coloneqq \norm{H_0^mU^{(T)}(t)\psi}\), and the non-decreasing function
    \begin{equation}
        \alpha_m(t)
        \coloneqq \norm{H_0^m \psi} + \tilde{b}_{m} t\,\norm{\psi}
        .
    \end{equation}
    We then have an integral inequality on \(y_m^{(T)}(t)\):
    \begin{equation}
        y^{(T)}_m(t)
        \leq \alpha_m(t) + \int_0^t \tilde{a}_{m} y^{(T)}_m(s) \dd{s}
        .
    \end{equation}
    We can now apply Grönwall's inequality \cite{gronwall1919note} to obtain
    \begin{equation}
        \norm{H_0^mU^{(T)}(t)\psi} = y^{(T)}_m(t)
        \leq \alpha_m(t) \exp(\tilde{a}_{m} t)
        \leq a'^{(T)}_m \norm*{H_0^m \psi} + b'^{(T)}_m\norm{\psi},\quad \forall t\in I
        ,
    \end{equation}
    with explicit relative boundedness coefficients
    \begin{align}
        a^{\prime (T)}_m
        \coloneqq \exp(\tilde{a}_m \tmax)
        ,
        \quad
        b^{\prime (T)}_{m}
        \coloneqq  \tilde{b}_m \tmax \exp(\tilde{a}_m \tmax)
        ,
    \end{align}
    where $\tmax>0$ is the smallest possible number such that $I\subseteq [-\tmax,\tmax]$. Then $a^{\prime (T)}_m$ and $b^{\prime (T)}_m$ are clearly independent of $T$, hence \cref{assump:A-U-bound} follows.
\end{proof}

\section{Examples}\label{sec:examples}
We illustrate the scope of our results in the previous sections with two representative examples: the quantum Rabi model in the interaction picture and the periodically driven harmonic oscillator.
Both of these examples satisfy either \assumpB{} or \assumpC{}, respectively, and thus we can construct effective Hamiltonians as discussed in \cref{sec:Heff}.

\subsection{Quantum Rabi model}\label{sec:quantum-Rabi}
The single-mode quantum Rabi model, see e.g.~\cite[Sec.~2.4]{larson2024jaynes}, is defined by
\begin{equation}
	H_{\mathrm{R}}
	= \frac{\omega + \Delta}{2} \sigma_z \otimes \one_{\mathrm{b}}
    	+ \omega \one_{\mathrm{s}} \otimes \nn
    	+ g \of*{\sigma_- + \sigma_+} \otimes \of*{\an + \ad}
\end{equation}
on $\HH \coloneqq \C^2 \otimes \Ltwo{\R}$.
Here,
\begin{gather}
    \an = \frac{1}{\sqrt{2}} \of*{\odv{}{x} + x}
    \quad \text{and} \quad
    \ad = \frac{1}{\sqrt{2}} \of*{-\odv{}{x} + x}
\end{gather}
are the usual bosonic annihilation and creation operators and
\begin{equation}
    \nn = \ad \an = \frac{1}{2}\of*{-\odv[2]{}{x} + x^2 - 1}
\end{equation}
the corresponding number operator on their natural domains, and \(\one_{\mathrm{b}}\) is the identity on the Hilbert space \(\Ltwo{\R}\).
The matrices
\begin{equation}
	\sigma_z
	= \begin{pmatrix}
	1 & 0 \\
	0 & -1
	\end{pmatrix}
	,\quad
	\sigma_+
	= \begin{pmatrix}
	0 & 1 \\
	0 & 0
	\end{pmatrix}
	,\quad
	\sigma_-
	= \begin{pmatrix}
	0 & 0 \\
	1 & 0
	\end{pmatrix}
\end{equation}
and the identity matrix \(\one_{\mathrm{s}}\) act as linear operators on \(\C^2\).
The parameters \(\omega, g > 0\) and \(\Delta \in \R\) respectively describe the frequency of the bosonic field, the coupling strength, and the detuning between the field and the spin system.
The domain of \(H_{\mathrm{R}}\) is
\begin{equation}
    \domof{H_{\mathrm{R}}}
    = \C^2 \otimes \domof{\nn}
    \subseteq \HH
    ,
\end{equation}
where \(\domof{\nn}\) is the domain of \(\nn\).
On this domain, we split \(H_{\mathrm{R}}\) as
\begin{align}
    H_{\mathrm{R}}
	&= \frac{1}{T} H_0 + V
	, \\
	H_0
	&= \pi \of*{ \frac{1}{2}\sigma_z  \otimes \one_{\mathrm{b}} + \one_{\mathrm{s}} \otimes \nn}
	, \\
	V
	&= g \of*{\sigma_- + \sigma_+} \otimes \of*{\an + \ad} + \frac{\Delta}{2} \sigma_z \otimes \one_{\mathrm{b}}
	, \\
	T
	&= \frac{\pi}{\omega}
\end{align}
and show that \(H_0\) and \(V\) satisfy \assumpB{}.

As \(\nn\) is self-adjoint on \(\domof{\nn}\), see \cite[Example~X.6.2]{reedIIFourierAnalysis1975}, and \(\sigma_z\), \(\one_{\mathrm{s}}\), \(\one_{\mathrm{b}}\) are bounded and self-adjoint, \(H_0\) is self-adjoint on \(\C^2 \otimes\domof{\nn}\).
The matrices \(\sigma_-\) and \(\sigma_+\), as well as the operators \(\ad\) and \(\an\) on \(\domof{\nn^{1/2} }\), are mutually adjoints.
Consequently, \(V\) is symmetric on \(\C^2 \otimes \domof{\nn^{1/2} }\) and hence on \(\C^2 \otimes \domof{\nn} \subseteq \C^2 \otimes \domof{\nn^{1/2} }\).
Furthermore, for every \(c \in \R\), the operator \((\nn + c)^m (\an+\ad)\) is relatively bounded with respect to \(\nn^{m+\frac12}\). By interpolation, it is therefore infinitesimally relatively bounded with respect to \(\nn^{m+1}\). These properties follow from the canonical commutation relations for \(\an\) and \(\ad\).
The operators \(\sigma_z \otimes \one_{\mathrm{b}}\) and \(\of*{\sigma_- + \sigma_+}\otimes \one_{\mathrm{b}}\) are clearly bounded and commute (strongly) with \(\one_{\mathrm{s}}\otimes \nn\).
Consequently, \(H_0^m V\) is infinitesimally relatively bounded with respect to \(H_0^{m+1}\) for any \(m \in \N_0\).
This shows that \cref{assump:B-V-bound} holds and, accordingly, \cref{assump:B-V-sym} holds by the Kato--Rellich theorem.

In the interaction picture, the Hamiltonian takes the form
\begin{equation}
    H(t)
    = \e^{+ \iu t H_0} V \e^{-\iu t H_0}
    = g \of*{\sigma_- \otimes \ad + \sigma_+ \otimes \an + \e^{\iu 2\pi t} \sigma_+ \otimes \ad + \e^{- \iu 2\pi t} \sigma_- \otimes \an} + \frac{\Delta}{2} \sigma_z \otimes \one_{\mathrm{b}}
    .
\end{equation}
It is evident that \(H(t+1) = H(t)\), i.e. that \cref{assump:B-V-period} holds.

To verify \cref{assump:B-V-conj}, we define a conjugation with respect to the eigenbasis of \(H_0\) by setting
\begin{equation}
    J v_\pm \otimes \varphi_n
    = v_\pm \otimes \varphi_n
\end{equation}
and extending it to the whole Hilbert space via anti-linear extension.
It is evident that \(J\) is a conjugation and commutes with \(V\) and \(H_0\) by direct computation using that \(V\) and \(H_0\) both have only real matrix elements in this basis.

Having verified \assumpB{}, we can apply our theory and compute \(\Heff{L}{T}\) at any desired order. For example, setting \(\Delta = 0\) for simplicity, the first two effective Hamiltonians read
\begin{align}
    \Heff{0}{T}
    &= g \of*{\sigma_- \otimes \ad + \sigma_+ \otimes \an},
    \\
    \Heff{1}{T}
    &= g \of*{\sigma_- \otimes \ad + \sigma_+ \otimes \an}
    + \frac{g^2}{2\omega} \of*{\sigma_z \otimes \nn - \sigma_- \sigma_+ \otimes \one_{\mathrm{b}} - \sigma_z \otimes \of*{\an^2 + \of*{\ad}^2 } }
    .
\end{align}
Transforming back from the interaction picture to the original frame, we obtain the Hamiltonian
\begin{equation}
\label{eq:BS-schroedinger-picture}
    \e^{-\iu t/T H_0} \Heff{1}{T} \e^{+\iu t/T H_0} + \frac{1}{T} H_0
    =
    \begin{aligned}[t]
        &\frac{\omega }{2} \sigma_z \otimes \one_{\mathrm{b}}
    	+ \omega \one_{\mathrm{s}} \otimes \nn
    	+ g \of*{\sigma_- \otimes \ad + \sigma_+ \otimes \an }
        \\
        &+ \frac{g^2}{2\omega} \of*{  \sigma_z \otimes \nn - \sigma_- \sigma_+ \otimes \one_{\mathrm{b}} - \sigma_z \otimes \of*{\e^{+\iu 2 \omega t}\an^2 + \e^{-\iu 2 \omega t}\of*{\ad}^2 } }
    \end{aligned}
\end{equation}
on \(\cinftyofh= \C^2\otimes \mathcal{S}(\R)\), where \(\mathcal{S}(\R)\) is the Schwartz space. The first three terms on the right-hand side can be identified as the \emph{Jaynes--Cummings Hamiltonian} \cite[Sec.~2.1]{larson2024jaynes}.
The correction term with prefactor \(\frac{g^2}{2\omega}\) consists of two components with different physical interpretations: the diagonal component (with respect to the above basis) is known as \emph{Bloch--Siegert shift} \cite[Eq.~(2.219)]{larson2024jaynes}, see \cite[Sec.~2.4.1]{larson2024jaynes} for a discussion of its physical significance, while the non-diagonal component corresponds to spin-dependent squeezing ~\cite{paingConditionalSqueezingInduced2026}.
We note, however, that the form of the correction term associated with the Bloch--Siegert shift is not uniform across the physics literature (see, e.g., \cite{beaudoinDissipationUltrastrongCoupling2011,rossattoSpectralClassificationCoupling2017,forn-diazUltrastrongCouplingRegimes2019}).

The effective Hamiltonian corresponding to \(L = 2\) has an additional correction term:
\begin{equation}
    \Heff{2}{T} - \Heff{1}{T}
    =\frac{g^3}{2 \omega^2}\biggl(
            \sigma_-\otimes \of*{ \of*{\ad}^{3}-\textfrac{1}{2} \an^{3} -  \ad - \of*{\ad}^2  \an}
            + \sigma_+  \otimes \of*{ \an^{3} -\textfrac{1}{2} \of*{\ad}^{3} - \an  - \ad  \an^2}
        \biggr) .
\end{equation}
\Cref{thm:effective-Hamiltonian-conjugation} guarantees that there is at least one self-adjoint extension of this operator on \(\cinftyofh\).
On the other hand, the findings in recent work~\cite{fischerSelfadjointRealizationsHigherorder2026} suggest that, being of third order in \(\ad\) and \(\an\), this operator is not necessarily essentially self-adjoint.

As pointed out earlier, our proof of \cref{thm:general-error-ip} also allows to compute explicit error bounds.
A work-in-progress will report on such bounds for the Rabi model elsewhere.

\subsection{Periodically driven quantum harmonic oscillator}\label{sec:periodically-driven-oscillator}
As in the previous example, we consider the creation and annihilation operators $\ad$ and $\an$ and the corresponding number operator $\nn = \ad \an$, and $q=\frac{1}{\sqrt{2}}(\an+\ad)$ and $p = \frac{1}{\iu\sqrt{2}}(\an-\ad)$ the position and momentum operators on $\Ltwo{\R}$, respectively.
Given a continuously differentiable real-valued antisymmetric and 1-periodic function $f$ on $\R$, we define the following Hamiltonian on $\mathcal{H}\coloneqq L^2(\R)$:
\begin{equation}
    H(t)
    \coloneqq \frac{1}{2} \omega \of*{p^2+q^2} + f(t) \sqrt{2} q
    = \omega \of*{\nn + \textfrac{1}{2} }  + f(t) \of*{\an + \ad}
\end{equation}
which describes a periodically driven quantum harmonic oscillator.
We study the high-frequency regime, namely the case where the frequency of $f$ becomes high, while $\omega>0$ is a constant. In this setting we can then compute the associated Floquet--Magnus expansion.

We first check that \assumpC{} are satisfied.
Let $H_0 = \omega(\nn+\frac{1}{2})$; then $\domof{H_0}=\domof{\nn}$ and $\nn$ is self-adjoint on this subspace \cite[Example~X.6.2]{reedIIFourierAnalysis1975}.
Furthermore, $\nn^m f(t)(\an+\ad)$ is relatively bounded with respect to $\nn^{m+\frac{1}{2}}$ as follows from the defining commutation relations for $\an$ and $\ad$.
It then follows from \cite[Lem.~6.2]{teschlMathematicalMethodsQuantum2009} that \(H(t)\) is self-adjoint on $\domof{H_0}$, so \cref{assump:C-H-self-adjoint,assump:C-H-bound} are satisfied.
\Cref{assump:C-H-diff} follows directly from the continuous differentiability of $f$ and the fact that the individual terms are all defined on $\domof{H_0^{m+1}}$.
\Cref{assump:C-H-periodic} is obvious. Iterating the commutator formula
\begin{equation}
    \comm*{\ad \an}{(\an+\ad)}
    = \ad - \an
\end{equation}
shows that $\comm*{H_0^m}{H(t)} = \comm*{H_0^m}{f(t) \of*{\an + \ad} }$ is relatively $H_0^{m-\frac{1}{2}}$-bounded and hence relatively $H_0^m$-bounded, so \cref{assump:C-commutator-bound} is fulfilled.
The conjugation in \cref{assump:C-H-conj} is chosen as follows:
\begin{equation}
    J \colon\: \Ltwo{\R} \to \Ltwo{\R}
    ,\quad
    \of*{J\psi}(s)
    = \overline{\psi(-s)}
    .
\end{equation}
Then
\begin{equation}
    \of*{JH(t) - H(-t)J}\psi(s)
    = \overline{ -s \sqrt{2} f(t) \psi(-s) } - s \sqrt{2} f(-t) \overline{ \psi(-s) },
\end{equation}
which vanishes as $f$ is antisymmetric.
For symmetric $f$ instead one could have picked $J \psi(s) = \overline{ \psi(s) }$.
Both choices leave \(\cinftyofh = \mathcal{S}(\R)\) invariant.

We can now apply our theory and compute $\Heff{L}{T}$. For example, for $L=3$ and $f(t) = \sin(2\pi t)$, we get
\begin{equation}
    \Heff{3}{T}
    = \omega \of*{\nn + \textfrac{1}{2} }
        + T \frac{\omega}{2\pi} \iu \of*{\an - \ad}
        + T^2 \frac{3 \omega}{8\pi^2 }\one
        + T^3 \frac{\omega^3}{8\pi^3} \iu \of*{\an - \ad}
        .
\end{equation}

As pointed out in \cref{rem:error-bounds}, one could now trace back the various relative boundedness constants to obtain and illustrate explicit error bounds for the effective time evolution at order $L$. While the general formulas are very complicated, for specific models such as this one and using numerical tools, this becomes doable. We plan to come back to this point in future work.

\appendix
\section*{Acknowledgements}

We thank Anirban Dey, Paolo Facchi, Karl-Hermann Neeb, and Lauritz van Luijk for useful discussions and suggestions. DL acknowledges financial support by Friedrich-Alexander-Universit\"at Erlangen-N\"urnberg through the funding program “Emerging Talent Initiative” (ETI), and was partially supported by the project TEC-2024/COM-84-QUITEMAD-CM.
We thank the anonymous referees for their helpful suggestions and comments.

\section*{Statements and Declarations}
\subsection*{Data availability}
No datasets were generated or analyzed during this study.
\subsection*{Competing interests}
The authors have no competing interests to declare that are relevant to the content of this article.

\printbibliography[
  heading=bibintoc,
]

@book{Bourbaki_1-3_1998, 
    address={Berlin Heidelberg}, 
    title={Lie groups and Lie algebras. Chapters 1-3}, 
    publisher={Springer}, 
    author={Bourbaki, Nicolas}, 
    year={1998},

}

@article{hochbruckMagnusIntegratorsTimeDependent2003,
  title = {On {{Magnus Integrators}} for {{Time-Dependent Schr\"odinger Equations}}},
  author = {Hochbruck, Marlis and Lubich, Christian},
  date = {2003-01},
  journaltitle = {SIAM Journal on Numerical Analysis},
  shortjournal = {SIAM J. Numer. Anal.},
  volume = {41},
  number = {3},
  pages = {945--963},
  issn = {0036-1429, 1095-7170},
  doi = {10.1137/S0036142902403875},
}

@article{gonzalezSecondorderMagnustypeIntegrator2006,
  title = {A Second-Order {{Magnus}}-Type Integrator for Nonautonomous Parabolic Problems},
  author = {Gonz\'alez, C. and Ostermann, A. and Thalhammer, M.},
  date = {2006-05-01},
  journaltitle = {Journal of Computational and Applied Mathematics},
  shortjournal = {J. Comput. Appl. Math.},
  series = {Proceedings of {{The}} 11th {{International Congress}} on {{Computational}} and {{Applied Mathematics}}},
  volume = {189},
  number = {1},
  pages = {142--156},
  issn = {0377-0427},
  doi = {10.1016/j.cam.2005.04.036},

}

@article{Burgarth_Facchi_Fraas_Hillier_2021,
 title={Non-Markovian noise that cannot be dynamically decoupled by periodic spin echo pulses},
 volume={11}, ISSN={2542-4653},
 DOI={10.21468/SciPostPhys.11.2.027},
 abstractNote={SciPost Journals Publication Detail SciPost Phys. 11, 027 (2021) Non-Markovian noise that cannot be dynamically decoupled by periodic spin echo pulses},
 number={2},
 journal={SciPost Physics},
 author={Burgarth, Daniel and Facchi, Paolo and Fraas, Martin and Hillier, Robin},
 year={2021},
 month=aug,
 pages={027}
}

@article{YiFasterDD26,
  title = {Faster Randomized Dynamical Decoupling},
  author = {Yi, Changhao and Kim, Leeseok and Marvian, Milad},
  journal = {Phys. Rev. Lett.},
  volume = {136},
  issue = {1},
  pages = {010601},
  numpages = {7},
  year = {2026},
  month = {1},
  publisher = {American Physical Society},
  doi = {10.1103/fk7j-y1vl},
  url = {https://link.aps.org/doi/10.1103/fk7j-y1vl}
}

@article{HahnTrotter24,
  title = {Strong error bounds for Trotter and strang-splittings and their implications for quantum chemistry},
  author = {Burgarth, Daniel and Facchi, Paolo and Hahn, Alexander and Johnsson, Mattias and Yuasa, Kazuya},
  journal = {Phys. Rev. Res.},
  volume = {6},
  issue = {4},
  pages = {043155},
  numpages = {42},
  year = {2024},
  month = {11},
  publisher = {American Physical Society},
  doi = {10.1103/PhysRevResearch.6.043155},
  url = {https://link.aps.org/doi/10.1103/PhysRevResearch.6.043155}
}

@article{bukovUniversalHighfrequencyBehavior2015,
  title = {Universal High-Frequency Behavior of Periodically Driven Systems: From Dynamical Stabilization to {{Floquet}} Engineering},
  shorttitle = {Universal High-Frequency Behavior of Periodically Driven Systems},
  author = {Bukov, Marin and D'Alessio, Luca and Polkovnikov, Anatoli},
  date = {2015-03-04},
  journaltitle = {Advances in Physics},
  volume = {64},
  number = {2},
  pages = {139--226},
  publisher = {Taylor \& Francis},
  issn = {0001-8732},
  doi = {10.1080/00018732.2015.1055918}
}

@article{rossattoSpectralClassificationCoupling2017,
  title = {Spectral Classification of Coupling Regimes in the Quantum {{Rabi}} Model},
  author = {Rossatto, Daniel Z. and Villas-B\^oas, Celso J. and Sanz, Mikel and Solano, Enrique},
  date = {2017-07-24},
  journaltitle = {Physical Review A},
  shortjournal = {Phys. Rev. A},
  volume = {96},
  number = {1},
  pages = {013849},
  publisher = {American Physical Society},
  doi = {10.1103/PhysRevA.96.013849}
}

@online{apelSharperMagnusExpansion2025,
  title = {A Sharper {{Magnus}} Expansion Bound Woven in Binary Branches},
  author = {Apel, Harriet and Cubitt, Toby and Onorati, Emilio},
  date = {2025-09-22},
  eprint = {2509.18312},
  eprinttype = {arXiv},
  eprintclass = {quant-ph},
  doi = {10.48550/arXiv.2509.18312}
}

@article{kato1970linear,
  title={Linear evolution equations of “hyperbolic” type},
  author={Kato, Tosio},
  journal={Journal of the Faculty of Science, the University of Tokyo, Sec. 1},
  volume={17},
  pages={241-258},
  year={1970}
}

@inbook{Schaubelt_VI_9_Engel_Nagel_2000, 
     address={New York},
     series={Graduate Texts in Mathematics},
     title={Semigroups for Nonautonomous Cauchy Problems},
     booktitle={One-Parameter Semigroups for Linear Evolution Equations}, 
     volume={194}, 
     ISBN={978-0-387-98463-6}, 
     DOI={10.1007/b97696}, 
     publisher={Springer-Verlag},
     author={Schnaubelt, Roland},
     bookauthor={Engel, Klaus-Jochen and Nagel, Rainer},
     year={2000}, 
     collection={Graduate Texts in Mathematics},
}

@article{Blanes_Casas_Oteo_Ros_2009,
    title={The {Magnus} expansion and some of its applications},
    volume={470},
    ISSN={0370-1573},
    DOI={10.1016/j.physrep.2008.11.001},
    number={5},
    journal={Physics Reports},
    author={Blanes, S. and Casas, F. and Oteo, J. A. and Ros, J.},
    year={2009},
    month=jan,
    pages={151–238}
}

@article{burgarth2022one,
  title={One bound to rule them all: from Adiabatic to Zeno},
  author={Burgarth, Daniel and Facchi, Paolo and Gramegna, Giovanni and Yuasa, Kazuya},
  journal={Quantum},
  volume={6},
  pages={737},
  year={2022},
  publisher={Verein zur F{\"o}rderung des Open Access Publizierens in den Quantenwissenschaften},
  doi={10.22331/q-2022-06-14-737}
}

@article{burgarth2024taming,
  title = {Taming the {{Rotating Wave Approximation}}},
  author = {Burgarth, Daniel and Facchi, Paolo and Hillier, Robin and Ligabò, Marilena},
  date = {2024-02-21},
  journaltitle = {Quantum},
  volume = {8},
  pages = {1262},
  doi = {10.22331/q-2024-02-21-1262},
  url = {https://quantum-journal.org/papers/q-2024-02-21-1262/},
}

@book{morettiSpectralTheoryQuantum2017,
  title = {Spectral Theory and Quantum Mechanics},
  author = {Moretti, Valter},
  date = {2017},
  series = {{{UNITEXT}}},
  volume = {110},
  publisher = {Springer International Publishing},
  location = {Cham},
  doi = {10.1007/978-3-319-70706-8},
  isbn = {978-3-319-70706-8}
}

@mvbook{reedIIFourierAnalysis1975,
  title = {{{II}}: {{Fourier}} Analysis, Self-Adjointness},
  shorttitle = {{{II}}},
  author = {Reed, Michael and Simon, Barry},
  date = {1975},
  series = {Methods of Mathematical Physics},
  volume = {2},
  publisher = {Elsevier},
  isbn = {978-0-12-585002-5},
  pagetotal = {388},
  volumes = {4}
}

@book{hilleFunctionalAnalysisSemigroups1957,
  title = {Functional Analysis and Semi-Groups},
  author = {Hille, Einar and Phillips, Ralph S.},
  date = {1957},
  series = {Colloquium {{Publications}}},
  edition = {Rev. ed.},
  number = {31},
  publisher = {American Mathematical Society},
  location = {Providence},
  isbn = {978-0-8218-3395-7},
  doi={10.1090/coll/031}
}

@book{diestelVectorMeasures1977a,
  title = {Vector {{Measures}}},
  author = {Diestel, J. and Uhl, J.},
  date = {1977-06-01},
  series = {Mathematical {{Surveys}} and {{Monographs}}},
  volume = {15},
  publisher = {American Mathematical Society},
  location = {Providence, Rhode Island},
  doi = {10.1090/surv/015},
  isbn = {978-1-4704-1242-5}
}

@article{richter2026quantifying,
	author={Richter, Leonhard and Burgarth, Daniel and Lonigro, Davide},
	title={Quantifying the rotating-wave approximation of the Dicke model},
	journal={Journal of Physics A: Mathematical and Theoretical},
	url={http://iopscience.iop.org/article/10.1088/1751-8121/ae42a3},
	year={2026},
	volume={59},
    pages={075203},
    doi={10.1088/1751-8121/ae42a3}
}

@book{teschlMathematicalMethodsQuantum2009,
  title = {{Mathematical Methods in Quantum Mechanics: With Applications to Schr\"odinger Operators}},
  shorttitle = {Mathematical Methods in Quantum Mechanics},
  author = {Teschl, Gerald},
  date = {2009},
  series = {Graduate Studies in Mathematics},
  number = {Vol. 99},
  publisher = {American Mathematical Society},
  location = {Providence, R.I},
  isbn = {978-0-8218-4660-5}
}

@book{yosidaFunctionalAnalysis1995,
  title = {Functional {{Analysis}}},
  author = {Yosida, K\^osaku},
  date = {1995},
  series = {Classics in {{Mathematics}}},
  volume = {123},
  publisher = {Springer Berlin Heidelberg},
  location = {Berlin, Heidelberg},
  doi = {10.1007/978-3-642-61859-8},
  isbn = {978-3-642-61859-8}
}

@article{bochnerIntegrationFunktionenDeren1933,
  title = {Integration von Funktionen, deren Werte die Elemente eines Vektorraumes sind},
  author = {Bochner, S.},
  date = {1933},
  journaltitle = {Fundamenta Mathematicae},
  volume = {20},
  pages = {262--276},
  publisher = {Instytut Matematyczny Polskiej Akademii Nauk},
  issn = {0016-2736, 1730-6329},
  doi = {10.4064/fm-20-1-262-176}
}

@article{deyErrorBoundsFloquetMagnus2025,
    title={Error bounds for the {Floquet-Magnus} expansion and their application to the semiclassical quantum {Rabi} model},
    volume={112},
    url={https://link.aps.org/doi/10.1103/6bgj-s987},
    DOI={10.1103/6bgj-s987},
    number={5},
    journal={Physical Review A},
    publisher={American Physical Society},
    author={Dey, Anirban and Lonigro, Davide and Yuasa, Kazuya and Burgarth, Daniel},
    year={2025},
    month=nov,
    pages={053723}
}

@book{pazySemigroupsLinearOperators1983,
  title = {Semigroups of {{Linear Operators}} and {{Applications}} to {{Partial Differential Equations}}},
  author = {Pazy, Amnon},
  date = {1983},
  series = {Applied {{Mathematical Sciences}}},
  number = {44},
  publisher = {Springer New York},
  location = {New York, NY},
  doi = {10.1007/978-1-4612-5561-1},
  isbn = {978-1-4612-5561-1},
  pagetotal = {296}
}

@article{gronwall1919note,
  title={Note on the derivatives with respect to a parameter of the solutions of a system of differential equations},
  author={Gronwall, Thomas Hakon},
  journal={Annals of Mathematics},
  volume={20},
  number={4},
  pages={292--296},
  year={1919},
  publisher={JSTOR},
  doi={10.2307/1967124}
}

@article{kato1953integration,
  title={Integration of the equation of evolution in a Banach space},
  author={Kato, Tosio},
  journal={Journal of the Mathematical Society of Japan},
  volume={5},
  number={2},
  pages={208--234},
  year={1953},
  publisher={The Mathematical Society of Japan}
}

@book{larson2024jaynes,
author = {Larson, Jonas and Mavrogordatos, Themistoklis},
title = {The Jaynes–Cummings Model and its Descendants (Second Edition)},
publisher = {IOP Publishing},
year = {2024},
isbn = {978-0-7503-6452-2},
url = {https://doi.org/10.1088/978-0-7503-6452-2},
doi = {10.1088/978-0-7503-6452-2}
}

@online{gloecknerInfiniteDimensionalLieGroups2026,
  title = {Infinite-{{Dimensional Lie Groups}}},
  author = {Glöckner, Helge and Neeb, Karl-Hermann},
  date = {2026-02-12},
  eprint = {2602.12362},
  eprinttype = {arXiv},
  eprintclass = {math},
  doi = {10.48550/arXiv.2602.12362},
  note={}
}

@article{fischerSelfadjointRealizationsHigherorder2026,
  title = {Self-Adjoint Realizations of Higher-Order Squeezing Operators},
  author = {Fischer, Felix and Burgarth, Daniel and Lonigro, Davide},
  date = {2026-06-26},
  journaltitle = {Journal of Physics A: Mathematical and Theoretical},
  shortjournal = {J. Phys. Math. Theor.},
  volume = {59},
  number = {25},
  pages = {255203},
  issn = {1751-8113, 1751-8121},
  doi = {10.1088/1751-8121/ae7acb},

}

@book{daleckiiStabilitySolutionsDifferential2002,
  title = {Stability of {{Solutions}} of {{Differential Equations}} in {{Banach Space}}},
  author = {family=Dalecki{\u \i}, given=Ju., given-i={{Ju}} and Kre{\u \i}n, M.},
  date = {2002-03-15},
  series = {Translations of {{Mathematical Monographs}}},
  volume = {43},
  publisher = {American Mathematical Society},
  location = {Providence, Rhode Island},
  doi = {10.1090/mmono/043},
  isbn = {978-0-8218-3238-7}
}

@article{cohenReminiscenceClassicalChaos2023,
  title = {Reminiscence of {{Classical Chaos}} in {{Driven Transmons}}},
  author = {Cohen, Joachim and Petrescu, Alexandru and Shillito, Ross and Blais, Alexandre},
  date = {2023-04-21},
  journaltitle = {PRX Quantum},
  shortjournal = {PRX Quantum},
  volume = {4},
  number = {2},
  pages = {020312},
  issn = {2691-3399},
  doi = {10.1103/PRXQuantum.4.020312}
}

@article{abaninEffectiveHamiltoniansPrethermalization2017,
  title = {Effective {{Hamiltonians}}, Prethermalization, and Slow Energy Absorption in Periodically Driven Many-Body Systems},
  author = {Abanin, Dmitry A. and De Roeck, Wojciech and Ho, Wen Wei and Huveneers, Fran\c cois},
  date = {2017-01-27},
  journaltitle = {Physical Review B},
  shortjournal = {Phys. Rev. B},
  volume = {95},
  number = {1},
  pages = {014112},
  publisher = {American Physical Society},
  doi = {10.1103/PhysRevB.95.014112}
}

@article{fangHighorderMagnusExpansion2026,
  title = {High-Order {{Magnus}} Expansion for {{Hamiltonian}} Simulation},
  author = {Fang, Di and Liu, Diyi and Zhu, Shuchen},
  date = {2026-10-07},
  journaltitle = {npj Quantum Information},
  shortjournal = {Npj Quantum Inf.},
  publisher = {Nature Publishing Group},
  issn = {2056-6387},
  doi = {10.1038/s41534-026-01385-x},

}

@article{beaudoinDissipationUltrastrongCoupling2011,
  title = {Dissipation and Ultrastrong Coupling in Circuit {{QED}}},
  author = {Beaudoin, F\'elix and Gambetta, Jay M. and Blais, A.},
  date = {2011-10-19},
  journaltitle = {Physical Review A},
  shortjournal = {Phys. Rev. A},
  volume = {84},
  number = {4},
  pages = {043832},
  publisher = {American Physical Society},
  doi = {10.1103/PhysRevA.84.043832}
}

@article{forn-diazUltrastrongCouplingRegimes2019,
  title = {Ultrastrong Coupling Regimes of Light-Matter Interaction},
  author = {Forn-D\'iaz, P. and Lamata, L. and Rico, E. and Kono, J. and Solano, E.},
  date = {2019-06-07},
  journaltitle = {Reviews of Modern Physics},
  shortjournal = {Rev. Mod. Phys.},
  volume = {91},
  number = {2},
  pages = {025005},
  publisher = {American Physical Society},
  doi = {10.1103/RevModPhys.91.025005}
}

@online{paingConditionalSqueezingInduced2026,
  title = {Conditional Squeezing Induced by a Two-Level System: A First-Principles Approach to {{QND}} Readout},
  shorttitle = {Conditional Squeezing Induced by a Two-Level System},
  author = {Paing, Phoenix M. M. and James, Daniel F. V.},
  date = {2026-01-30},
  eprint = {2508.03506},
  eprinttype = {arXiv},
  eprintclass = {quant-ph},
  doi = {10.48550/arXiv.2508.03506}
}

@article{blanesHighorderCommutatorfreeQuasiMagnus2017,
  title = {High-order commutator-free quasi-{Magnus} exponential integrators for non-autonomous linear evolution equations},
  volume = {220},
  issn = {0010-4655},
  url = {https://www.sciencedirect.com/science/article/pii/S0010465517302357},
  doi = {10.1016/j.cpc.2017.07.016},
  urldate = {2026-10-02},
  journal = {Computer Physics Communications},
  author = {Blanes, Sergio and Casas, Fernando and Thalhammer, Mechthild},
  month = nov,
  year = {2017},
  pages = {243--262},

}

@article{blanesConvergenceAnalysisHighorder2018,
  title = {Convergence analysis of high-order commutator-free quasi-{Magnus} exponential integrators for nonautonomous linear evolution equations of parabolic type},
  volume = {38},
  copyright = {http://academic.oup.com/journals/pages/about\_us/legal/notices},
  issn = {0272-4979, 1464-3642},
  doi = {10.1093/imanum/drx012},
  number = {2},
  journal = {IMA Journal of Numerical Analysis},
  author = {Blanes, Sergio and Casas, Fernando and Thalhammer, Mechthild},
  month = apr,
  year = {2018},
  pages = {743--778},

}

@article{casaresQuantumSimulationTimedependent2024,
  title = {Quantum simulation of time-dependent {Hamiltonians} via commutator-free quasi-{Magnus} operators},
  volume = {8},
  url = {https://quantum-journal.org/papers/q-2024-12-17-1567/},
  doi = {10.22331/q-2024-12-17-1567},
  urldate = {2026-10-02},
  journal = {Quantum},
  publisher = {Verein zur Förderung des Open Access Publizierens in den Quantenwissenschaften},
  author = {Casares, Pablo Antonio Moreno and Zini, Modjtaba Shokrian and Arrazola, Juan Miguel},
  month = dec,
  year = {2024},
  pages = {1567},

}

@article{fangTimeDependentHamiltonianSimulation2025a,
  title = {Time-{{Dependent Hamiltonian Simulation}} via {{Magnus Expansion}}: {{Algorithm}} and {{Superconvergence}}},
  shorttitle = {Time-{{Dependent Hamiltonian Simulation}} via {{Magnus Expansion}}},
  author = {Fang, Di and Liu, Diyi and Sarkar, Rahul},
  date = {2025-05-07},
  journaltitle = {Communications in Mathematical Physics},
  shortjournal = {Commun. Math. Phys.},
  volume = {406},
  number = {6},
  pages = {128},
  issn = {1432-0916},
  doi = {10.1007/s00220-025-05314-5},

}

@article{wangPurePointSpectrum2008,
  title = {Pure {{Point Spectrum}} of the {{Floquet Hamiltonian}} for the {{Quantum Harmonic Oscillator Under Time Quasi-Periodic Perturbations}}},
  author = {{W. -M. Wang}},
  date = {2008-01-01},
  journaltitle = {Communications in Mathematical Physics},
  shortjournal = {Commun. Math. Phys.},
  volume = {277},
  number = {2},
  pages = {459--496},
  issn = {1432-0916},
  doi = {10.1007/s00220-007-0379-z},

}

@article{grebertKAMQuantumHarmonic2011,
  title = {{{KAM}} for the {{Quantum Harmonic Oscillator}}},
  author = {Gr\'ebert, Beno\^it and Thomann, Laurent},
  date = {2011-10-01},
  journaltitle = {Communications in Mathematical Physics},
  shortjournal = {Commun. Math. Phys.},
  volume = {307},
  number = {2},
  pages = {383--427},
  issn = {1432-0916},
  doi = {10.1007/s00220-011-1327-5},

}

@article{howlandStationaryScatteringTheory1974a,
  title = {Stationary Scattering Theory for Time-Dependent {{Hamiltonians}}},
  author = {Howland, James S.},
  date = {1974-12-01},
  journaltitle = {Mathematische Annalen},
  shortjournal = {Math. Ann.},
  volume = {207},
  number = {4},
  pages = {315--335},
  issn = {1432-1807},
  doi = {10.1007/BF01351346},

}
\end{document}